\RequirePackage{snapshot}
\documentclass[twocolumn]{svjour3}

\usepackage{algorithm}
\usepackage[noend]{algpseudocode}

\usepackage{amsmath}
\usepackage{amssymb}

\usepackage{amsthm}
\usepackage{amsopn}
\usepackage{empheq}
\usepackage{mathtools}

\usepackage{color}
\usepackage{graphicx}
\usepackage{listings}
\usepackage{rotating}
\usepackage{subfig}
\usepackage{tikz}
\usepackage{xspace}
\usepackage{colortbl}
\usepackage{graphicx}
\usepackage{balance}
\usepackage{url}

\usepackage{breakurl}
\usepackage[breaklinks]{hyperref}

\usepackage{paralist}
\usepackage[misc]{ifsym}
\usepackage[T1]{fontenc}

\usepackage{stmaryrd}

\usetikzlibrary{decorations,arrows,shapes,shadows}

\usepackage{environ}
\usepackage{etoolbox}

\usepackage[disable]{todonotes}

\newbool{CompileTechReport}
\newcommand{\detailedproof}[2]{\ifbool{CompileTechReport}{
\begin{proof}
#2
\end{proof}
}{\noindent\textit{Proof Sketch:}#1\qed}\smallskip}
\newcommand{\ifnottechreport}[1]{\ifbool{CompileTechReport}{}{#1}}
\newcommand{\iftechreport}[1]{\ifbool{CompileTechReport}{#1}{}}

\booltrue{CompileTechReport}

\newcommand{\BGDel}[1]{\todo[inline,color=red]{\textbf{Boris deleted:} #1}}

\newcommand{\mypara}[1]{\noindent\textbf{#1.}}
\newcommand{\myunderpar}[1]{\smallskip\noindent\underline{{#1}}}
\newcommand{\mypartitle}[1]{\smallskip\noindent\textbf{#1.}}

\DeclareMathOperator{\defas}{:=}
\newcommand{\listconcat}{\,{\tt ::}\,}
\newcommand{\mathtext}[1]{\,\,\text{#1}\,\,}

\newcommand{\dlImp}[0]{\,\ensuremath{\mathtt{{:}-}}\,}
\newcommand{\dlNeg}{\neg\,}
\newcommand{\bodyOf}[1]{body(#1)}
\newcommand{\headOf}[1]{head(#1)}
\newcommand{\varsOf}[1]{vars(#1)}
\newcommand{\attrsOf}[1]{attrs(#1)}
\newcommand{\argsOf}[1]{args(#1)}
\newcommand{\predOf}[1]{pred(#1)}

\newcommand{\matches}{\curlyeqprec}
\newcommand{\rel}[1]{\ensuremath{\mathtt{#1}}}

\newcommand{\cnst}[1]{\textsf{#1}}
\newcommand{\depthP}[1]{d(#1)}
\newcommand{\aDepth}{d}
\newcommand{\isSucc}{\models}
\newcommand{\isFailed}{\not\models}

\usepackage[yyyymmdd,hhmmss]{datetime}

\newcommand{\KPlus}{\ensuremath{+}_{\SomeK}\xspace}
\newcommand{\KTimes}{\ensuremath{\cdot}_{\SomeK}\xspace}
\newcommand{\Kzero}{0_{\SomeK}}
\newcommand{\Kone}{1_{\SomeK}}

\newcommand{\ProvPoly}{\ensuremath{\mathbb{N}[X]}\xspace}
\newcommand{\ProvPolyDual}{\ensuremath{\mathbb{N}[X,\bar{X}]}\xspace}
 
\newcommand{\BoolProvPoly}{\ensuremath{\mathbb{B}[X]}\xspace}
\newcommand{\WhyProv}{\ensuremath{\mathsf{Why}(X)}\xspace}

\newcommand{\WhichProv}{\ensuremath{\mathsf{Which}(X)}\xspace}

\newcommand{\TrioProv}{\ensuremath{{\mathsf{Trio}(X)}}\xspace}
\newcommand{\PosBool}{\ensuremath{\mathsf{PosBool}(X)}\xspace}

\newcommand{\SomeK}{\ensuremath{{\cal K}}\xspace}
\newcommand{\LinksL}{\ensuremath{\mathsf{Links^L}}\xspace}

\newcommand{\Tau}{\mathcal{T}\xspace}

\newcommand{\fTree}{\Tau}

\newcommand{\kInter}{\pi}
\newcommand{\fVal}{\nu}
\newcommand{\fTup}[2]{t_{{#1},{#2}}}
\newcommand{\fG}[2]{g_{{#1},{#2}}}
\newcommand{\fP}[2]{p_{{#1},{#2}}}
\newcommand{\kInterOf}[1]{\kInter \llbracket{#1}\rrbracket_{\fVal}}
\newcommand{\kInterOfVal}[2]{\kInter \llbracket{#1}\rrbracket_{#2}}

\newcommand{\dualOf}[1]{\bar{#1}}
\newcommand{\aForm}{\varphi}
\newcommand{\freeOf}[1]{\textsc{free}(#1)}
\newcommand{\varOrder}{<_{Var}}
\newcommand{\formToQ}{\ensuremath{\textsc{Tl}_{\aForm \to Q}}}
\newcommand{\formQ}[1]{Q_{\aForm_{#1}}}
\newcommand{\domQ}{Dom}
\newcommand{\fDom}{A}
\newcommand{\formOp}{\operatorname{{\bf op}}}
\newcommand{\varVec}{{\bf x}}
\newcommand{\valVec}{{\bf a}}
\newcommand{\nnf}{{\bf nnf}}
\newcommand{\provToFO}{\ensuremath{\textsc{Tr}_{\explainq \to \ProvPolyDual}}}
\newcommand{\Iinter}{I_{\kInter}}

\newcommand{\gprov}{\ensuremath{\mathrm{\Gamma}}\xspace}

\newcommand{\gameLabel}{\ensuremath{\lambda}}

\newcommand{\provGraph}{\ensuremath{{\cal PG}}\xspace}

\newcommand{\wonLabel}{\textcolor{wonColor}{W}}
\newcommand{\lostLabel}{\textcolor{lostColor}{L}}

\newcommand{\nodeLabel}{{\cal L}}
\newcommand{\successLabel}{{\cal S}}
\newcommand{\stringDom}{\mathbb{L}}

\newcommand{\gameToProv}{\ensuremath{\textsc{Tr}_{\gprov \to \explainq}}}
\newcommand{\provToGame}{\ensuremath{\textsc{Tr}_{\explainq \to \gprov}}}
\newcommand{\gameToNX}{\ensuremath{\textsc{Tr}_{\gprov \to \ProvPoly}}}
\newcommand{\provToNX}{\ensuremath{\textsc{Tr}_{\explainq \to \ProvPoly}}}

\newcommand{\RAplus}{\ensuremath{\mathcal{RA}^+}\xspace}

\newcommand{\GPProg}[3]{\mathbb{GP}_{{#1},{#2}}^{#3}}
\newcommand{\whyq}{\textsc{Why}\,}
\newcommand{\whynotq}{\textsc{Whynot}\,}
\newcommand{\explainq}{\textsc{Expl}}

\newcommand{\qType}{typeof}

\newcommand{\provQ}{\textsc{PQ}}
\newcommand{\aProvQ}{\psi}
\newcommand{\qMatch}{\textsc{Match}}
\newcommand{\qPattern}{\textsc{Pattern}}

\newcommand{\unProg}[1]{#1_{Unified}}
\newcommand{\adProg}[1]{#1_{Annot}}
\newcommand{\fireProg}[1]{{#1}_{Fire}}
\newcommand{\fireCProg}[1]{#1_{FC}}
\newcommand{\moveProg}[1]{#1_M}
\newcommand{\fire}[3]{\rel{F}_{\rel{#1},#3}}
\newcommand{\fireC}[4]{\rel{FC_{#1,{#4},{#3}}}}
\newcommand{\boolT}{true}
\newcommand{\boolF}{false}
\newcommand{\adornment}{\sigma}
\newcommand{\greenT}{\textcolor{DarkGreen}{T}}
\newcommand{\redF}{\textcolor{DarkRed}{F}}
\newcommand{\yellowD}{\textcolor{DarkYellow}{U}}
\newcommand{\nodeSk}[3]{f_{\rel{#2}}^{#3}}

\definecolor{DarkGreen}{rgb}{0,0.45,0}
\definecolor{DarkRed}{rgb}{0.8,0,0}
\definecolor{DarkYellow}{rgb}{0.6,0.6,0}
\definecolor{DarkGray}{rgb}{0.2,0.2,0.2}

\newcommand{\adom}[1]{\ensuremath{\mathit{adom}(#1)}}
\newcommand{\domA}{\ensuremath{\mathit{dom}}}

\newcommand{\tupDom}{\ensuremath{\textsc{Tup}}}

\newcommand{\union}{\cup}

\newcommand{\thead}[1]{{\cellcolor{black}{\textcolor{white}{\textbf{#1}}}}}

\usepackage{multirow}

\newcommand{\mathtab}{\ensuremath\thickspace\thickspace\thickspace}

\newcommand{\card}[1]{| {#1} |}
\algrenewcommand\algorithmicindent{0.8em}
\newtheorem{Theorem}{Theorem}
\newtheorem{Definition}{Definition}

\newtheorem{Example}{Example}

\newcommand{\myproofpar}[1]{\smallskip\noindent\underline{{#1}:}\,}

\definecolor{black}{rgb}{0,0,0}
\definecolor{grey}{rgb}{0.8,0.8,0.8}
\definecolor{red}{rgb}{1,0,0}
\definecolor{green}{rgb}{0,1,0}
\definecolor{darkgreen}{rgb}{0,0.5,0}
\definecolor{darkpurple}{rgb}{0.5,0,0.5}
\definecolor{darkdarkpurple}{rgb}{0.3,0,0.3}
\definecolor{blue}{rgb}{0,0,1}
\definecolor{shadegreen}{rgb}{0.95,1,0.95}
\definecolor{shadeblue}{rgb}{0.95,0.95,1}
\definecolor{shadered}{rgb}{1,0.85,0.85}
\definecolor{shadegrey}{rgb}{0.85,0.85,0.85}
\definecolor{oddRowGrey}{rgb}{0.80,0.80,0.80}
\definecolor{evenRowGrey}{rgb}{0.85,0.85,0.85}
\definecolor{wonColor}{rgb}{0.63,0.79,0.95}
\definecolor{lostColor}{rgb}{1.0,0.55,0.35}

 \makeatletter

\newcommand\PrologPredicateStyle{}
\newcommand\PrologVarStyle{}
\newcommand\PrologAnonymVarStyle{}
\newcommand\PrologAtomStyle{}
\newcommand\PrologOtherStyle{}
\newcommand\PrologCommentStyle{}

\newif\ifpredicate@prolog@
\newif\ifwithinparens@prolog@

\lst@SaveOutputDef{`_}\underscore@prolog

\newcount\currentchar@prolog

\newcommand\@testChar@prolog{  \ifnum\lst@mode=\lst@Pmode    \detectTypeAndHighlight@prolog  \else
    \ifwithinparens@prolog@      \detectTypeAndHighlight@prolog    \fi
  \fi
  \global\predicate@prolog@false}

\newcommand\detectTypeAndHighlight@prolog
{  \def\lst@thestyle{\PrologAtomStyle}  \ifpredicate@prolog@    \def\lst@thestyle{\PrologPredicateStyle}  \else
    \expandafter\splitfirstchar@prolog\expandafter{\the\lst@token}    \expandafter\ifx\@testChar@prolog\underscore@prolog      \ifnum\lst@length=1        \let\lst@thestyle\PrologAnonymVarStyle      \else
        \let\lst@thestyle\PrologVarStyle      \fi
    \else
      \currentchar@prolog=65
      \loop
        \expandafter\ifnum\expandafter`\@testChar@prolog=\currentchar@prolog          \let\lst@thestyle\PrologVarStyle          \let\iterate\relax
        \fi
        \advance \currentchar@prolog by 1
        \unless\ifnum\currentchar@prolog>90
      \repeat
    \fi
  \fi
}
\newcommand\splitfirstchar@prolog{}
\def\splitfirstchar@prolog#1{\@splitfirstchar@prolog#1\relax}
\newcommand\@splitfirstchar@prolog{}
\def\@splitfirstchar@prolog#1#2\relax{\def\@testChar@prolog{#1}}

\def\beginlstdelim#1#2{  \def\endlstdelim{\PrologOtherStyle #2\egroup}  {\PrologOtherStyle #1}  \global\predicate@prolog@false  \withinparens@prolog@true  \bgroup\aftergroup\endlstdelim}

\newcommand\lang@prolog{Prolog-pretty}
\expandafter\lst@NormedDef\expandafter\normlang@prolog  \expandafter{\lang@prolog}

\expandafter\expandafter\expandafter\lstdefinelanguage\expandafter{\lang@prolog}
{  language            = Prolog,
  keywords            = {},        showstringspaces    = false,
  alsoletter          = (,
  alsoother           = @$,
  moredelim           = **[is][\beginlstdelim{(}{)}]{(}{)},
  MoreSelectCharTable =
    \lst@DefSaveDef{`(}\opparen@prolog{\global\predicate@prolog@true\opparen@prolog},
}

\newcommand\@ddedToOutput@prolog\relax
\lst@AddToHook{Output}{\@ddedToOutput@prolog}

\lst@AddToHook{PreInit}
{  \ifx\lst@language\normlang@prolog    \let\@ddedToOutput@prolog\@testChar@prolog  \fi
}

\lst@AddToHook{DeInit}{\renewcommand\@ddedToOutput@prolog{}}

\makeatother

\definecolor{PrologPredicate}{RGB}{000,031,255}
\definecolor{PrologVar}      {RGB}{024,021,125}
\definecolor{PrologAnonymVar}{RGB}{000,127,000}
\definecolor{PrologAtom}     {RGB}{186,032,032}
\definecolor{PrologComment}  {RGB}{063,128,127}
\definecolor{PrologOther}    {RGB}{000,000,000}

\renewcommand\PrologPredicateStyle{\color{PrologPredicate}}
\renewcommand\PrologVarStyle{\color{PrologVar}}
\renewcommand\PrologAnonymVarStyle{\color{PrologAnonymVar}}
\renewcommand\PrologAtomStyle{\color{PrologAtom}}
\renewcommand\PrologCommentStyle{\itshape\color{PrologComment}}
\renewcommand\PrologOtherStyle{\color{PrologOther}}

\lstdefinestyle{Prolog-pygsty}
{
  language     = Prolog-pretty,
  upquote      = true,
  stringstyle  = \PrologAtomStyle,
  commentstyle = \PrologCommentStyle,
  literate     =
    {:-}{{\PrologOtherStyle :-}}2
    {,}{{\PrologOtherStyle ,}}1
    {.}{{\PrologOtherStyle .}}1
} 
\title{PUG: A Framework and Practical Implementation for Why \& Why-Not Provenance (Extended version)}

\author{Seokki~Lee        \and
        Bertram~Lud\"ascher \and
	Boris~Glavic }

\institute{\Letter\, Seokki~Lee, Boris~Glavic \at
              10 W 31st Street, Chicago, IL 60616 \\
              \email{slee195@hawk.iit.edu}, \email{bglavic@iit.edu}                      \and
           Bertram~Lud\"ascher \at
	      501 E. Daniel St, Champaign, IL 61820\\
              \email{ludaesch@illinois.edu}
}

\date{Received: date / Accepted: date}

\iftechreport{
\newcommand{\TRtitle}{PUG: A Framework and Practical Implementation for Why \& Why-Not Provenance\\(extended version)}
\newcommand{\TRauthors}{Seokki Lee, Bertram Lud\"ascher, Boris Glavic}
\newcommand{\TRnumber}{IIT/CS-DB-2018-02}
\newcommand{\TRdate}{2018-08}
}

\begin{document}

\graphicspath{ {./}, {tr/} }
\iftechreport{
\twocolumn[{

\vspace{1cm}

\begin{minipage}{0.5\linewidth}
  \colorbox{black}{\includegraphics[width=1\linewidth]{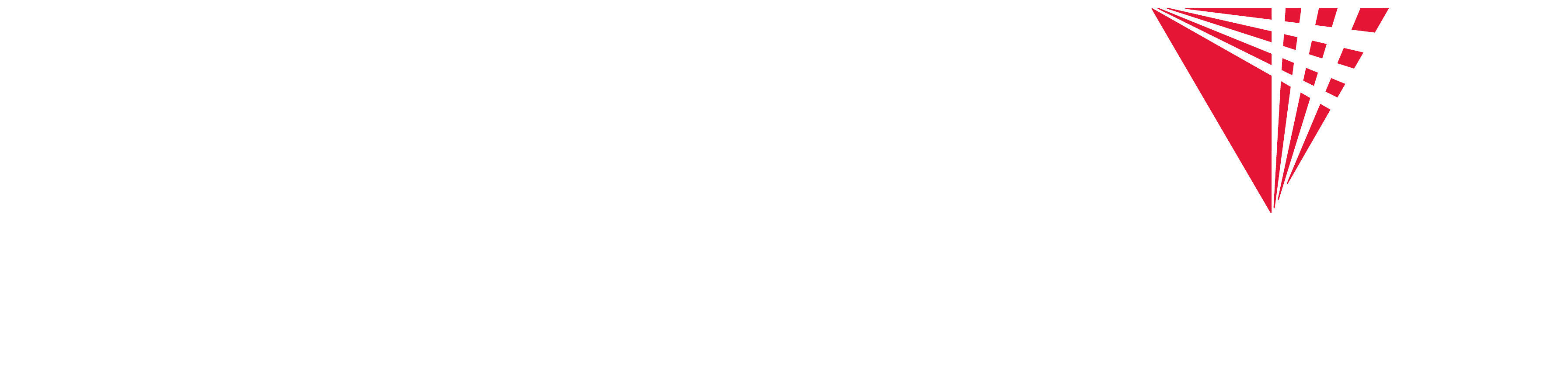}}
\end{minipage}
\hfill
\begin{minipage}{0.16\linewidth}
  \includegraphics[width=1\linewidth]{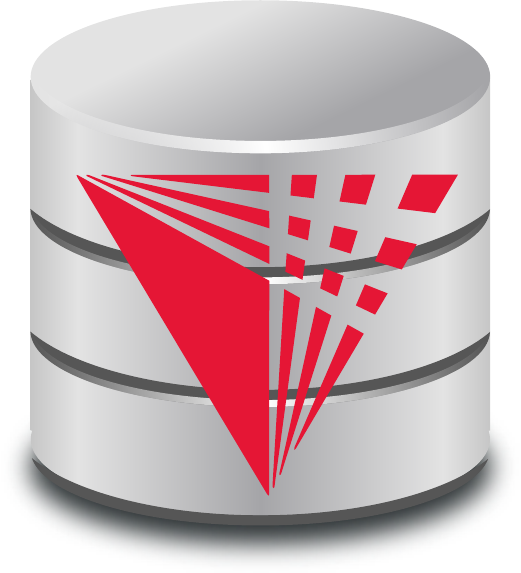}  
\end{minipage}\\
\vspace{4cm}

\centering
\begin{minipage}{1.0\linewidth} 
\centering
{\Huge \bf \TRtitle}
\end{minipage}
\\
\vspace{1cm}

{\huge \TRauthors}\\
\vspace{1cm}

{\huge \tt IIT DB Group Technical Report \TRnumber}\\
\vspace{1cm}

{\Large \TRdate}
\vspace{3cm}

{\huge \url{http://www.cs.iit.edu/~dbgroup/}}

\vspace{3cm}
\begin{minipage}{1.0\linewidth}
\textbf{LIMITED DISTRIBUTION NOTICE}: The research presented in this report may be submitted as a whole or in parts for publication  and will probably be copyrighted if accepted for publication. It has been issued as a Technical Report for early dissemination of its contents. In view of the transfer of copyright to the outside publisher, its distribution outside of IIT-DB prior to publication should be limited to peer communications and specific requests. After outside publication, requests should be filled only by reprints or legally obtained copies of the article (e.g. payment of royalties).  
\end{minipage}

}]

\clearpage
 }

\definecolor{lstpurple}{rgb}{0.5,0,0.5}
\definecolor{lstred}{rgb}{1,0,0}
\definecolor{lstreddark}{rgb}{0.7,0,0}
\definecolor{lstredl}{rgb}{0.64,0.08,0.08}
\definecolor{lstmildblue}{rgb}{0.66,0.72,0.78}
\definecolor{lstblue}{rgb}{0,0,1}
\definecolor{lstmildgreen}{rgb}{0.42,0.53,0.39}
\definecolor{lstgreen}{rgb}{0,0.5,0}
\definecolor{lstorangedark}{rgb}{0.6,0.3,0}	
\definecolor{lstorange}{rgb}{0.75,0.52,0.005}
\definecolor{lstorangelight}{rgb}{0.89,0.81,0.67}
\definecolor{lstbeige}{rgb}{0.90,0.86,0.45}

\lstdefinestyle{links}
{
tabsize=2,
basicstyle=\small\upshape\ttfamily,
language=c,
morekeywords={for,where},
extendedchars=false,
keywordstyle=\bfseries\color{lstpurple},
deletekeywords={==,&&},
keywords=[2]{==,&&,<--},
keywordstyle=[2]\color{lstblue},
stringstyle=\color{lstreddark},
mathescape=true,
escapechar=@,
sensitive=true,
alsoletter=<-&\=
}

\lstdefinestyle{psql}
{
tabsize=2,
basicstyle=\small\upshape\ttfamily,
language=SQL,
morekeywords={PROVENANCE,BASERELATION,INFLUENCE,COPY,ON,TRANSPROV,TRANSSQL,TRANSXML,CONTRIBUTION,COMPLETE,TRANSITIVE,NONTRANSITIVE,EXPLAIN,SQLTEXT,GRAPH,IS,ANNOT,THIS,XSLT,MAPPROV,cxpath,OF,TRANSACTION,SERIALIZABLE,COMMITTED,INSERT,INTO,WITH,SCN,UPDATED},
extendedchars=false,
keywordstyle=\bfseries,
mathescape=true,
escapechar=@,
sensitive=true
}

\lstdefinestyle{psqlcolor}
{
tabsize=2,
basicstyle=\small\upshape\ttfamily,
language=SQL,
morekeywords={PROVENANCE,BASERELATION,INFLUENCE,COPY,ON,TRANSPROV,TRANSSQL,TRANSXML,CONTRIBUTION,COMPLETE,TRANSITIVE,NONTRANSITIVE,EXPLAIN,SQLTEXT,GRAPH,IS,ANNOT,THIS,XSLT,MAPPROV,cxpath,OF,TRANSACTION,SERIALIZABLE,COMMITTED,INSERT,INTO,WITH,SCN,UPDATED},
extendedchars=false,
keywordstyle=\bfseries\color{lstpurple},
deletekeywords={count,min,max,avg,sum},
keywords=[2]{count,min,max,avg,sum},
keywordstyle=[2]\color{lstblue},
stringstyle=\color{lstreddark},
commentstyle=\color{lstgreen},
mathescape=true,
escapechar=@,
sensitive=true
}

\lstdefinestyle{datalog}
{
basicstyle=\footnotesize\upshape\ttfamily,
language=prolog
}

\lstdefinestyle{pseudocode}
{
  tabsize=3,
  basicstyle=\small,
  language=c,
  morekeywords={if,else,foreach,case,return,in,or},
  extendedchars=true,
  mathescape=true,
  literate={:=}{{$\gets$}}1 {<=}{{$\leq$}}1 {!=}{{$\neq$}}1 {append}{{$\listconcat$}}1 {calP}{{$\cal P$}}{2},
  keywordstyle=\color{lstpurple},
  escapechar=&,
  numbers=left,
  numberstyle=\color{lstgreen}\small\bfseries, 
  stepnumber=1, 
  numbersep=5pt,
}

\lstdefinestyle{xmlstyle}
{
  tabsize=3,
  basicstyle=\small,
  language=xml,
  extendedchars=true,
  mathescape=true,
  escapechar=£,
  tagstyle=\color{keywordpurple},
  usekeywordsintag=true,
  morekeywords={alias,name,id},
  keywordstyle=\color{lstred}
}

\maketitle

\begin{abstract}
Explaining why an answer is (or is not) returned by a query is important for many applications including auditing, debugging data and queries, and answering hypothetical questions about data.
In this work, we present the first \emph{practical} approach for answering such questions for queries with negation (first-order queries). Specifically, we introduce a graph-based provenance model that, while syntactic in nature, supports reverse reasoning and is proven to encode a wide range of provenance models from the literature.
The implementation of this model in our PUG (Provenance Unification through Graphs) system takes a provenance question and Datalog query as an input and generates a Datalog program that computes an \emph{explanation}, i.e., the part of the provenance 
that is relevant to answer the question.
Furthermore, we demonstrate how a desirable factorization of provenance  can be achieved by rewriting an input query. We experimentally evaluate our approach demonstrating its efficiency. \end{abstract}

 \section{Introduction}
\label{sec:intro}

Provenance for relational queries records how results of a query
depend on the query's inputs.  This type of information can be used to
explain \emph{why} (and \emph{how}) a result is derived by a query
over a given database.  Recently, provenance-like techniques have been used to explain why a tuple (or a set of tuples described declaratively by a
pattern) is \emph{missing} from the query
result (see~\cite{HD17} for a survey covering both provenance and missing answer techniques).
However, the two problems have been treated mostly in isolation.
Consider the following observation from~\cite{KL13}: asking
why a tuple $t$ is absent from the result of a query $Q$ is equivalent to asking why $t$ is present in $\neg Q$ (i.e., the complement of the result of $Q$ wrt. the active domain). Thus, a unification of \textit{why} and \textit{why-not} provenance is naturally achieved by developing a provenance model for queries with negation.
The approach for provenance and missing answers from~\cite{XZ18} is based on the same observation.
\definecolor{LightCyan}{rgb}{0.88,1,1}
\definecolor{TeaGreen}{rgb}{0.82,0.94,0.75}
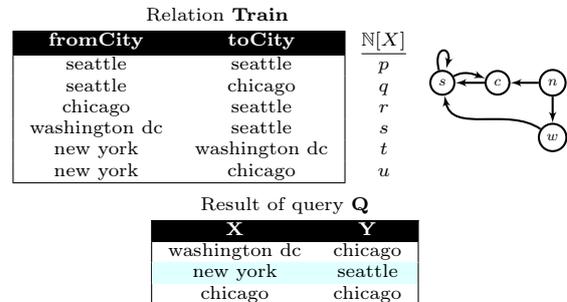
\begin{figure}[t]
  \centering $\,$\\[-3mm]
    \begin{minipage}{1\linewidth}
      \centering
     \begin{minipage}{.69\linewidth}
      \centering
      \begin{align*}
		r_1: \rel{Q}(X,Y) \dlImp \rel{Train}(X,Z), \rel{Train}(Z,Y), \neg \rel{Train}(X,Y)\\
      \end{align*}
     \end{minipage}
    \end{minipage}\\[-3mm]
      \begin{minipage}{1\linewidth}
       \centering
     \begin{minipage}{0.68\linewidth}
      \centering
        \scriptsize \centering $\,$\\[-1mm]
        \begin{tabular}{|cc|c}
          \multicolumn{3}{c}{Relation \textbf{Train}}  \\[0.5mm]\cline{1-2}
	  \thead {fromCity} & \thead {toCity} & \underline{$\ProvPoly$} \\
          seattle & seattle & $p$\\
	  seattle & chicago & $q$\\
	  chicago & seattle & $r$\\
          washington dc & seattle & $s$\\
          new york & washington dc & $t$\\
          new york & chicago & $u$\\
          \cline{1-2}
        \end{tabular}
    \end{minipage}
     \begin{minipage}{0.2\linewidth} \hspace{-3mm}
       \resizebox{1\columnwidth}{!}{\begin{minipage}{1.4\linewidth}
\begin{tikzpicture}[>=latex',
line join=bevel,
line width=0.4mm,
every node/.style={ellipse},
minimum height=4mm]

  \definecolor{fillcolor}{rgb}{0.0,0.0,0.0};
  \node (s) at (0bp,0bp) [draw=black,circle] {$s$};
  \node (c) at (30bp,0bp) [draw=black,circle] {$c$};
  \node (n) at (60bp,0bp) [draw=black,circle] {$n$};
  \node (w) at (60bp,-30bp) [draw=black,circle] {$w$};

   \path[]
   		(n) edge [->] node [right] {} (w)
   		(w) edge [in=280,out=145,->] node [below] {} (s)
   		(n) edge [->] node [above] {} (c)
   		(c) edge [->] node [below] {} (s)
		(s) edge [in=160,out=30,->] node [above] {} (c)
		(s) edge [loop above] node [left] {} (s);

\end{tikzpicture}
\end{minipage}

 }
     \end{minipage}
    \end{minipage}\\$\,$\\[1mm]
     \begin{minipage}{1\linewidth}
      \centering
      \begin{minipage}{0.99\linewidth}
 	\scriptsize \centering
         \begin{tabular}{|cc|}
           \multicolumn{2}{c}{Result of query $\textbf{Q}$}  \\[0.5mm]\cline{1-2}
 	  \thead {X} & \thead {Y}  \\  	   washington dc & chicago \\  	   \rowcolor{LightCyan} new york & seattle\\ 	   chicago & chicago\\            \cline{1-2}
         \end{tabular}
     \end{minipage}
   \end{minipage}
   $\,$\\[-4mm]
  \caption{Example train connection database and query}
  \label{fig:running-example-db}
\end{figure}

In this paper, we introduce a graph model for provenance of
first-order (FO) queries expressed as \emph{non-recursive Datalog queries with negation}\footnote{or, equivalently, queries in full relational algebra (without aggregation), formulas in FO logic under the closed world assumption, and SPJUD-queries (select, project, join, union, difference).}
(or \emph{Datalog} for short)
and an efficient method for explaining a (missing) answer using SQL.
Our approach is based on the observation that typically only a part of the provenance, which we call \textit{explanation} in this work,
is actually relevant for answering the user's provenance question
about the existence or absence of a result.

\begin{figure}[t]
  \centering
  \resizebox{0.73\columnwidth}{!}{\begin{tikzpicture}[>=latex',line join=bevel,line width=0.3mm]
  \definecolor{fillcolor}{rgb}{0.83,1.0,0.8};
  \definecolor{failcolor}{rgb}{0.63,0,0};

  \node (REL_Q_WON_n_s_) at (150bp,175bp) [draw=black,fill=fillcolor,ellipse] {$\boldsymbol{Q(n,s)}$};

  \node (RULE_0_LOS_n_s_w_) at (91bp,150bp) [draw=black,fill=fillcolor,rectangle] {$\boldsymbol{r_1(n,s,w)}$};
  \node (GOAL_0_0_WON_n_w_) at (33bp,125bp) [draw=black,fill=fillcolor,rounded corners=.15cm,inner sep=3pt] {$\boldsymbol{g_{1}^{1}(n,w)}$};
  \node (EDB_T_LOS_n_w_) at (33bp,95bp) [draw=black,fill=fillcolor,ellipse] {$\boldsymbol{T(n,w)}$};

  \node (GOAL_0_1_WON_w_s_) at (93bp,125bp) [draw=black,fill=fillcolor,rounded corners=.15cm,inner sep=3pt] {$\boldsymbol{g_{1}^{2}(w,s)}$};
  \node (EDB_T_LOS_w_s_) at (93bp,95bp) [draw=black,fill=fillcolor,ellipse] {$\boldsymbol{T(w,s)}$};

  \node (GOAL_0_2_WON_n_s_) at (150bp,125bp) [draw=black,fill=fillcolor,rounded corners=.15cm,inner sep=3pt] {$\boldsymbol{g_{1}^{3}(n,s)}$};
  \node (EDB_T_LOS_n_s_) at (150bp,95bp) [draw=black,fill=failcolor,ellipse,text=white] {$\boldsymbol{T(n,s)}$};

  \node (RULE_0_LOS_n_s_c_) at (211bp,150bp) [draw=black,fill=fillcolor,rectangle] {$\boldsymbol{r_1(n,s,c)}$};
  \node (GOAL_0_0_WON_n_c_) at (209bp,125bp) [draw=black,fill=fillcolor,rounded corners=.15cm,inner sep=3pt] {$\boldsymbol{g_{1}^{1}(n,c)}$};
  \node (EDB_T_LOS_n_c_) at (209bp,95bp) [draw=black,fill=fillcolor,ellipse] {$\boldsymbol{T(n,c)}$};

  \node (GOAL_0_1_WON_c_s_) at (265bp,125bp) [draw=black,fill=fillcolor,rounded corners=.15cm,inner sep=3pt] {$\boldsymbol{g_{1}^{2}(c,s)}$};
  \node (EDB_T_LOS_c_s_) at (265bp,95bp) [draw=black,fill=fillcolor,ellipse] {$\boldsymbol{T(c,s)}$};

  \draw [->] (RULE_0_LOS_n_s_c_) -> (GOAL_0_0_WON_n_c_);
  \draw [->] (REL_Q_WON_n_s_) -> (RULE_0_LOS_n_s_c_);
  \draw [->] (RULE_0_LOS_n_s_c_) -> (GOAL_0_1_WON_c_s_);

  \draw [->] (GOAL_0_0_WON_n_c_) -> (EDB_T_LOS_n_c_);

  \draw [->] (REL_Q_WON_n_s_) -> (RULE_0_LOS_n_s_w_);
  \draw [->] (RULE_0_LOS_n_s_w_) -> (GOAL_0_1_WON_w_s_);
  \draw [->] (RULE_0_LOS_n_s_c_) -> (GOAL_0_2_WON_n_s_);
  \draw [->] (RULE_0_LOS_n_s_w_) -> (GOAL_0_2_WON_n_s_);
  \draw [->] (GOAL_0_2_WON_n_s_) -> (EDB_T_LOS_n_s_);

  \draw [->] (GOAL_0_1_WON_c_s_) -> (EDB_T_LOS_c_s_);

  \draw [->] (GOAL_0_0_WON_n_w_) -> (EDB_T_LOS_n_w_);

  \draw [->] (GOAL_0_1_WON_w_s_) -> (EDB_T_LOS_w_s_);

  \draw [->] (RULE_0_LOS_n_s_w_) -> (GOAL_0_0_WON_n_w_);
\end{tikzpicture}
 }
  $\,$\\[-3mm]
  \caption{Provenance graph explaining $\whyq \rel{Q}(n,s)$}
  \label{fig:exam-pg-why-NY-seattle}
\end{figure}
\begin{figure}[t]
  \centering
  $\,$\\[-3mm]
  \resizebox{0.83\columnwidth}{!}{\begin{tikzpicture}[>=latex',line join=bevel,line width=0.3mm]

  \definecolor{fillcolor}{rgb}{0.63,0,0};
  \node (REL_Q_LOS_c_s_) at (285bp,240bp) [draw=black,fill=fillcolor,ellipse,text=white] {$\boldsymbol{Q(s,n)}$};

  \node (RULE_0_WON_c_s_w_) at (190bp,213bp) [draw=black,fill=fillcolor,rectangle,text=white] {$\boldsymbol{r_1(s,n,w)}$};
  \node (GOAL_0_0_LOS_c_w_) at (165bp,183bp) [draw=black,fill=fillcolor,rounded corners=.15cm,inner sep=3pt,text=white] {$g_{1}^{1}(s,w)$};
  \node (REL_TR_LOS_c_w_) at (155bp,152bp) [draw=black,fill=fillcolor,ellipse,text=white] {$\boldsymbol{T(s,w)}$};
  \node (GOAL_0_1_LOS_w_n_) at (215bp,183bp) [draw=black,fill=fillcolor,rounded corners=.15cm,inner sep=3pt,text=white] {$\boldsymbol{g_{1}^{2}(w,n)}$};
  \node (REL_TR_LOS_w_n_) at (215bp,152bp) [draw=black,fill=fillcolor,ellipse,text=white] {$\boldsymbol{T(w,n)}$};

  \node (RULE_0_WON_c_s_n_) at (255bp,213bp) [draw=black,fill=fillcolor,rectangle,text=white] {$\boldsymbol{r_1(s,n,c)}$};

  \node (GOAL_0_0_LOS_c_n_) at (265bp,183bp) [draw=black,fill=fillcolor,rounded corners=.15cm,inner sep=3pt,text=white] {$\boldsymbol{g_{1}^{2}(c,n)}$};
  \node (REL_TR_LOS_c_n_) at (275bp,152bp) [draw=black,fill=fillcolor,ellipse,text=white] {$\boldsymbol{T(c,n)}$};

  \node (RULE_0_WON_c_s_c_) at (320bp,213bp) [draw=black,fill=fillcolor,rectangle,text=white] {$\boldsymbol{r_1(s,n,s)}$};
  \node (GOAL_0_1_LOS_s_n_) at (320bp,183bp) [draw=black,fill=fillcolor,rounded corners=.15cm,inner sep=3pt,text=white] {$\boldsymbol{g_{1}^{2}(s,n)}$};
  \node (REL_TR_LOS_s_n_) at (340bp,152bp) [draw=black,fill=fillcolor,ellipse,text=white] {$\boldsymbol{T(s,n)}$};

  \node (RULE_0_WON_c_s_s_) at (390bp,213bp) [draw=black,fill=fillcolor,rectangle,text=white] {$\boldsymbol{r_1(s,n,n)}$};
  \node (GOAL_0_0_LOS_s_s_) at (370bp,183bp) [draw=black,fill=fillcolor,rounded corners=.15cm,inner sep=3pt,text=white] {$\boldsymbol{g_{1}^{1}(s,n)}$};
  \node (GOAL_0_1_LOS_n_n_) at (415bp,183bp) [draw=black,fill=fillcolor,rounded corners=.15cm,inner sep=3pt,text=white] {$\boldsymbol{g_{1}^{2}(n,n)}$};
  \node (REL_TR_LOS_n_n_) at (415bp,152bp) [draw=black,fill=fillcolor,ellipse,text=white] {$\boldsymbol{T(n,n)}$};

  \draw [->] (RULE_0_WON_c_s_s_) -> (GOAL_0_0_LOS_s_s_);
  \draw [->] (RULE_0_WON_c_s_s_) -> (GOAL_0_1_LOS_n_n_);
  \draw [->] (GOAL_0_1_LOS_n_n_) -> (REL_TR_LOS_n_n_);
  \draw [->] (REL_Q_LOS_c_s_) -> (RULE_0_WON_c_s_w_);
  \draw [->] (GOAL_0_0_LOS_c_w_) -> (REL_TR_LOS_c_w_);
  \draw [->] (REL_Q_LOS_c_s_) -> (RULE_0_WON_c_s_s_);
  \draw [->] (RULE_0_WON_c_s_n_) -> (GOAL_0_0_LOS_c_n_);
  \draw [->] (RULE_0_WON_c_s_c_) -> (GOAL_0_1_LOS_s_n_);
  \draw [->] (GOAL_0_1_LOS_s_n_) -> (REL_TR_LOS_s_n_);
  \draw [->] (GOAL_0_0_LOS_s_s_) -> (REL_TR_LOS_s_n_);
  \draw [->] (REL_Q_LOS_c_s_) -> (RULE_0_WON_c_s_n_);
  \draw [->] (RULE_0_WON_c_s_w_) -> (GOAL_0_0_LOS_c_w_);
  \draw [->] (RULE_0_WON_c_s_w_) -> (GOAL_0_1_LOS_w_n_);
  \draw [->] (GOAL_0_1_LOS_w_n_) -> (REL_TR_LOS_w_n_);

  \draw [->] (REL_Q_LOS_c_s_) -> (RULE_0_WON_c_s_c_);
  \draw [->] (GOAL_0_0_LOS_c_n_) -> (REL_TR_LOS_c_n_);

\end{tikzpicture}
 }
  $\,$\\[-3mm]
  \caption{Provenance graph explaning $\whynotq \rel{Q}(s,n)$}
  \label{fig:exam-pg-whynot-chicago-seattle}
\end{figure}

\begin{Example}
\label{ex:example1}
Consider the relation $\rel{Train}$ in Fig.\,\ref{fig:running-example-db} that stores train connections.
Datalog rule $r_1$ in Fig.\,\ref{fig:running-example-db} computes which cities can be reached
with exactly one transfer, but not directly.
We use the following abbreviations in provenance graphs: T = Train;
n = New York; s = Seattle; w = Washington DC and c = Chicago.
Given the result of this query,  the user may be interested to know why he/she is
able to reach Seattle from New York ($\whyq \rel{Q}(n,s)$) with one
intermediate stop but not directly
or why it is not possible to reach New York from Seattle  in the same fashion ($\whynotq \rel{Q}(s,n)$).
\end{Example}
An explanation for either type of question should justify the existence (absence) of a result as the success (failure) to derive the result
through the rules of the query.
Furthermore, it should explain how the existence (absence) of tuples in the database caused the derivation to succeed (fail).
Provenance graphs providing this type of justification for $\whyq \rel{Q}(n,s)$ and $\whynotq \rel{Q}(s,n)$ are shown in
Fig.\,\ref{fig:exam-pg-why-NY-seattle} and
Fig.\,\ref{fig:exam-pg-whynot-chicago-seattle}, respectively.
These graphs contain three types of nodes: \emph{rule nodes} (boxes labeled with a rule identifier and the constant arguments of a rule derivation), \emph{goal nodes} (rounded boxes labeled with a rule identifier and the goal's position in the rule's body), and \emph{tuple nodes} (ovals). In these provenance graphs, nodes are either colored in
\emph{light green} (successful/existing) or \emph{dark red} (failed/missing).

\begin{Example}\label{ex:example-why}
Consider the explanation (provenance graph in Fig.\,\ref{fig:exam-pg-why-NY-seattle}) for question $\whyq \rel{Q}(n,s)$.  Seattle can be reached from New York by stopping in
Washington DC or Chicago and there is no direct connection between these two cities.
These options correspond to two successful
derivations using rule $r_1$ with $X {=} n$, $Y {=} s$, and $Z {=} w$ (or $Z{=}c$, respectively).
In the provenance graph, there are two
\emph{rule nodes} representing these derivations of $\rel{Q}(n,s)$ based on rule $r_1$.
A derivation is successful if all goals in the body evaluate to true, i.e., a successful \emph{rule node} is connected to successful \emph{goal nodes}
(e.g., $r_1$ is connected to $g_1^1$, the $1^{st}$ goal in the rule's body).
A positive (negated) goal is successful
if the corresponding tuple is (is not) in the database.
Thus, a successful goal node
is connected to the node corresponding to the existing (green) or missing (red) tuple justifying the goal, respectively.
\end{Example}

Supporting negation and missing answers is challenging, because we need to enumerate all potential ways of deriving a missing answer (or intermediate result corresponding to a negated subgoal) and explain why each of these derivations has failed. For that, we have to decide how to bound the set of missing answers to be considered.
Using the closed world assumption, only values that exist in the database or are postulated by the query are used to construct missing tuples. As is customary, we refer to this set of values as the active domain $\adom{I}$ of a database instance $I$.
\BGDel{We will revisit the assumption that all derivations with constants from $\adom{I}$ are meaningful
later on.}

\begin{Example}\label{ex:example-whynot}
Fig.\,\ref{fig:exam-pg-whynot-chicago-seattle} shows the explanation for $\whynotq$ $\rel{Q}(s,n)$,
i.e., why it is not true that New York is reachable from Seattle with exactly one transfer, but not directly.
The tuple $\rel{Q}(s,n)$ is missing from the result
because all potential ways of deriving this tuple through rule $r_1$ have failed.
In this example, $\adom{I} {=} \{c,n,\\ s,w\}$ and, thus, there exist four failed derivations of $\rel{Q}(s,n)$ choosing either of these cities as the intermediate stop between Seattle and New York.
A rule derivation fails if at least one goal in the body evaluates to false.
Failed positive goals in the body of a failed rule are explained by missing
tuples (red \emph{tuple nodes}). For instance, we cannot reach New York from Seattle with an intermediate stop in
Washington DC (the first failed
rule derivation from the left in Fig.\,\ref{fig:exam-pg-whynot-chicago-seattle})
because there exists no connection from Seattle to Washington DC
(\emph{tuple node} $\rel{T}(s,w)$ in red),
and Washington DC to New York
(\emph{tuple node} $\rel{T}(w,n)$ in red). The successful goal $\dlNeg \rel{T}(s,n)$ (there is no direct connection from Seattle to New York) does not contribute to the failure of this derivation and, thus, is not part of the explanation.

\end{Example}

Observe that nodes for missing tuples and successful rule derivations are conjunctive in nature (they depend on all their children) while existing tuples and failed rule derivations are disjunctive (they only require at least one of their children to be present).

\mypartitle{Provenance Model}
\label{sec:intro-model}
  By recording which rule derivations justify the existence or absence of a query result, our model is suited well for debugging both data and queries.
However, simpler provenance types, e.g., only tracking data dependencies, are sufficient for some applications. For example, assume that we record for each train connection from which webpage we retrieved information about this train connection. A user may be interested in knowing based on which webpages a query answer was computed. This question can be answered using a simpler provenance type called Lineage (semiring \WhichProv~\cite{GT17}) which records the set of input tuples a result depends on.
For such applications, we prefer simpler provenance types, because they are easier to interpret and more efficient to compute. Importantly, only minor modifications to our framework were required to support such  provenance types.

 \mypartitle{Relationship to Other Provenance Models}
 In comparison to other provenance models, our model is more syntax-driven. We argue that this is actually a feature (not a bug).
An important question is what is the semantic justification of our model, i.e., how do
we know whether it correctly models Datalog query evaluation semantics and whether all (and only) relevant provenance is captured. First, we observe that our model encodes Datalog query semantics by construction. We justify that all relevant provenance is captured indirectly by demonstrating that our model captures sufficient information to derive provenance according to well-established models. Specifically, we demonstrate that our model is equivalent to provenance games~\cite{KL13} which also support FO queries. It was proven in~\cite{KL13} that provenance polynomials, the most general form of provenance in the semiring model~\cite{GT17,grigoris-tj-simgodrec-2012}, for a result of a positive query can be ``read out'' from a provenance game. By being equivalent to provenance games, our provenance model also enjoys this property. We extend this result to queries with negation by relating our model to semiring provenance for FO  model checking~\cite{GV17,T17,XZ18}. We prove that, for any FO formula $\aForm$, we can generate a query $\formQ{}$ such that the semiring provenance annotation of the formula $\kInter(\aForm)$ according to a  $\SomeK$-interpretation $\kInter$ (annotation of positive and negated literals~\cite{GV17}) can be extracted efficiently from our provenance graph for $\formQ{}$. Note that non-recursive Datalog queries with negation and FO formulas  under the closed world assumption have the same expressive power and, thus, we use these languages interchangeably.

\mypartitle{Reverse Reasoning (How-to Queries)}
 For some applications, a user may not be interested in how a result was derived, but wants to understand how a result of interest can be achieved through updates to the database (see e.g.,~\cite{MG11c,T17,meliou2012tiresias}). We extend our approach to support such ``reverse reasoning'' by introducing a third possible state of nodes in a provenance graph reserved for facts and derivations whose truth is undetermined. The provenance graph generated over an instance with undetermined facts represents a set of provenance graphs - one for each instance that is derived by assigning a truth value to each undetermined fact.
We demonstrate that these graphs can be used to compute the semiring provenance of an FO formula under a provenance tracking interpretation as defined in~\cite{T17}.

\mypartitle{Computing Explanations}
\label{sec:intro-expl}
We utilize Datalog to generate provenance graphs that explain a (missing) query result.
Specifically,
we instrument the input
Datalog program to compute provenance bottom-up. Evaluated over an instance $I$, the instrumented program
returns the edge relation of an explanation (provenance graph).

The main driver of our approach is a rewriting of Datalog rules (so-called \textit{firing rules}) that captures successful and failed rule derivations.
Firing rules for positive queries were first introduced in~\cite{kohler2012declarative}. We have generalized this concept to negation and failed rule derivations. Firing rules provide sufficient information for constructing any of the provenance graph types we support.
 To make provenance capture efficient, we avoid capturing derivations that will not contribute to an explanation.
We achieve this by propagating information from a user's provenance question  throughout the query
to prune derivations that
1) do not agree with the constants of the question or 2) cannot be part of the explanation based on their success/failure status. For instance, in the running example, $\rel{Q}(n,s)$ is only
connected to derivations of rule $r_1$ with $X {=} n$ and $Y {=} s$.

We implemented our approach in \textit{PUG}~\cite{LS17} (Provenance Unification through Graphs), an extension of our \textit{GProM}~\cite{AG14} system.
      Using PUG,       we compile rewritten Datalog programs into relational algebra,       and translate such algebra expressions       into SQL code that
      can be executed by a standard relational database backend.

\mypartitle{Factorizing Provenance}
\label{sec:intro-factorization}
Nodes in our provenance graphs are uniquely identified by their label.
Thus, common subexpressions are shared leading to more compact graphs. For instance, observe that $g_1^3(n,s)$ in Fig.\,\ref{fig:exam-pg-why-NY-seattle} is shared by two rule nodes. We exploit this fact by rewriting the input program to generate more concise, but equivalent, provenance.
This is akin to factorization of provenance polynomials in the semiring model and utilizes factorized databases techniques~\cite{OZ12,OZ15}.

\mypartitle{Contributions}
\label{sec:intro-contributions}
This paper extends our previous work\\~\cite{LS17} in the following ways: we extend our model to support less informative, but more concise, provenance types; we extend our provenance model to support reverse reasoning~\cite{GV17} where the truth of some facts in the database is left undetermined; we demonstrate that our provenance graphs (explanations) are equivalent to provenance games~\cite{KL13} and how semiring provenance and its FO extension as presented in~\cite{T17} can be extracted from our provenance model; we demonstrate how to rewrite an input program to generate a desirable (concise) factorization of provenance and evaluate the performance impact of this technique; finally, we present an experimental comparison with the language-integrated provenance techniques implemented in Links~\cite{FS17}.

The remainder of this paper is organized as follows.
We discuss related work in Sec.\,\ref{sec:rel-work} and review Datalog in Sec.\,~\ref{sec:datalog}.
We define our model in Sec.\,\ref{sec:probl-defin-backgr} and prove it to be equivalent to provenance games~\cite{KL13}  in Sec.~\ref{sec:prov-game}.
We, then, show how our approach relates to semiring provenance for positive queries and FO model checking in Sec.~\ref{sec:semi-annot-model} and~\ref{sec:semi-foq}, respectively.
We
present our approach for computing explanations  in Sec.\,\ref{sec:compute-gp}, and factorization in Sec.\,\ref{sec:factorize}.
 Sec.\,\ref{sec:transl-into-relat} covers our implementation in \textit{PUG}
which we evaluate in
Sec.\,\ref{sec:experiments}. We  conclude in Sec.\,\ref{sec:concl}.

 \section{Related Work}
\label{sec:rel-work}

Our provenance graphs have strong connections to other provenance models
for relational queries and to approaches for explaining missing answers. 

\mypartitle{Provenance Games}
Provenance games~\cite{KL13}
model the evaluation of a given query (input program) $P$ over an instance $I$ as a 2-player game in a way that
resembles SLD(NF) resolution.
\BGDel{If the
position (a node in the game graph) corresponding to a tuple $t$ is won
(the player starting in this position has a winning strategy), then $t \in P(I)$ and if the position is lost, then $t \not\in P(I)$.}
By virtue of supporting negation, provenance games can uniformly answer why and why-not questions.  We prove our approach to be equivalent to provenance games in Sec.\,\ref{sec:prov-game}.
K\"ohler et al.~\cite{KL13} present an algorithm that computes the provenance game for a  program
$P$ and database $I$. However, this approach requires instantiation of the full game graph (which enumerates all existing and missing tuples) and evaluation of a recursive Datalog$^\neg$ program over this graph using the well-founded semantics~\cite{FK97}.  In constrast, our approach directly computes succinct explanations that contain only relevant provenance.

\mypartitle{Database Provenance}
Several provenance models for database queries have been introduced in related work, e.g., see~\cite{cheney2009provenance,grigoris-tj-simgodrec-2012}.
The semiring annotation framework generalizes these models for positive relational algebra (and, thus, positive non-recursive Datalog). An essential property of the \SomeK-relational model is that the semiring of \emph{provenance polynomials} $\ProvPoly$   generalizes all other semirings.
It has been shown in~\cite{KL13} that provenance games generalize $\ProvPoly$ for positive queries. Since our graphs are equivalent to provenance games in the sense that there exist lossless transformations between both models
(see Sec.\,\ref{sec:prov-game}), our graphs also encode $\ProvPoly$ and, thus, all other provenance models expressible as semirings (see Sec.\,\ref{sec:expl-types}).
Provenance graphs which are similar to our graphs
restricted to positive queries have been used as graph representations of semiring provenance
(e.g., see~\cite{DG15c,DM14c,grigoris-tj-simgodrec-2012}).
Both our graphs and the boolean circuits representation of semiring provenance~\cite{DM14c} explicitly share common subexpressions in the provenance.
While these circuits support recursive queries, they do not support negation. Recently, extension of circuits for semirings with monus (supporting set difference) have been discussed~\cite{S18}.
\BGDel{Exploring the relationship of provenance graphs for queries with negation and m-semirings (semirings with support for set difference) is an interesting avenue for future work.}
The semiring model has also been  applied to record provenance of model checking for first-order (FO) logic formulas~\cite{T17,GV17,XZ18}. This work also supports missing answers using the observation made earlier in~\cite{KL13}.
Support for negation relies on 1) translating formulas into negation normal form (\textit{nnf}), i.e., pushing all negations down to literals, and 2) annotating both positive and negative literals using a separate set $X$ and $\bar{X}$ of indeterminates in provenance expressions where variables from $X$ are reserved for positive literals and variables from $\bar{X}$ for negated literals. This idea of using dual (positive and negative) indeterminates is an independent rediscovery of the approach from~\cite{DA13} which applied this idea for FO queries. The main differences between these approaches are 1) that the results from~\cite{DA13} where only shown for  one particular semiring ($Bool(X \cup \bar{X})$, the semiring of boolean expressions over dual indeterminates) and 2) that~\cite{DA13} supports recursion in the form of well-founded Datalog and answer set programming (disjunctive Datalog). We prove that our model encompasses the model from~\cite{GV17}.
The notion of causality is also closely related to provenance.
Meliou et al.~\cite{MG10} compute causes for answers and non-answers. However, the approach requires the user to specify which missing inputs are considered as causes for a missing output.
Roy et al.~\cite{RO15,RS14} employ causality to compute explanations for high or low outcomes of aggregation queries as sets of input tuples which have a large impact on the result. Such sets of tuples are represented compactly through selection queries. A similar method was developed independently by Wu et al.~\cite{WM13}.

\mypartitle{Why-not and Missing Answers}
Approaches for explaining missing answers
are either based on the query~\cite{BH14a,BH14,CJ09,TC10} (i.e., which operators filter out tuples that would have contributed to the missing answer)
or based on the instance~\cite{HH10,huang2008provenance} (i.e., what tuples need to be inserted into the database to turn the missing answer into an answer). The missing answer problem was first stated for query-based explanations in the seminal paper by Chapman et al.~\cite{CJ09}. Huang et al.~\cite{huang2008provenance} first introduced an instance-based approach. Since then, several techniques have been developed to exclude spurious explanations, to support larger classes of queries~\cite{HH10},   and to support distributed Datalog systems in Y!~\cite{WZ14a}.
The approaches for instance-based explanations (with the exception of Y!) have in common that they treat the missing answer problem as a view update problem: the missing answer is a tuple that should be inserted into a view corresponding to the query and this insertion has to be translated as an insertion into the database instance. An explanation is then one particular solution to this view update problem. In contrast to these previous works, our provenance graphs explain missing answers
by enumerating all failed rule derivations that justify why  the answer
is not in the result.
Thus, they are arguably a better fit for use cases such as debugging queries, where in addition
 to determining which missing inputs justify a missing answer, the user also needs to understand why derivations have failed. Furthermore, we do support queries with negation.
Importantly, solutions for view update missing answer problems
can be extracted from our provenance graphs. Thus, in a sense, provenance graphs with our approach generalize some of the previous approaches
(for the class of queries supported, e.g., we  do not support aggregation yet).
Interestingly, recent work has shown that it may be possible to generate
more concise summaries of provenance games~\cite{GK15,RK14} and provenance graphs~\cite{LN17} that are particularly useful for negation and missing answers
to deal with the potentially large size of the resulting provenance. Similarly, some missing answer approaches~\cite{HH10} use c-tables to compactly represent sets of missing answers.
These approaches are complementary to our work.

\mypartitle{Computing Provenance Declaratively}
The concept of rewriting a Datalog program using firing rules to capture provenance as variable bindings of derivations was introduced by K\"ohler et al.~\cite{kohler2012declarative}. They apply this idea for provenance-based debugging of positive Datalog. Firing rules are also similar to relational implementations of provenance capture in Perm~\cite{GM13}, LogicBlox~\cite{GA12}, Orchestra~\cite{GK07a}, and GProM~\cite{AG14}. Zhou et al.~\cite{ZS10} leverage such rules for the distributed ExSPAN system using either full propagation or reference based provenance.
The extension of firing rules for negation is the main enabler of our approach.

\section{Datalog}
\label{sec:datalog}

A Datalog program $P$ consists of a finite set of rules $r_i: \rel{R}(\vec{X}) \dlImp \rel{R_1}(\vec{X_1}), \ldots,$ $\rel{R_n}(\vec{X_n})$ where $\vec{X_j}$ denotes a tuple of variables and/or constants. 
We assume that the rules of a program are labeled $r_1$ to $r_m$. $\rel{R}(\vec{X})$ is the \emph{head} of
the rule, denoted as $\headOf{r_i}$, and $\rel{R_1}(\vec{X_1}), \ldots,
\rel{R_n}(\vec{X_n})$ is the \emph{body} (each $\rel{R_j}(\vec{X_j})$
is a \emph{goal}). 
We use $\varsOf{r_i}$ to denote the set of variables in $r_i$. In this paper, consider non-recursive Datalog with negation (FO queries), so
 goals $\rel{R_j}(\vec{X_j})$ in the body are \emph{literals}, i.e.,
 atoms $\rel{L}(\vec{X_j})$ or their negation $\neg
 \rel{L}(\vec{X_j})$, and recursion is not allowed. 
All rules $r$ of a program have to be \emph{safe}, i.e., every
variable in $r$ must occur positively in $r$'s body (thus, head
variables and variables in negated goals must also occur in a positive goal).
For example, Fig.\,\ref{fig:running-example-db} shows a Datalog query with a single rule $r_1$. 
Here, 
$\headOf{r_1}$ is $\rel{Q}(X,Y)$ and $\varsOf{r_1}$ is $\{X,Y,Z\}$. The rule is safe since the head
variables  ($\{X,Y\}$) and the variables in the negated goal  ($\{X,Y\}$) also occur positively
in the body.
The set of relations in the schema 
 over which $P$ is defined is referred to as the extensional database
 (EDB), while relations defined through rules in $P$ form the
 intensional database (IDB), i.e., the IDB relations are those defined in the head of rules.
We require that $P$ has a distinguished IDB relation $Q$, called the \emph{answer} relation. Given $P$ and instance $I$, we use $P(I)$ to denote the result of $P$ evaluated over $I$. Note that $P(I)$ includes the instance $I$, i.e., all EDB atoms that are true in $I$. 
For an EDB or IDB predicate $R$, we use $R(I)$ to denote the instance of $R$ computed by $P$ and $R(t) \in P(I)$ to denote that $t \in R(I)$ according to $P$.

We use $\adom{I}$ to denote the active domain of instance $I$, i.e., the set of all constants that occur in $I$. Similarly, we use $\adom{\rel{R.A}}$ to denote the active domain of attribute $A$ of relation $\rel{R}$.
In the following, we make use of the concept of a rule derivation. A \emph{derivation} of a rule $r$ is an assignment of variables in $r$ to constants from $\adom{I}$.
For a rule with $n$ variables, we use $r(c_1, \ldots, c_n)$ to denote the 
derivation that is the result of binding $X_i {=} c_i$. We call a derivation \textit{successful} wrt.\ an instance $I$ if each atom in the body of the rule is true in $I$ and \textit{failed} otherwise. 
\BGDel{A \emph{conjunctive query} (CQ) is a Datalog program that consists of a single rule with only EDB atoms in the body. A \textit{union of conjunctive queries} (UCQ) consists of a set of CQ rules with the same head predicate.}

\section{Provenance Model}
\label{sec:probl-defin-backgr}

We now introduce our provenance model and formalize the problem addressed in this work: compute the subgraph of a provenance graph for a given query (input program) $P$ and instance $I$ that explains existence/absence of a tuple in/from the result of $P$.

\subsection{Negation and Domains}
\label{sec:domains}

To be able to explain why a tuple is missing, 
we have to enumerate all failed derivations of this tuple and, for each such derivation, explain why it failed. As mentioned in Sec.\,\ref{sec:intro}, we have to decide how to bound the set of missing answers. We propose a simple, yet general, solution by assuming that each attribute of an IDB or EDB relation has an associated domain. 

\begin{Definition}[Domain Assignment]
\label{def:definition1}
  Let $S = \{\rel{R_1}, \ldots, \rel{R_n}\}$ be a database schema where each $\rel{R_i(A_{1}, \ldots, A_{m})}$ is a relation. Given an instance $I$ of $S$, a \emph{domain assignment} $\domA$ is a function that associates with each attribute $\rel{R.A}$ a domain of values.
We require $\domA(\rel{R.A}) \supseteq \adom{\rel{R.A}}$. \end{Definition}

In our approach, the user specifies each $\domA(\rel{R.A})$ as a query 
$\domA_{\rel{R.A}}$ that returns the set of admissible values for the domain of attribute $\rel{R.A}$.
These associated domains fulfill two purposes: 1) to reduce the size of explanations 
and 2) to avoid semantically meaningless answers.
For instance, if there exists another attribute $\rel{Price}$ in the relation $\rel{Train}$ in Fig.\,\ref{fig:running-example-db}, 
then $\adom{I}$ would also include all the values that appear in this attribute. 
Thus, some failed rule derivations for $r_1$ would assign prices as intermediate stops.
Different attributes may represent the same type of entity (e.g., $\rel{fromCity}$ and $\rel{toCity}$ in our example) and, thus, it would make sense to use their combined domain values when constructing missing answers. For now, we leave it up to the user to specify attribute domains. \BGDel{Using techniques for discovering semantic relationships among attributes to automatically determine feasible attribute domains is an interesting avenue for future work.}

When defining provenance graphs in the following, we are only interested in rule derivations that use constants from the associated domains of attributes accessed by the rule. Given a rule $r$ and variable $X$ used in this rule, let $\attrsOf{r,X}$ denote the set of attributes that variable $X$ is bound to in the body of the rule. In Fig.\,\ref{fig:running-example-db}, $\attrsOf{r_1,Z} {=} \{\rel{Train.fromCity}, \rel{Train.toCity}\}$. 
We say a rule derivation $r(c_1, \ldots, c_n)$ is \emph{domain grounded} iff $c_i \in \bigcap_{A \in \attrsOf{r,X_i}} \domA(A)$ for all $i \in \{1, \ldots, n\}$. 
For a relation $\rel{R(A_1, \ldots, A_n)}$, we use $\tupDom(\rel{R})$ to denote the set of all possible tuples for $\rel{R}$, i.e., $\tupDom(\rel{R}) = \domA(\rel{R.A_1}) \times \ldots \times \domA(\rel{R.A_n})$. \subsection{Provenance Graphs}
\label{sec:provenance-graphs}

Provenance graphs justify the existence (or absence) of a query result based on the success (or failure) to derive it using a query's rules. They also explain how the existence or absence of tuples in the database caused derivations to succeed or fail, respectively. Here, we present a constructive definition of provenance graphs that provide this type of justification. Nodes in these graphs  carry two types of labels: 1) a label that determines the node type (tuple, rule, or goal) and additional information, e.g., the arguments and rule identifier of a derivation; 2) the success/failure status of nodes. Note that the first type of labels uniquely identifies nodes.

\begin{Definition}[Provenance Graph]\label{def:prov=graph}
  Let $P$ be a Datalog program,
$I$ a database instance, $\domA$ a domain assignment for $I$, and $\stringDom$ the domain of strings. The \emph{provenance graph} $\provGraph(P,I)$ is a graph $(V,E,\nodeLabel,\successLabel)$ with nodes $V$, edges $E$, and node labelling functions $\nodeLabel: V \to \stringDom$ and  $\successLabel: V \to \{\greenT, \redF\}$ ($\greenT$ for true/success and $\redF$ for false/failure).
We require that $\forall v,v' \in V: \nodeLabel(v) = \nodeLabel(v') \rightarrow v = v'$.
The graph $\provGraph(P,I)$ is defined as follows: \begin{compactitem}  \item \textbf{Tuple nodes:} For each n-ary EDB or IDB predicate $R$ and tuple $(c_1, \ldots,c_n)$ of constants from the associated domains ($c_i \in \domA(\rel{R.A_i})$), there exists a node $v$ labeled $R(c_1, \ldots,c_n)$. $\successLabel(v) = \greenT$ iff $R(c_1, \ldots,c_n) \\\in P(I)$      and $\successLabel(v) = \redF$ otherwise.
  \item \textbf{Rule nodes:} For every successful domain grounded derivation $r_i(c_1, \ldots, c_n)$, there exists a node $v$ in $V$ labeled $r_i(c_1, \ldots,c_n)$ with $\successLabel(v) = \greenT$. For every failed domain grounded derivation $r_i(c_1, \ldots, c_n)$ where $head\\(r_i$ $(c_1,\ldots,c_n)) \not\in P(I)$, there exists a node $v$ as above but with $\successLabel(v) = \redF$. In both cases, $v$ is connected to the tuple node $\headOf{r_i(c_1, \ldots, c_n)}$.
  \item \textbf{Goal nodes:} Let $v$ be the node corresponding to a derivation $r_i(c_1, \ldots, c_n)$ with $m$ goals. If $\successLabel(v) = \greenT$, then for all $j \in \{1,\ldots,m\}$, $v$ is connected to a goal node $v_j$ labeled $g_i^j$  with $\successLabel(v_j) = \greenT$. 
If $\successLabel(v) = \redF$, then for all $j \in \{1,\ldots,m\}$, $v$ is connected to a goal node $v_j$ with $\successLabel(v_j) = \redF$
if the $j^{th}$ goal is failed in $r_i(c_1, \ldots, c_n)$. Each goal is connected to the corresponding tuple node.   \end{compactitem}
\end{Definition}
Our provenance graphs model query evaluation by construction. 
A tuple node $R(t)$ is successful in $\provGraph(P,I)$ iff $R(t) \in P(I)$. This is guaranteed, because each tuple built from values of the associated domain exists as a node $v$ in the graph and its label $\successLabel(v)$ is decided based on $R(t) \in P(I)$.  Furthermore, there exists a successful rule node $r(\vec{c}) \in \provGraph(P,I)$ iff the derivation $r(\vec{c})$ succeeds for $I$. Likewise, a failed rule node $r(\vec{c})$ exists iff the derivation $r(\vec{c})$ is failed over $I$ and $head(r(\vec{c})) \not\in P(I)$. 
Fig.\,\ref{fig:exam-pg-why-NY-seattle} and \ref{fig:exam-pg-whynot-chicago-seattle} show subgraphs of $\provGraph(P,I)$ for the query from Fig.\,\ref{fig:running-example-db}. Since $\rel{Q}(n,s) \in P(I)$ (Fig.\,\ref{fig:exam-pg-why-NY-seattle}), this tuple node is connected to all successful 
 derivations with $\rel{Q}(n,s)$ in the head which in turn are connected to goal nodes for each of the three goals of rule $r_1$. 
$\rel{Q}(s,n) \notin P(I)$ (Fig.\,\ref{fig:exam-pg-whynot-chicago-seattle}) and, thus, its node is connected to all failed derivations with $\rel{Q}(s,n)$ as a head. Here, we have assumed that all cities can be considered as start and end points of missing train connections, i.e., both $\domA(\rel{T.fromCity})$ and $\domA(\rel{T.toCity})$ are defined as $\adom{\rel{T.fromCity}} \cup \adom{\rel{T.toCity}}$. 
Thus, we have considered derivations $r_1(s,n,Z)$ for $Z \in \{c,n,s,w\}$.

\BGDel{\mypartitle{Relationship to Semiring Provenance}
Similar types of provenance graphs have been successfully used as graph representations of  provenance polynomials
(e.g., see~\cite{DG15c,grigoris-tj-simgodrec-2012}).
In fact, for positive queries we can apply graph transformations that translate a provenance graph into a simpler graph that encodes provenance polynomials $\ProvPoly$ or other semirings 
which correspond to less powerful provenance models (see Sec.\,\ref{sec:simp-hierarchy}). 
}

\subsection{Provenance Questions and Explanations}
\label{sec:question-expl}

Recall that the problem we address in this work is how to explain the existence or absence
of (sets of) tuples using provenance graphs.
 Such a set of tuples specified as a pattern and paired with a qualifier ($\whyq$/$\whynotq$) is called a \textit{provenance question} (\textit{PQ}) in this paper. The two questions presented in Example\,\ref{ex:example1} use constants only,
but we also support provenance questions with variables, 
e.g., for a question $\whynotq \rel{Q}(n,X)$ we return all explanations for missing
tuples where the first attribute is $n$, i.e., why it is not the case that a city $X$
can be reached from New York with one transfer, but not directly.
We say a tuple $t'$ of constants  \emph{matches} a tuple $t$ of variables and constants written as  $t' \matches t$
if we can unify $t'$ with $t$, i.e., we can equate $t'$ with $t$ by applying a valuation that substitutes variables
in $t$ with constants from $t'$.

\begin{Definition}[Provenance Question]\label{def:question}
Let $P$ be a  query, 
$I$ an instance, $Q$ an IDB predicate, and $\domA$ a domain assignment for $I$.
A \emph{provenance question} $\aProvQ$ is of the form $\whyq Q(t)$ or $\whynotq Q(t)$ where $t = (v_1, \ldots, v_n)$ consists of variables and domain constants ($\domA(\rel{Q.A})$ for each attribute $\rel{Q.A}$). 
We define:
\end{Definition}
\vspace{-5mm}
{\small
\begin{align*}
\qPattern(\aProvQ) &= Q(t)\\
\qMatch(\whyq Q(t)) &= \{Q(t') | t' \in P(I) \wedge t' \matches t\}\\
\qMatch(\whynotq Q(t)) &= \{Q(t') | t' \notin P(I) \wedge t' \matches t \wedge t' \in \tupDom(Q) \}
\end{align*}
}
\vspace{-4mm}

In Example\,\ref{ex:example-why} and \ref{ex:example-whynot}, we have presented subgraphs of $\provGraph(P,I)$ 
as \textit{explanations} for \provQ s,
implicitly claiming that these subgraphs are sufficient for explaining these \provQ{}s. 
We now formally define this type of explanation. 

\begin{Definition}[Explanation]\label{def:explanation}
The \emph{explanation} $\explainq(P,$ $\aProvQ,I,\domA)$ for a \provQ~$\aProvQ$ according to $P$, $I$, and $\domA$ is the subgraph of $\provGraph(P,I)$ containing
only nodes that are connected to at least one node in $\qMatch(\aProvQ)$. 
\end{Definition}

In the following we will drop $\domA$ from $\explainq(P,$ $\aProvQ,I,\domA)$ if it is clear from the context or irrelevant for the discussion.
Given this definition of explanation, note that 1) all nodes connected to a tuple node matching the \provQ{} are relevant for computing this tuple and 2) only nodes connected to this node are relevant for the outcome. 
Consider $Q(t') \in \qMatch(\aProvQ)$ for a question $\whyq Q(t)$. Since $Q(t') \in P(I)$,  all successful derivations with head $Q(t')$ justify the existence of $t'$ and these are precisely the rule nodes connected to $Q(t')$ in $\provGraph(P,I)$. 
For $\whynotq Q(t)$ and matching $Q(t')$ we have $Q(t') \not\in P(I)$ which is the case if all derivations with head $Q(t')$ have failed. In this case, all such derivations are connected to $Q(t')$ in the provenance graph. Each such derivation is connected to all of its failed goals which are responsible for the failure. 
Now, if a rule body references IDB predicates, then the same argument can be applied to reason that all rules directly connected to these tuples explain why they (do not) exist. Thus, by induction, the explanation contains all relevant tuple and rule nodes that explain the \provQ{}.

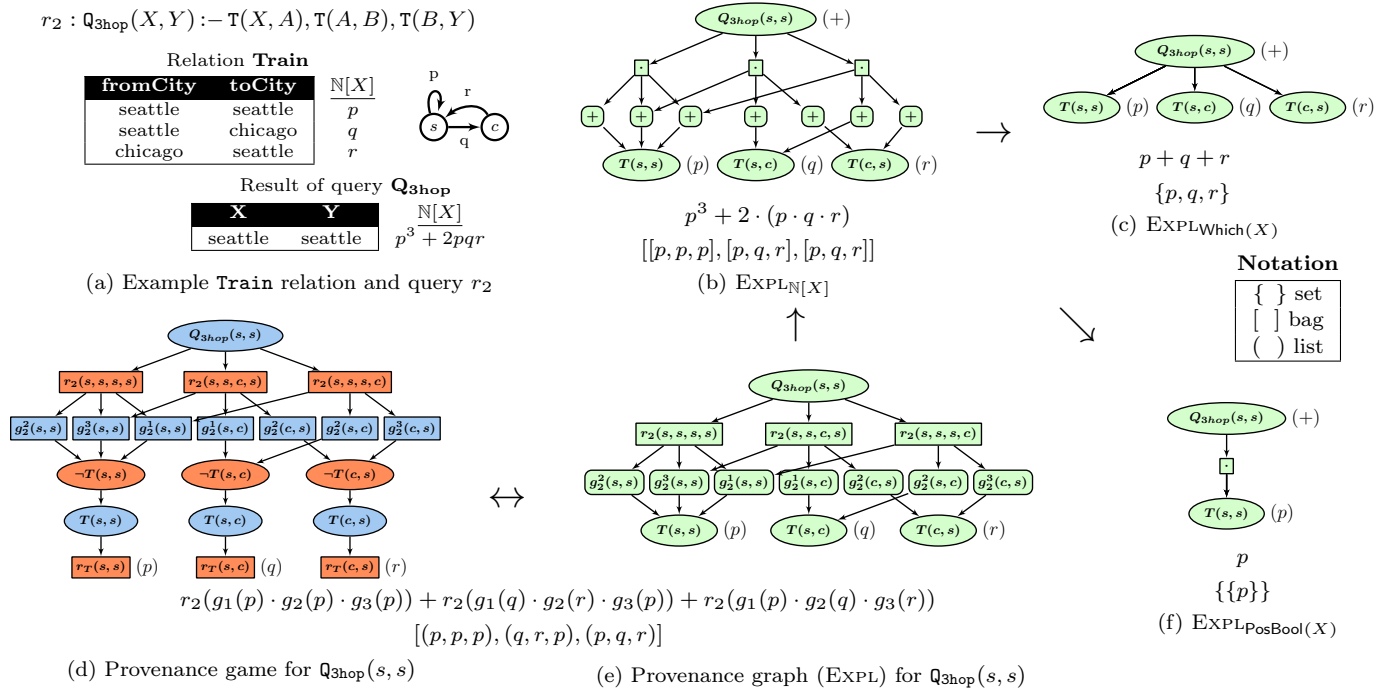
\begin{figure*}[t]
  \centering$\,$\\[-5mm]
    \begin{minipage}{0.35\linewidth}
    \centering     \hspace{-5mm}
    \begin{minipage}{0.75\linewidth}
      \centering 
      \begin{align*}
	r_2: \rel{Q_{3hop}}(X,Y) \dlImp \rel{T}(X,A), \rel{T}(A,B), \rel{T}(B,Y)\\
      \end{align*}
    \end{minipage}\\[-2mm]
     \begin{minipage}{0.75\linewidth}
      \centering
        \scriptsize \centering $\,$\\[-1mm]
        \begin{tabular}{|cc|c}
          \multicolumn{3}{c}{Relation \textbf{Train}}  \\[0.5mm]\cline{1-2}
	  \thead {fromCity} & \thead {toCity} & \underline{$\ProvPoly$} \\
          seattle & seattle & $p$\\
	  seattle & chicago & $q$\\
          chicago & seattle & $r$\\  
          \cline{1-2}
        \end{tabular}
    \end{minipage}    \begin{minipage}{0.14\linewidth}
       \resizebox{1.15\columnwidth}{!}{\begin{minipage}{1.4\linewidth}
\begin{tikzpicture}[>=latex',
line join=bevel,
line width=0.4mm,
every node/.style={ellipse},
minimum height=4mm]

  \definecolor{fillcolor}{rgb}{0.0,0.0,0.0};
  \node (s) at (30bp,30bp) [draw=black,circle] {$s$};
  \node (c) at (60bp,30bp) [draw=black,circle] {$c$};

   \path[]
   		(c) edge [in=40,out=120,->] node [above] {r} (s)
   		(s) edge [->] node [below] {q} (c)
		(s) edge [loop above] node [above] {p} (s);

\end{tikzpicture}
\end{minipage}

 }  
    \end{minipage}\hspace{-15mm}\\
    \begin{minipage}{0.35\linewidth}
	\scriptsize \centering$\,$\\[1mm]
         \begin{tabular}{|cc|c}
           \multicolumn{3}{c}{Result of query $\bf{Q_{3hop}}$}  \\[0.7mm]\cline{1-2}
 	  \thead {X} & \thead {Y} & \underline{$\ProvPoly$} \\
	   seattle & seattle & $p^3 + 2pqr $ \\
           \cline{1-2}
         \end{tabular}
     \end{minipage}$\,$\\[1mm]
    \subfloat[Example \rel{Train} relation and query $r_2$]{\hspace{65mm}\label{fig:partial-exam}}
    \end{minipage}
  \begin{minipage}{0.38\linewidth}
  	\centering
 	$\,$\\[3mm]
	\resizebox{0.73\columnwidth}{!}{\begin{tikzpicture}[>=latex',line join=bevel,line width=0.3mm]

  \definecolor{fillcolor}{rgb}{0.83,1.0,0.8};
  \node (REL_Q_WON_j_j_) at (339bp,159bp) [draw=black,fill=fillcolor,ellipse,label=right:{\large $(+)$}] {$\boldsymbol{Q_{3hop}(s,s)}$};

  \node (RULE_0_LOST_j_j_j_j_) at (269bp,129bp) [draw=black,fill=fillcolor,rectangle] {$\boldsymbol{\cdot}$};
  \node (RULE_0_LOST_j_j_a_j_) at (339bp,129bp) [draw=black,fill=fillcolor,rectangle] {$\boldsymbol{\cdot}$};
  \node (RULE_0_LOST_j_j_j_a_) at (404bp,129bp) [draw=black,fill=fillcolor,rectangle] {$\boldsymbol{\cdot}$};

  \node (GOaL_0_1_WON_j_j_) at (239bp,99bp) [draw=black,fill=fillcolor,rounded corners=.15cm,inner sep=3pt] {$\boldsymbol{+}$};
  \node (GOaL_0_2_WON_j_j_) at (269bp,99bp) [draw=black,fill=fillcolor,rounded corners=.15cm,inner sep=3pt] {$\boldsymbol{+}$};
  \node (GOaL_0_0_WON_j_j_) at (299bp,99bp) [draw=black,fill=fillcolor,rounded corners=.15cm,inner sep=3pt] {$\boldsymbol{+}$};
  \node (GOaL_0_0_WON_j_a_) at (339bp,99bp) [draw=black,fill=fillcolor,rounded corners=.15cm,inner sep=3pt] {$\boldsymbol{+}$};
  \node (GOaL_0_1_WON_a_j_) at (374bp,99bp) [draw=black,fill=fillcolor,rounded corners=.15cm,inner sep=3pt] {$\boldsymbol{+}$};
  \node (GOaL_0_1_WON_j_a_) at (404bp,99bp) [draw=black,fill=fillcolor,rounded corners=.15cm,inner sep=3pt] {$\boldsymbol{+}$};
  \node (GOaL_0_2_WON_a_j_) at (434bp,99bp) [draw=black,fill=fillcolor,rounded corners=.15cm,inner sep=3pt] {$\boldsymbol{+}$};

  \node (EDB_F_LOST_j_j_) at (269bp,69bp) [draw=black,fill=fillcolor,ellipse,label=right:{\large $(p)$}] {$\boldsymbol{T(s,s)}$};
  \node (EDB_F_LOST_j_a_) at (339bp,69bp) [draw=black,fill=fillcolor,ellipse,label=right:{\large $(q)$}] {$\boldsymbol{T(s,c)}$};
  \node (EDB_F_LOST_a_j_) at (409bp,69bp) [draw=black,fill=fillcolor,ellipse,label=right:{\large $(r)$}] {$\boldsymbol{T(c,s)}$};

  \draw [->] (GOaL_0_2_WON_a_j_) -> (EDB_F_LOST_a_j_);
  \draw [->] (RULE_0_LOST_j_j_j_a_) -> (GOaL_0_1_WON_j_a_);
  \draw [->] (GOaL_0_2_WON_j_j_) -> (EDB_F_LOST_j_j_);
  \draw [->] (GOaL_0_1_WON_j_a_) -> (EDB_F_LOST_j_a_);
  \draw [->] (REL_Q_WON_j_j_) -> (RULE_0_LOST_j_j_j_j_);
  \draw [->] (REL_Q_WON_j_j_) -> (RULE_0_LOST_j_j_j_a_);
  \draw [->] (RULE_0_LOST_j_j_j_a_) -> (GOaL_0_2_WON_a_j_);
  \draw [->] (GOaL_0_1_WON_j_j_) -> (EDB_F_LOST_j_j_);
  \draw [->] (RULE_0_LOST_j_j_a_j_) -> (GOaL_0_2_WON_j_j_);
  \draw [->] (GOaL_0_1_WON_a_j_) -> (EDB_F_LOST_a_j_);
  \draw [->] (RULE_0_LOST_j_j_j_a_) -> (GOaL_0_0_WON_j_j_);
  \draw [->] (REL_Q_WON_j_j_) -> (RULE_0_LOST_j_j_a_j_);
  \draw [->] (GOaL_0_0_WON_j_j_) -> (EDB_F_LOST_j_j_);
  \draw [->] (RULE_0_LOST_j_j_j_j_) -> (GOaL_0_1_WON_j_j_);
  \draw [->] (GOaL_0_0_WON_j_a_) -> (EDB_F_LOST_j_a_);
  \draw [->] (RULE_0_LOST_j_j_a_j_) -> (GOaL_0_0_WON_j_a_);
  \draw [->] (RULE_0_LOST_j_j_j_j_) -> (GOaL_0_2_WON_j_j_);
  \draw [->] (RULE_0_LOST_j_j_a_j_) -> (GOaL_0_1_WON_a_j_);
  \draw [->] (RULE_0_LOST_j_j_j_j_) -> (GOaL_0_0_WON_j_j_);
\end{tikzpicture}

 }
	\begin{minipage}{.43\linewidth}
	$\,$\\[-5mm]
	 \begin{align*}
          {p}^3 &+ 2 \cdot (p \cdot q \cdot r) \\
          [ [p,p,p&], [p,q,r], [p,q,r] ]
         \end{align*}
	\end{minipage}
    	$\,$\\[-1mm]
	\subfloat[$\explainq_{\ProvPoly}$]{\hspace{50mm}\label{fig:poly-model}}
  \end{minipage}
  {\Large \hspace{-8mm} $\rightarrow$}\hspace{3mm}
  \begin{minipage}{0.2\linewidth}
  	\centering 	\resizebox{1.29\columnwidth}{!}{\begin{tikzpicture}[>=latex',line join=bevel,line width=0.3mm]

  \definecolor{fillcolor}{rgb}{0.83,1.0,0.8};
  \node (REL_Q_WON_J_J_) at (374bp,44bp) [draw=black,fill=fillcolor,ellipse,label=right:{\large $(+)$}] {$\boldsymbol{Q_{3hop}(s,s)}$};

  \node (EDB_SCN_LOST_J_J_) at (304bp,9bp) [draw=black,fill=fillcolor,ellipse,label=right:{\large $(p)$}] {$\boldsymbol{T(s,s)}$};
  \node (EDB_SCN_LOST_J_A_) at (374bp,9bp) [draw=black,fill=fillcolor,ellipse,label=right:{\large $(q)$}] {$\boldsymbol{T(s,c)}$};
  \node (EDB_SCN_LOST_A_J_) at (444bp,9bp) [draw=black,fill=fillcolor,ellipse,label=right:{\large $(r)$}] {$\boldsymbol{T(c,s)}$};

  \draw [->] (REL_Q_WON_J_J_) -> (EDB_SCN_LOST_J_A_);
  \draw [->] (REL_Q_WON_J_J_) -> (EDB_SCN_LOST_A_J_);  
  \draw [->] (REL_Q_WON_J_J_) -> (EDB_SCN_LOST_J_J_);
 
  \draw [->] (REL_Q_WON_J_J_) -> (EDB_SCN_LOST_J_J_);
  \draw [->] (REL_Q_WON_J_J_) -> (EDB_SCN_LOST_J_J_);  
  \draw [->] (REL_Q_WON_J_J_) -> (EDB_SCN_LOST_J_A_);
 
  \draw [->] (REL_Q_WON_J_J_) -> (EDB_SCN_LOST_J_J_);
  \draw [->] (REL_Q_WON_J_J_) -> (EDB_SCN_LOST_A_J_);  
  \draw [->] (REL_Q_WON_J_J_) -> (EDB_SCN_LOST_J_J_);

\end{tikzpicture}

 }\\
        \hspace{3mm}
	\begin{minipage}{.1\linewidth}
	$\,$\\[-5mm] \centering
	  \begin{align*}
        p + q + r \\
          \{ p,q,r \}
         \end{align*}
	\end{minipage}
	$\,$\\[-1mm]
    \subfloat[$\explainq_{\WhichProv}$]{\hspace{40mm}\label{fig:lin-model}}
  \end{minipage}$\,$\\[1mm]
{\Large \hspace{70mm}$\uparrow$ \hspace{30mm} $\searrow$}$\,$\\[-1mm]
\vspace{-3mm}
\begin{minipage}{0.35\linewidth}
	\centering \hspace{-10mm}
	\resizebox{0.95\columnwidth}{!}{    \vspace{-20mm}\begin{tikzpicture}[>=latex',line join=bevel,line width=0.3mm]

  \definecolor{fillcolor}{rgb}{0.63,0.79,0.95};
  \node (REL_Q_WON_j_j_) at (374bp,159bp) [draw=black,fill=fillcolor,ellipse] {$\boldsymbol{Q_{3hop}(s,s)}$};

  \definecolor{fillcolor}{rgb}{1.0,0.55,0.35};
  \node (RULE_0_LOST_j_j_j_j_) at (294bp,129bp) [draw=black,fill=fillcolor,rectangle] {$\boldsymbol{r_2(s,s,s,s)}$};
  \definecolor{fillcolor}{rgb}{1.0,0.55,0.35};
  \node (RULE_0_LOST_j_j_a_j_) at (374bp,129bp) [draw=black,fill=fillcolor,rectangle] {$\boldsymbol{r_2(s,s,c,s)}$};
  \definecolor{fillcolor}{rgb}{1.0,0.55,0.35};
  \node (RULE_0_LOST_j_j_j_a_) at (454bp,129bp) [draw=black,fill=fillcolor,rectangle] {$\boldsymbol{r_2(s,s,s,c)}$};

  \definecolor{fillcolor}{rgb}{0.63,0.79,0.95};
  \node (GOaL_0_1_WON_j_j_) at (254bp,99bp) [draw=black,fill=fillcolor,rectangle] {$\boldsymbol{g_{2}^{2}(s,s)}$};
\definecolor{fillcolor}{rgb}{0.63,0.79,0.95};
  \node (GOaL_0_2_WON_j_j_) at (294bp,99bp) [draw=black,fill=fillcolor,rectangle] {$\boldsymbol{g_{2}^{3}(s,s)}$};
  \definecolor{fillcolor}{rgb}{0.63,0.79,0.95};
  \node (GOaL_0_0_WON_j_j_) at (334bp,99bp) [draw=black,fill=fillcolor,rectangle] {$\boldsymbol{g_{2}^{1}(s,s)}$};
  \definecolor{fillcolor}{rgb}{0.63,0.79,0.95};
  \node (GOaL_0_0_WON_j_a_) at (374bp,99bp) [draw=black,fill=fillcolor,rectangle] {$\boldsymbol{g_{2}^{1}(s,c)}$};
  \definecolor{fillcolor}{rgb}{0.63,0.79,0.95};
  \node (GOaL_0_1_WON_a_j_) at (414bp,99bp) [draw=black,fill=fillcolor,rectangle] {$\boldsymbol{g_{2}^{2}(c,s)}$};
  \definecolor{fillcolor}{rgb}{0.63,0.79,0.95};
  \node (GOaL_0_1_WON_j_a_) at (454bp,99bp) [draw=black,fill=fillcolor,rectangle] {$\boldsymbol{g_{2}^{2}(s,c)}$};
  \definecolor{fillcolor}{rgb}{0.63,0.79,0.95};
  \node (GOaL_0_2_WON_a_j_) at (494bp,99bp) [draw=black,fill=fillcolor,rectangle] {$\boldsymbol{g_{2}^{3}(c,s)}$};

  \definecolor{fillcolor}{rgb}{1.0,0.55,0.35};
  \node (notREL_F_LOST_j_j_) at (294bp,69bp) [draw=black,fill=fillcolor,ellipse] {$\boldsymbol{\neg T(s,s)}$};
  \definecolor{fillcolor}{rgb}{1.0,0.55,0.35};
  \node (notREL_F_LOST_j_a_) at (374bp,69bp) [draw=black,fill=fillcolor,ellipse] {$\boldsymbol{\neg T(s,c)}$};
  \definecolor{fillcolor}{rgb}{1.0,0.55,0.35};
  \node (notREL_F_LOST_a_j_) at (454bp,69bp) [draw=black,fill=fillcolor,ellipse] {$\boldsymbol{\neg T(c,s)}$};

  \definecolor{fillcolor}{rgb}{0.63,0.79,0.95};
  \node (REL_F_WON_j_j_) at (294bp,39bp) [draw=black,fill=fillcolor,ellipse] {$\boldsymbol{T(s,s)}$};
  \definecolor{fillcolor}{rgb}{0.63,0.79,0.95};
  \node (REL_F_WON_j_a_) at (374bp,39bp) [draw=black,fill=fillcolor,ellipse] {$\boldsymbol{T(s,c)}$};
  \definecolor{fillcolor}{rgb}{0.63,0.79,0.95};
  \node (REL_F_WON_a_j_) at (454bp,39bp) [draw=black,fill=fillcolor,ellipse] {$\boldsymbol{T(c,s)}$};

  \definecolor{fillcolor}{rgb}{1.0,0.55,0.35};
  \node (EDB_F_LOST_j_j_) at (294bp,9bp) [draw=black,fill=fillcolor,rectangle,label=right:{\large $(p)$}] {$\boldsymbol{r_{T}(s,s)}$};
  \definecolor{fillcolor}{rgb}{1.0,0.55,0.35};
  \node (EDB_F_LOST_j_a_) at (374bp,9bp) [draw=black,fill=fillcolor,rectangle,label=right:{\large $(q)$}] {$\boldsymbol{r_{T}(s,c)}$};
  \definecolor{fillcolor}{rgb}{1.0,0.55,0.35};
  \node (EDB_F_LOST_a_j_) at (454bp,9bp) [draw=black,fill=fillcolor,rectangle,label=right:{\large $(r)$}] {$\boldsymbol{r_{T}(c,s)}$};

  \draw [->] (GOaL_0_2_WON_a_j_) -> (notREL_F_LOST_a_j_);
  \draw [->] (RULE_0_LOST_j_j_j_a_) -> (GOaL_0_1_WON_j_a_);
  \draw [->] (notREL_F_LOST_j_j_) -> (REL_F_WON_j_j_);
  \draw [->] (notREL_F_LOST_j_a_) -> (REL_F_WON_j_a_);
  \draw [->] (GOaL_0_2_WON_j_j_) -> (notREL_F_LOST_j_j_);
  \draw [->] (REL_F_WON_a_j_) -> (EDB_F_LOST_a_j_);
  \draw [->] (GOaL_0_1_WON_j_a_) -> (notREL_F_LOST_j_a_);
  \draw [->] (REL_Q_WON_j_j_) -> (RULE_0_LOST_j_j_j_j_);
  \draw [->] (REL_F_WON_j_a_) -> (EDB_F_LOST_j_a_);
  \draw [->] (REL_Q_WON_j_j_) -> (RULE_0_LOST_j_j_j_a_);
  \draw [->] (RULE_0_LOST_j_j_j_a_) -> (GOaL_0_2_WON_a_j_);
  \draw [->] (GOaL_0_1_WON_j_j_) -> (notREL_F_LOST_j_j_);
  \draw [->] (REL_F_WON_j_j_) -> (EDB_F_LOST_j_j_);
  \draw [->] (RULE_0_LOST_j_j_a_j_) -> (GOaL_0_2_WON_j_j_);
  \draw [->] (GOaL_0_1_WON_a_j_) -> (notREL_F_LOST_a_j_);
  \draw [->] (RULE_0_LOST_j_j_j_a_) -> (GOaL_0_0_WON_j_j_);
  \draw [->] (REL_Q_WON_j_j_) -> (RULE_0_LOST_j_j_a_j_);
  \draw [->] (GOaL_0_0_WON_j_j_) -> (notREL_F_LOST_j_j_);
  \draw [->] (RULE_0_LOST_j_j_j_j_) -> (GOaL_0_1_WON_j_j_);
  \draw [->] (GOaL_0_0_WON_j_a_) -> (notREL_F_LOST_j_a_);
  \draw [->] (RULE_0_LOST_j_j_a_j_) -> (GOaL_0_0_WON_j_a_);
  \draw [->] (RULE_0_LOST_j_j_j_j_) -> (GOaL_0_2_WON_j_j_);
  \draw [->] (RULE_0_LOST_j_j_a_j_) -> (GOaL_0_1_WON_a_j_);
  \draw [->] (notREL_F_LOST_a_j_) -> (REL_F_WON_a_j_);
  \draw [->] (RULE_0_LOST_j_j_j_j_) -> (GOaL_0_0_WON_j_j_);
\end{tikzpicture}

 }
    \hspace{2cm} \begin{minipage}{0.73\linewidth}
      \centering
	$\,$\\[-7mm]
         \begin{align*}\hspace{20mm}
          r_2(g_1(p) \cdot g_2(p) \cdot g_3(p)) &+ r_2(g_1(q) \cdot g_2(r) \cdot g_3(p)) + r_2(g_1(p) \cdot g_2(q) \cdot g_3(r)) \\
          &[ (p,p,p), (q,r,p), (p,q,r) ]
         \end{align*}
        \end{minipage}
    	$\,$\\[-1mm]
    \begin{minipage}{0.9\linewidth}
    \subfloat[Provenance game for $\rel{Q_{3hop}}(s,s)$]{\hspace{50mm}\label{fig:pgame-3hop}}       
    \end{minipage}
  \end{minipage}  {\Large \hspace{-3mm} $\leftrightarrow$}\hspace{6mm}
  \begin{minipage}{0.35\linewidth}
  	\centering
	$\,$\\[7mm]
	\resizebox{0.98\columnwidth}{!}{\begin{tikzpicture}[>=latex',line join=bevel,line width=0.3mm]

  \definecolor{fillcolor}{rgb}{0.83,1.0,0.8};
  \node (REL_Q_WON_j_j_) at (374bp,159bp) [draw=black,fill=fillcolor,ellipse] {$\boldsymbol{Q_{3hop}(s,s)}$};

  \node (RULE_0_LOST_j_j_j_j_) at (294bp,129bp) [draw=black,fill=fillcolor,rectangle] {$\boldsymbol{r_2(s,s,s,s)}$};
  \node (RULE_0_LOST_j_j_a_j_) at (374bp,129bp) [draw=black,fill=fillcolor,rectangle] {$\boldsymbol{r_2(s,s,c,s)}$};
  \node (RULE_0_LOST_j_j_j_a_) at (454bp,129bp) [draw=black,fill=fillcolor,rectangle] {$\boldsymbol{r_2(s,s,s,c)}$};

  \node (GOaL_0_1_WON_j_j_) at (254bp,99bp) [draw=black,fill=fillcolor,rounded corners=.15cm,inner sep=3pt] {$\boldsymbol{g^{2}_{2}(s,s)}$};
  \node (GOaL_0_2_WON_j_j_) at (294bp,99bp) [draw=black,fill=fillcolor,rounded corners=.15cm,inner sep=3pt] {$\boldsymbol{g^{3}_{2}(s,s)}$};
  \node (GOaL_0_0_WON_j_j_) at (334bp,99bp) [draw=black,fill=fillcolor,rounded corners=.15cm,inner sep=3pt] {$\boldsymbol{g^{1}_{2}(s,s)}$};
  \node (GOaL_0_0_WON_j_a_) at (374bp,99bp) [draw=black,fill=fillcolor,rounded corners=.15cm,inner sep=3pt] {$\boldsymbol{g^{1}_{2}(s,c)}$};
  \node (GOaL_0_1_WON_a_j_) at (414bp,99bp) [draw=black,fill=fillcolor,rounded corners=.15cm,inner sep=3pt] {$\boldsymbol{g^{2}_{2}(c,s)}$};
  \node (GOaL_0_1_WON_j_a_) at (454bp,99bp) [draw=black,fill=fillcolor,rounded corners=.15cm,inner sep=3pt] {$\boldsymbol{g^{2}_{2}(s,c)}$};
  \node (GOaL_0_2_WON_a_j_) at (494bp,99bp) [draw=black,fill=fillcolor,rounded corners=.15cm,inner sep=3pt] {$\boldsymbol{g^{3}_{2}(c,s)}$};

  \node (EDB_F_LOST_j_j_) at (294bp,69bp) [draw=black,fill=fillcolor,ellipse,label=right:{\large $(p)$}] {$\boldsymbol{T(s,s)}$};
  \node (EDB_F_LOST_j_a_) at (374bp,69bp) [draw=black,fill=fillcolor,ellipse,label=right:{\large $(q)$}] {$\boldsymbol{T(s,c)}$};
  \node (EDB_F_LOST_a_j_) at (454bp,69bp) [draw=black,fill=fillcolor,ellipse,label=right:{\large $(r)$}] {$\boldsymbol{T(c,s)}$};

  \draw [->] (GOaL_0_2_WON_a_j_) -> (EDB_F_LOST_a_j_);
  \draw [->] (RULE_0_LOST_j_j_j_a_) -> (GOaL_0_1_WON_j_a_);
  \draw [->] (GOaL_0_2_WON_j_j_) -> (EDB_F_LOST_j_j_);
  \draw [->] (GOaL_0_1_WON_j_a_) -> (EDB_F_LOST_j_a_);
  \draw [->] (REL_Q_WON_j_j_) -> (RULE_0_LOST_j_j_j_j_);
  \draw [->] (REL_Q_WON_j_j_) -> (RULE_0_LOST_j_j_j_a_);
  \draw [->] (RULE_0_LOST_j_j_j_a_) -> (GOaL_0_2_WON_a_j_);
  \draw [->] (GOaL_0_1_WON_j_j_) -> (EDB_F_LOST_j_j_);
  \draw [->] (RULE_0_LOST_j_j_a_j_) -> (GOaL_0_2_WON_j_j_);
  \draw [->] (GOaL_0_1_WON_a_j_) -> (EDB_F_LOST_a_j_);
  \draw [->] (RULE_0_LOST_j_j_j_a_) -> (GOaL_0_0_WON_j_j_);
  \draw [->] (REL_Q_WON_j_j_) -> (RULE_0_LOST_j_j_a_j_);
  \draw [->] (GOaL_0_0_WON_j_j_) -> (EDB_F_LOST_j_j_);
  \draw [->] (RULE_0_LOST_j_j_j_j_) -> (GOaL_0_1_WON_j_j_);
  \draw [->] (GOaL_0_0_WON_j_a_) -> (EDB_F_LOST_j_a_);
  \draw [->] (RULE_0_LOST_j_j_a_j_) -> (GOaL_0_0_WON_j_a_);
  \draw [->] (RULE_0_LOST_j_j_j_j_) -> (GOaL_0_2_WON_j_j_);
  \draw [->] (RULE_0_LOST_j_j_a_j_) -> (GOaL_0_1_WON_a_j_);
  \draw [->] (RULE_0_LOST_j_j_j_j_) -> (GOaL_0_0_WON_j_j_);
\end{tikzpicture}

 }
        \begin{minipage}{0.43\linewidth}
	 $\,$\\[-5mm]
	\end{minipage}
    	$\,$\\[7mm]
    \subfloat[Provenance graph ($\explainq$) for $\rel{Q_{3hop}}(s,s)$]{\hspace{58mm}\label{fig:pg-3hop}}
  \end{minipage}  \begin{minipage}{0.2\linewidth}
    \centering $\,$\\[-15mm] 
    \begin{minipage}{1.8\linewidth}
      \centering \small
      \textbf{Notation}\\[1mm]  
	  \begin{tabular}{|c|} \hline
          $\{ \hspace{1.5mm} \}$ set \\ 
          $[ \hspace{2mm} ]$ bag \\
          $( \hspace{2mm} )$ list \\
          \hline
        \end{tabular}
     \end{minipage}$\,$\\[5mm] 
     \begin{minipage}{1.5\linewidth}
       \centering 
	\resizebox{0.41\columnwidth}{!}{\begin{tikzpicture}[>=latex',line join=bevel,line width=0.3mm]

  \definecolor{fillcolor}{rgb}{0.83,1.0,0.8};
  \node (REL_Q_WON_J_J_) at (344bp,69bp) [draw=black,fill=fillcolor,ellipse,label=right:{\large $(+)$}] {$\boldsymbol{Q_{3hop}(s,s)}$};

  \node (RULE_0_LOST_J_J_J_J_) at (344bp,39bp) [draw=black,fill=fillcolor,rectangle] {$\boldsymbol{\cdot}$};

  \node (EDB_SCN_LOST_J_J_) at (344bp,9bp) [draw=black,fill=fillcolor,ellipse,label=right:{\large $(p)$}] {$\boldsymbol{T(s,s)}$};

  \draw [->] (REL_Q_WON_J_J_) -> (RULE_0_LOST_J_J_J_J_);

  \draw [->] (RULE_0_LOST_J_J_J_J_) -> (EDB_SCN_LOST_J_J_);  
  \draw [->] (RULE_0_LOST_J_J_J_J_) -> (EDB_SCN_LOST_J_J_);
  \draw [->] (RULE_0_LOST_J_J_J_J_) -> (EDB_SCN_LOST_J_J_);

\end{tikzpicture}

 }\\
	\begin{minipage}{.12\linewidth}
	$\,$\\[-6mm]
         \begin{align*}
           &p\\
           \{\{ &p \}\}
         \end{align*}
	\end{minipage}
	$\,$\\[-1mm]
    \subfloat[$\explainq_{\PosBool}$]{\hspace{50mm}\label{fig:posb-model}}
    \end{minipage}
  \end{minipage}
  $\,$\\[-3mm]
  \caption{Transformations exemplified using the provenance graph for $\rel{Q_{3hop}}(s,s)$. For each graph, we show the structure of the provenance encoded by this graph and the corresponding semiring annotation where applicable.}
  \label{fig:mig-exam-why}
\end{figure*}

\section{Provenance Graphs and Provenance Games}
\label{sec:prov-game}
We now prove that provenance graphs according to Def.\,\ref{def:prov=graph} are equivalent to provenance games. Thus, our model inherits the semantic foundation of provenance games. Specifically, provenance games were shown to encode Datalog query evaluation. Furthermore, the interpretation of provenance game graphs as 2-player games provides a strong justification for why the nodes reachable from a tuple node justify the existence/absence of the tuple.
We show how to transform a provenance game $\gprov(P,\aProvQ,I)$ into an explanation $\explainq(P,\aProvQ,I)$ and vice versa to demonstrate that both are equivalent representations of provenance. We define a function $\gameToProv$ that maps provenance games to graphs and its inverse $\provToGame$. Before that, we first give an overview of provenance games (see~\cite{KL13} for more details).

\mypara{Provenance Games}
Similar to our provenance graphs, provenance games are graphs that record successful and failed rule derivations. Provenance games consist of four types of nodes (e.g. Fig.\,\ref{fig:pgame-3hop}): rule nodes (boxes labeled with a rule identifier and the constant arguments of a rule derivation), goal nodes (boxes labeled with a rule identifier and the goal's position in the rule's body), tuple nodes (ovals), and EDB fact nodes (boxes labeled with an EDB relation name and the constants of a tuple).
Every tuple node in a provenance game appears both positively and negatively, i.e., for every tuple node $R(t)$, there exists a tuple node $\neg R(t)$.
Given a program $P$ and database instance $I$, a provenance game is constructed by creating
a positive and negative tuple node
$\rel{R}(c_1 , \cdots , c_n)$ for each \textit{n}-ary predicate $\rel{R}$ and for all combinations of constants $c_i$ from the active domain \adom{I}.
Similarly, nodes are created for rule derivations, i.e., a rule where variables have been replaced with constants from \adom{I} and each goal in the body of a rule (similar to Def.\,\ref{def:prov=graph}).
In the game, a derivation of rule $r$ for a vector of constants $\vec{c}$ is labeled as $r(\vec{c})$, e.g.,
a derivation $\rel{Q_{3hop}}(s,s) \dlImp \rel{T}(s,c), \rel{T}(c,s), \rel{T}(s,s)$ of $r_2$ in Fig.\,\ref{fig:partial-exam} is represented as a rule node labeled with $r_2(s,s,c,s)$.
Finally, EDB fact nodes are added for each tuple in $I$, e.g., $r_T (s,s)$ for the tuple (\cnst{seattle}, \cnst{seattle}) in the \rel{Train} relation (Fig.\,\ref{fig:partial-exam}).
Tuple nodes are connected to the grounded rule nodes that derive them (have the tuple in their head), rule nodes to goal nodes for the grounded
goals in their body, and goal nodes to negated tuple nodes corresponding to the goal (positive goals) or positive tuple nodes (negated goals).
Such a game is interpreted as a 2-player game where the players argue for/against the existence of a tuple in the result of evaluating $P$ over $I$. The existence of strategies for a player in this game determines tuple existence and success of rule derivations. A \textit{solved game} is one where each node in the game graph is labeled as either won $\wonLabel$ (there exists a strategy for the player starting in this position) or lost $\lostLabel$ (no such strategy exists).
A tuple node $\rel{R}(t)$ is labeled as $\wonLabel$ iff the tuple $\rel{R}(t)$ exists. A corollary of this is that a rule is labelled $\lostLabel$ if the corresponding derivation is successful and $\wonLabel$ otherwise.\footnote{This follows from the semantics of the type of 2-player game used here. The details are beyond the scope of this paper.}
Given such a solved game (denoted as $\gprov(P,I)$), we can extract a subgraph rooted at an IDB tuple $\rel{Q}(t)$ as the provenance of $\rel{Q}(t)$. Similar to how we derive an explanation for a PQ with $\qPattern(\aProvQ) = Q(t)$ where $t$ may contain variables as the subgraph of the provenance graph $\provGraph(P,I)$ containing all IDB tuple nodes matching $t$ and nodes reachable from these nodes, we can derive the corresponding subgraph in the provenance game $\gprov(P,I)$ and denote it as $\gprov(P,\aProvQ,I)$ (we call such subgraphs \textit{game explanations}).

\mypartitle{Translating between Provenance Graphs and Provenance Games} The translation $\provToGame$ of a provenance graph into the corresponding game and the reverse transformation $\gameToProv$ are straightforward. Thus, we only sketch $\provToGame$ here. EDB tuple nodes are expanded to subgraphs $\neg \rel{R}(t) \to \rel{R}(t) \to r_{R}(t)$ for existing tuples and $\neg \rel{R}(t) \to \rel{R}(t)$ for missing tuples. IDB tuple nodes are always expanded to subgraphs of the later form. Rule and goal nodes and their inter-connections are preserved. Goal nodes are connected to negated tuple nodes (positive goals) and to positive tuple nodes (negated goals). For positive tuple and goal nodes, we translate $\greenT$ to $\wonLabel$ (won) and $\redF$ to $\lostLabel$ (lost). For negated tuple nodes and rule nodes, this mapping is reversed, i.e.,  $\greenT$ to $\lostLabel$ and $\redF$ to $\wonLabel$.

\iftechreport{
\mypartitle{\provToGame}
$\provToGame$ consists of the following steps executed in the order shown below:
\begin{itemize}
\item Replace each EDB tuple node $v$ with $\nodeLabel(v) = \rel{R}(t)$ and $\successLabel(v) = \greenT$ with a subgraph $v_1 \to v_2 \to v_3$ where $v_1 = \neg \rel{R}(t)$, $v_2= \rel{R}(t)$, and $v_3 = r_{R}(t)$. Label these nodes as follows: $\gameLabel(v_1) = \gameLabel(v_3) = \lostLabel$ and $\gameLabel(v_2) = \wonLabel$.
\item Replace each EDB tuple node $v$ with $\nodeLabel(v) = \rel{R}(t)$ and $\successLabel(v) = \redF$
  with a subgraph $v_1 \to v_2$ where $v_1 = \neg \rel{R}(t)$ and $v_2= \rel{R}(t)$ using labels $\gameLabel(v_1) = \wonLabel$ and $\gameLabel(v_2) = \lostLabel$.
\item Replace each IDB tuple node $v$ with $\nodeLabel(v) = \rel{R}(t)$ with a subgraph $v_1 \to v_2$ where $v_1 = \neg \rel{R}(t)$ and $v_2= \rel{R}(t)$. If $\successLabel(v) = \greenT$, then
$\gameLabel(v_1) = \lostLabel$ and $\gameLabel(v_2) = \wonLabel$. Otherwise, $\gameLabel(v_1) = \wonLabel$ and $\gameLabel(v_2) = \lostLabel$.
\item For each rule node $v$ set $\gameLabel(v) = \lostLabel$ if $\successLabel(v) = \greenT$ and $\gameLabel(v) = \wonLabel$ otherwise. For each goal node $v$ set $\gameLabel(v) = \wonLabel$ if $\successLabel(v) = \greenT$ and $\gameLabel(v) = \lostLabel$ otherwise.
\item Edges between rule and goal nodes are preserved unmodified. Edges from tuple nodes to rule nodes stem from the positive tuple node. Consider a goal node $v$ with $\successLabel(v) = \greenT$ that is connected to a tuple node $v'$. If $\successLabel(v) = \successLabel(v')$ (positive goals), then $v$ is connected to the negative tuple node derived from $v'$; if $\successLabel(v) \neq \successLabel(v')$ then $v$ is connected to the positive tuple node.
\end{itemize}

\mypartitle{\gameToProv} The reverse translation is defined as below: 

\begin{itemize}
\item Remove all EDB fact nodes. \item Replace each subgraph  $v_1 \to v_2$ where $v_1$ is labeled $\neg \rel{R}(t)$ and $v_2$ is labeled $ \rel{R}(t)$ with a tuple node $v$ with $\nodeLabel(v) = \rel{R}(t)$ and $\successLabel(v) = \greenT$ if $\gameLabel(v_2) = \wonLabel$ and $\successLabel(v) = \redF$ otherwise. All incoming and outgoing edges to/from $v_1$ and $v_2$ are connected to $v$. \item For every rule node $v$ labeled $r(\vec{c})$, set $\successLabel(v) = \greenT$ if $\gameLabel(v) = \lostLabel$ and $\successLabel(v) = \redF$ otherwise.
\item For every goal node $v$ labeled $g_i^j(\vec{c})$, set $\successLabel(v) = \greenT$ if $\gameLabel(v) = \wonLabel$ and $\successLabel(v) = \redF$ otherwise.
\end{itemize}
}

\begin{Theorem}\label{theo:equi-pgame}   Let $P$ be a program, $I$ a database instance, and $\aProvQ$ a \provQ{}.
  We have:
  \begin{align*}\hspace{10mm}
    \gameToProv(\gprov(P,\aProvQ,I)) &= \explainq(P,\aProvQ,I)\\ \provToGame(\explainq(P,\aProvQ,I)) &= \gprov(P,\aProvQ,I)
  \end{align*}

\end{Theorem}
\begin{proof}
\ifnottechreport{See our accompanying technical report\cite{LL18}.}
\iftechreport{
   By construction,  the label of a tuple node in a provenance graph is $\greenT$ if the tuple exists and $\redF$ otherwise. Rule nodes are labeled $\greenT$ if the rule derivation is successful and $\redF$ otherwise. As shown in~\cite{KL13}, a tuple node in a provenance game is labeled $\wonLabel$ if the tuple exists and $\lostLabel$ otherwise. A rule node is labeled $\lostLabel$ if the derivation is successful and $\wonLabel$ else. Thus, based on the definition of $\gameToProv$ and $\provToGame$, tuple nodes are translated correctly. An analogous argument holds for rule nodes and goal nodes. What remains to be shown is that nodes are correctly connected in the result of $\gameToProv$ and $\provToGame$. In both provenance games and our provenance graphs, IDB tuple nodes are connected to the rule nodes deriving them. These edges are preserved by both mappings. Rule nodes representing successful derivations are connected to all corresponding goal nodes and, for unsuccessful derivations, only to goals that are failed. This is true for both provenance games and our provenance graphs. Both mappings preserve these edges. In provenance graphs, goal nodes are connected to the tuple node corresponding to the grounded goal.
 Finally, in provenance games, there are additional nodes (negated tuple nodes and EDB fact nodes). These nodes always co-occur with a positive tuple nodes in a fixed graph fragment, e.g., $\neg \rel{R}(t) \to \rel{R}(t) \to r_{R}(t)$. Mapping $\provToGame$ creates these fragments. Mapping $\gameToProv$ collapses these fragments and reroutes the incoming and outgoing edges of such a  subgraph to the tuple node the graph fragment is mapped to.
 Since explanations and game explanations are defined based on connectivity, this concludes the proof.
}
\end{proof}
\begin{Example}\label{ex:remove-idb-rel}
Consider rule $r_2$ in Fig.\,\ref{fig:partial-exam} computing
 which cities can be reached from another city  through a path of length 3. The provenance game and provenance graph for $\rel{Q_{3hop}}(s,s)$ are shown in Fig.\,\ref{fig:pgame-3hop}
and Fig.\,\ref{fig:pg-3hop}, respectively. In the provenance graph, goal nodes are directly connected to tuple nodes. In the game, they are represented as positive and negative tuple nodes and EDB fact nodes (the lower three levels).
That is, every subgraph $\neg T(X,Y) \rightarrow T(X,Y) \rightarrow r_{T}(X,Y)$ in Fig.\,\ref{fig:pgame-3hop}
is equivalently encoded as a single tuple node $T(X,Y)$ in Fig.\,\ref{fig:pg-3hop}. Both graphs record the 3 paths of length 3 which start and end in Seattle: 1) $s \rightarrow s \rightarrow s \rightarrow s$,
2) $s \rightarrow c \rightarrow s \rightarrow s$, and 3) $s \rightarrow s \rightarrow c \rightarrow s$.
\end{Example}

 \section{Semiring Provenance for Positive Queries}
\label{sec:semi-annot-model}

The semiring annotation model~\cite{grigoris-tj-simgodrec-2012,green2007provenance,GT17} is widely accepted as a provenance model for positive queries. An interesting question is how our model compares to provenance polynomials (semiring \ProvPoly), the most general form of annotation in the semiring model.
It was shown in~\cite{KL13} that, for positive queries, the result of a query annotated with semiring \ProvPoly can be extracted from the provenance game by applying a graph transformation.
The equivalence shown in Sec.\,\ref{sec:prov-game} extends this result to our provenance graph model. That being said, we develop simplified versions of our graph model to directly support less informative provenance semirings such as Lineage which only tracks data-dependencies between input and output tuples. We now introduce the semiring annotation framework for positive queries and its use in provenance tracking and, then, explain our simplified provenance graph types.

\subsection{\SomeK-relations}
\label{sec:k-relations}

In the semiring framework, relations are annotated with elements from a commutative semiring. A commutative semiring is a structure $\SomeK = (K, \KPlus, \KTimes, \Kzero, \Kone)$
over a set $K$ where the addition and multiplication operations are associative, commutative, and have a neutral element ($\Kzero$ and $\Kone$, respectively). Furthermore, multiplication with zero yields zero and multiplication distributes over addition. A relation annotated with the elements of a semiring $\SomeK$ is called a $\SomeK$-relation.
Operators of positive relational algebra ($\RAplus$) for $\SomeK$-relations compute annotations for tuples in their output by combining annotations from their input using the operations of the semiring. Intuitively, multiplication represents conjunctive use of inputs (as in a join) whereas addition represents alternative use of inputs (as in a union or projection).
We are interested in \SomeK-relations, because it was shown that many provenance types  can be expressed as semiring annotations.

Semiring homomorphisms are important for our purpose since they allow us to translate between different provenance semirings and understand their relative informativeness.
A semiring homomorphism $h: \SomeK_1 \to \SomeK_2$ is a function from $K_1$ to $K_2$ that respects the operations of semirings, e.g., $h(k_1 +_{\SomeK_1} k_2) = h(k_1) +_{\SomeK_2} h(k_2)$. As shown in~\cite{green2011containment}, if there exists a surjective homomorphism between one provenance semiring $\SomeK_1$ and another semiring $\SomeK_2$, then $\SomeK_1$ is more informative than $\SomeK_2$ (see~\cite{GT17} for the technical details justifying this argument). We introduce several provenance semirings below and explain the homomorphisms that link the most informative semiring ($\ProvPoly$)
to less informative semirings. 

\myunderpar{$(\ProvPoly,+,\cdot,0,1)$}: The elements of semiring $\ProvPoly$ are polynomials  with natural number coefficients and exponents over a set of variables $X$ representing tuples. Any polynomial can be written as a sum of products by applying the equational laws of semirings, e.g.,  the provenance polynomial for query result $\rel{Q_{3hop}}(s,s)$ is $p^3+2pqr$ (Fig.\,\ref{fig:partial-exam}). An important property of $\ProvPoly$ is that there exist homomorphisms from $\ProvPoly$ to any other semiring.

  \myunderpar{$(\PosBool,+, \cdot, 0, 1)$}: The elements of $\PosBool$ are derived from $\ProvPoly$ by making both addition and multiplication idempotent and applying an additional equational law: $x + x \cdot y = x$.   An element from $\PosBool$ can be encoded as a set of sets of variables with the restriction that every inner set $k$ is minimal, i.e., there is no other inner set $k'$ that is a subset of $k$. For example, the provenance polynomial $p^3 + 2pqr$ of  $\rel{Q_{3hop}}(s,s)$ is simplified as follows: $p^3 + 2pqr = p + pqr = p$. 

  \myunderpar{$(\WhichProv,+,\cdot,0,1)$}: In the $\WhichProv$ semiring, addition is equivalent to multiplication: $x + y = x \cdot y$ for $x,y \not \in \{0,1\}$,   and both addition and multiplication are idempotent. This semiring has sometimes also be called the Lineage semiring. Alternatively, the semiring can be defined over the powerset of the set of variables $X$~\cite{cheney2009provenance}. 

  Other semirings of interest are $(\BoolProvPoly,+,\cdot,0,1)$ which is derived from $\ProvPoly$ by making addition idempotent ($x + x \equiv x$), semiring $(\TrioProv,+,\cdot,0,1)$ where multiplication is idempotent ($x \cdot x \equiv x$), and $(\WhyProv, + , \cdot, 0, 1)$ where both addition and multiplication are idempotent.

\subsection{$\SomeK$-explanations}
\label{sec:expl-types}

We now introduce simplified versions of our provenance graphs that each corresponds to a certain provenance semiring. Given a positive query $P$, \provQ~$\aProvQ$, and database $I$, we use $\explainq_{\SomeK}(P,$ $\aProvQ,I)$ to denote a $\SomeK$-explanation for $\aProvQ$. A $\SomeK$-explanation is a provenance graph that encodes the $\SomeK$-provenance of all query results from $\qMatch(\aProvQ)$, i.e., the set of answers the user
is interested in.
In the following, we first show how to extract $\ProvPoly$ from our provenance graph.
Then, for each homomorphism implementing the derivation of a less informative provenance model from a more informative provenance model in the semiring framework,
there is a corresponding graph transformation over our provenance graphs that maps $\explainq_{\ProvPoly}(P, \aProvQ, I)$ to $\explainq_{\SomeK}(P, \aProvQ, I)$. The following theorem shows that we can reuse the existing mapping from provenance games to provenance polynomials by composing it with the mapping $\provToGame$.

\begin{Theorem}\label{theo:expl-NX}
Let $\provToNX$ denote the function\\ $\gameToNX \circ \provToGame$.
Given a positive input program $P$, database instance $I$, and tuple $t \in P(I)$, denote by $\ProvPoly(P,I,t)$ the $\ProvPoly$ annotation of $t$ over an abstractly tagged version of $I$ (each tuple $t$ is annotated with a unique variable $x_t$).
Then, \\
\begin{minipage}{1\linewidth}
$$\provToNX(\explainq(P,I,t)) = \ProvPoly(P,I,t)$$
\end{minipage}
\end{Theorem}
\begin{proof}
\ifnottechreport{The proof is shown in~\cite{LL18}.}
\iftechreport{
It has already proven in~\cite{KL13} that provenance games generalize $\ProvPoly$ (provenance polynomial can be extracted from games).
That is, using our notation, $$\gameToNX(\gprov(P,I,t)) = \ProvPoly(P,I,t)$$
To prove Theorem\,\ref{theo:expl-NX}, we have to show that the game explanation $\gprov(P,I,t)$ is equivalent to our explanation $\explainq(P,I,t)$.
In Theorem\,\ref{theo:equi-pgame},
we have proven that our provenance graphs are equivalent to provenance games (there exists a lossless transformation). By applying the transformation $\provToGame$ to the formula above,
we get $$\gameToNX(\provToGame(\explainq(P,I,t))) = \ProvPoly(P,I,t)$$
Thus, it is immediate that the function $\gameToNX \circ \provToGame$ is correct and, thus, the claim holds.
}
\end{proof}
Consider the explanation for $\whyq \rel{Q_{3hop}}(s,s)$ shown in Fig.\,\ref{fig:pg-3hop}. Recall that there are three options for reaching Seattle from Seattle with two intermediate stops corresponding to three derivations of $\rel{Q_{3hop}}$ using rule $r_2$.
These three derivations are shown in the provenance graph, e.g., $r_2(s,s,s,s)$ is the derivation that uses the local train connection inside Seattle three times. Annotating the train connections with variables $p$, $q$, and $r$ as shown in Fig.\,\ref{fig:partial-exam} and ignoring rule information encoded in the graph, the provenance encoded by our model is a bag (denoted as $[$ $]$) of lists (denoted as $($ $)$) of these variables. Each list corresponds to a rule derivation where variables are ordered according to the order of their occurrence in the body of the rule. For instance, $(q,r,p)$ corresponds to taking a train from Seattle to Chicago ($q$), then from Chicago to Seattle ($r$), and finally a local connection inside Seattle $(p)$.
We now illustrate the graph transformations yielding $\SomeK$-explanations from $\explainq_{\ProvPoly}$ based on this example.

\mypara{Semiring $\ProvPoly$}
In Fig.\,\ref{fig:pg-3hop}, we (1) replace rule nodes with multiplication (i.e., $r_2(s,s,s,s) \rightarrow \cdot$) and (2) replace goal nodes with addition (e.g., $g_4^1(s,s) \rightarrow +$) to generate a graph that encodes $\ProvPoly$ as shown in Fig.\,\ref{fig:poly-model} (denoted as $\explainq_{\ProvPoly}$).
Applying this transformation, the rule instantiation $r_2(s,s,c,s)$ deriving result tuple $\rel{Q_{3hop}}(s,s)$ can no longer be distinguished from $r_2(s,s,s,c)$, because they are connected to the same tuple nodes.
The only information retained is which arguments are used how often by a rule (labelled with $\cdot$). To extract $\ProvPoly$, we (1)  replace labels of leaf nodes with their annotations from Fig.\,\ref{fig:partial-exam} (e.g., $T(s,s)$ is replaced with $p$) and (2) replace IDB tuple nodes with addition.

\mypara{Semiring $\PosBool$}
$\explainq_{\PosBool}$ (Fig.\,\ref{fig:posb-model}) is computed from $\explainq_{\ProvPoly}$ by first collapsing rule nodes if the subgraphs rooted at these rule nodes are isomorphic and dropping all the goal nodes. Then, ``$\cdot$'' nodes are removed
if one or more subgraphs rooted at children of such a node is isomorphic to the subgraphs rooted at the children of another ``$\cdot$'' node.
Applying this process to our example, after the first step, two ``$\cdot$'' nodes (one connects $\rel{Q_{3hop}}(s,s)$ to $p$ and the other for each $p$, $q$, and $r$) exist in the graph corresponding to $p$ and $p \cdot q \cdot r$. In the second step, $p \cdot q \cdot r$ is removed because it contains $p$ ($\rel{T}(s,s)$) as a subgraph.

\mypara{Semiring $\WhichProv$}
The semiring $\WhichProv$ (aka Lineage) is reached by collapsing all intermediate nodes and directly connecting  tuple nodes (e.g., $\rel{Q_{3hop}}(s,s)$) with other tuple nodes (e.g., $\rel{T}(s,s)$)  as shown in Fig.\,\ref{fig:lin-model}.

$\explainq_{\BoolProvPoly}$ and $\explainq_{\TrioProv}$ are derived from $\explainq_{\ProvPoly}$ by collapsing isomorphic subgraphs rooted at rule nodes and by dropping all the goal nodes, respectively. The combination of these transformations achieves $\explainq_{\WhyProv}$.

  \section{Semiring Provenance for FO Model Checking}\label{sec:semi-foq}

The semiring framework was recently extended for capturing provenance of first-order (FO) model checking~\cite{T17,GV17}. We now study the relationship of our model to semiring provenance for FO queries.
Based on the observation first stated in~\cite{KL13} (provenance for FO queries and, thus,  also FO logic, naturally supports missing answers), the authors  explain missing answers based on FO provenance~\cite{XZ18}. Another interesting aspect of~\cite{GV17} is that it allows some facts to be left undetermined (their truth is undecided). This enables how-to queries~\cite{meliou2012tiresias}, i.e., given an expected outcome, which possible world compatible with the undecided facts would produce this outcome.
In this section, we first introduce the model from~\cite{GV17}, then demonstrate how our approach can be extended to support undetermined truth values. Finally, we show how the annotation computed by the approach presented in~\cite{GV17} for a FO formula $\aForm$ can be efficiently extracted from the provenance graph generated by our approach for a query $\formQ{}$ which is derived from $\aForm$ through a translation $\formToQ$.

\subsection{K-Interpretations and Dual Polynomials}\label{sec:k-interpr-dual}

In~\cite{T17,GT17}, the authors define semiring provenance for formulas in FO logic. Let $\fDom$ be a domain of values.
We use $\fVal$ to denote an assignment of the free variables of $\aForm$ to values from $\fDom$.
Given a so-called $\SomeK$-interpretation $\kInter$ which is a function mapping positive and negative literals to annotations from $\SomeK$, the annotation of a formula $\kInterOf{\aForm}$ for a given valuation $\fVal$ is derived using the rules below. For sentences, i.e., formulas without free variables, we omit the valuation and write $\kInter(\aForm)$ to denote $\kInterOf{\aForm}$ for the empty valuation $\fVal$. Furthermore, $\formOp$ is used to denote a comparison operator (either $=$ or $\neq$).

\vspace{-2mm}
\noindent
\resizebox{1\linewidth}{!}{
  \begin{minipage}{1.05\linewidth}
  \begin{align*}
  \kInterOf{R(\varVec)} &= \kInter(R(\fVal(\varVec)))
  &\kInterOf{\dlNeg R(\varVec)} &= \kInter(\dlNeg R(\fVal(\varVec)))
  \end{align*}\\[-12mm]
    \begin{align*}
    \kInterOf{x \formOp y} &= \mathtext{if} \fVal(x) \formOp \fVal(y) \mathtext{then} 1 \mathtext{else} 0
  &\kInterOf{\neg \aForm} &= \kInterOf{\nnf(\aForm)}
    \end{align*}\\[-12mm]
    \begin{align*}
      \kInterOf{\aForm_1 \vee \aForm_2} &= \kInterOf{\aForm_1} + \kInterOf{\aForm_2}
      &\kInterOf{\aForm_1 \wedge \aForm_2} &= \kInterOf{\aForm_1} \cdot \kInterOf{\aForm_2}
    \end{align*}\\[-12mm]
    \begin{align*}
  \kInterOf{\exists x\, \aForm} &= \sum_{a \in \fDom}\kInterOfVal{\aForm}{\fVal[x \mapsto a]}
  &\kInterOf{\forall x\, \aForm} &= \prod_{a \in \fDom}\kInterOfVal{\aForm}{\fVal[x \mapsto a]}
\end{align*}
\end{minipage}
}

Both conjunction and universal quantification correspond to multiplication,  and annotations of positive and negative literals are read from $\kInter$.
This model deals with negation as follows.  A negated formula is first translated into negation normal form (\textit{nnf}). A formula in $\nnf$ does not contain negation except for negated literals. Any formula can be translated into this form using DeMorgan rules, e.g., $\dlNeg (\forall x\, \aForm) \equiv (\exists x\, \dlNeg \aForm)$.
By pushing negation to the literal level using $\nnf$, and annotating both positive and negative literals, the approach avoids extending the semiring structure with
an explicit negation operation. 

Provenance tracking for FO formula has to take into account the dual nature of the literals. The solution presented in~\cite{GV17,T17} is to use polynomials over two sets of variables: variables from $X$ and $\dualOf{X}$ are exclusively used to annotate positive and negative literals, respectively. For any variable $x \in X$, there exists a corresponding variable $\dualOf{x} \in \dualOf{X}$ and vice versa. Furthermore, if $x$ annotates $R(\valVec)$, then $\dualOf{x}$ can only annotate $\dlNeg R(\valVec)$ (and vice versa). The semiring of dual indeterminate polynomials is then defined as the structure generated by applying the congruence $x \cdot \dualOf{x} = 0$ to the polynomials from $\mathbb{N}[X \cup \dualOf{X}]$. The resulting structure is denoted by $\ProvPolyDual$. Intuitively, this congruence encodes the standard logic equivalence $R(\valVec) \wedge \dlNeg R(\valVec) \equiv false$. Importantly, in a $\ProvPolyDual$-interpretation $\kInter$, we can decide which facts are true/false and whether to track provenance for these facts. Furthermore, we can leave the truth of some literals undetermined. Below, we show all feasible combinations for annotating $R(\valVec)$ and $\dlNeg R(\valVec)$ in $\kInter$ and their meaning. For instance, if we annotate $R(\valVec)$ with $1$ or $0$, this corresponds to asserting the fact $R(\valVec)$, but not tracking provenance for it. By setting $R(\valVec) = x$ and $\dlNeg R(\valVec) = \dualOf{x}$, we leave the truth of $R(\valVec)$ undecided.
Note that $R(\valVec) = 0$ and $\dlNeg R(\valVec) = 0$ (as well as $R(\valVec) = 1$ and $\dlNeg R(\valVec) = 1$) are not considered here since they lead to incompleteness (inconsistency).
\vspace{-2mm}
\begin{align*}
 \kInter( R(\valVec)) &= 1 &\kInter(\dlNeg R(\valVec)) &= 0 \tag{true, no provenance}\\
 \kInter( R(\valVec)) &= 0 &\kInter(\dlNeg R(\valVec)) &= 1 \tag{false, no provenance}\\
 \kInter( R(\valVec)) &= x &\kInter(\dlNeg R(\valVec)) &= 0\tag{true, track provenance}\\
 \kInter( R(\valVec)) &= 0 &\kInter(\dlNeg R(\valVec)) &= \dualOf{x} \tag{false, track provenance}\\
 \kInter( R(\valVec)) &= x &\kInter(\dlNeg R(\valVec)) &= \dualOf{x} \tag{undetermined}
\end{align*}

Consider a sentence $\aForm$.\footnote{We only restrict the discussion to sentences for simplicity. The arguments here also hold for formulas with free variables.}
The annotation $\kInter(\aForm)$ computed for $\aForm$ over $\kInter$ with undetermined facts represents a set of possible models for $\aForm$.
By choosing for each undetermined fact $R(\valVec)$ in $\kInter(\aForm)$ whether it is true or not, we ``instantiate'' one possible model for $\aForm$. By encoding a set of possible models, $\kInter(\aForm)$  allows for reverse reasoning: we can find models that fulfill certain properties from the set of models encoded by $\kInter(\aForm)$.

\begin{Example}\label{ex:dual-poly}
Reconsider query $r_1$ from Fig.\,\ref{fig:running-example-db}.
Assume that we want to determine what effect building a direct train connection from New York to Seattle would have on the query result $\rel{Q}(n,s)$. Thus, we make the assumption that the database instance is as in Fig.\,\ref{fig:running-example-db} with the exception that we keep $\rel{T}(n,s)$ undetermined.
In first-order logic, $r_1$ is expressed as: $\rel{only2hop}(x,y) \equiv \exists z(\rel{T}(x,z) \wedge \rel{T}(z,y) \wedge \neg \rel{T}(x,y))$
and fact $\rel{Q}(n,s)$ as: $\aForm \equiv \rel{only2hop}(n,s)$.
The database when keeping $\rel{T}(n,s)$ undetermined is encoded as a $\ProvPolyDual$-interpretations $\kInter$ which assigns variables to positive literals as shown in Fig.\,\ref{fig:running-example-db} (the corresponding negated literals are annotated with $0$). $\kInter(T(n,s)) = v$, $\kInter(\dlNeg T(n,s)) = \dualOf{v}$, and we annotate all remaining positive literals with $0$ and negative literals with $1$.
Computing $\kInter(\aForm)$ using the rules above, we get
$(t \cdot s \cdot \dualOf{v}) + (u \cdot r \cdot \dualOf{v})$.
There are two ways of deriving the query result $\rel{Q}(n,s)$ which both depend on the absence of a direct train connection from New York to Seattle ($\dualOf{v}$). Now if we decide to introduce such a connection, we can evaluate the effect of this choice by setting $\dualOf{v} = 0$ in the provenance polynomial above (the absence of this connection has been refuted), i.e., we get $(t \cdot s \cdot 0) + (u \cdot r \cdot 0) = 0$. Thus, if we were to introduce such a connection, then $\rel{Q}(n,s)$ would no longer be a result.

\end{Example}

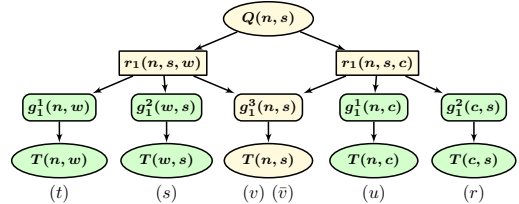
\begin{figure}[t]
  \centering
  \resizebox{0.8\columnwidth}{!}{\begin{tikzpicture}[>=latex',line join=bevel,line width=0.3mm]
  \definecolor{uncolor}{rgb}{1.0,0.98,0.865};
  \definecolor{fillcolor}{rgb}{0.83,1.0,0.8};

  \node (REL_Q_WON_n_s_) at (150bp,175bp) [draw=black,fill=uncolor,ellipse] {$\boldsymbol{Q(n,s)}$};
  \node (RULE_0_LOS_n_s_w_) at (91bp,150bp) [draw=black,fill=uncolor,rectangle] {$\boldsymbol{r_1(n,s,w)}$};

  \node (GOAL_0_0_WON_n_w_) at (33bp,125bp) [draw=black,fill=fillcolor,rounded corners=.15cm,inner sep=3pt] {$\boldsymbol{g_{1}^{1}(n,w)}$};
  \node (EDB_T_LOS_n_w_) at (33bp,95bp) [draw=black,fill=fillcolor,ellipse,label=below:{\normalsize $(t)$}] {$\boldsymbol{T(n,w)}$};

  \node (GOAL_0_1_WON_w_s_) at (93bp,125bp) [draw=black,fill=fillcolor,rounded corners=.15cm,inner sep=3pt] {$\boldsymbol{g_{1}^{2}(w,s)}$};
  \node (EDB_T_LOS_w_s_) at (93bp,95bp) [draw=black,fill=fillcolor,ellipse,label=below:{\normalsize $(s)$}] {$\boldsymbol{T(w,s)}$};

  \node (GOAL_0_2_WON_n_s_) at (150bp,125bp) [draw=black,fill=uncolor,rounded corners=.15cm,inner sep=3pt] {$\boldsymbol{g_{1}^{3}(n,s)}$};
  \node (EDB_T_LOS_n_s_) at (150bp,95bp) [draw=black,fill=uncolor,ellipse,label=below:{\normalsize $(v)$ $(\bar{v})$}] {$\boldsymbol{T(n,s)}$};

  \node (RULE_0_LOS_n_s_c_) at (211bp,150bp) [draw=black,fill=uncolor,rectangle] {$\boldsymbol{r_1(n,s,c)}$};
  \node (GOAL_0_0_WON_n_c_) at (209bp,125bp) [draw=black,fill=fillcolor,rounded corners=.15cm,inner sep=3pt] {$\boldsymbol{g_{1}^{1}(n,c)}$};
  \node (EDB_T_LOS_n_c_) at (209bp,95bp) [draw=black,fill=fillcolor,ellipse,label=below:{\normalsize $(u)$}] {$\boldsymbol{T(n,c)}$};

  \node (GOAL_0_1_WON_c_s_) at (265bp,125bp) [draw=black,fill=fillcolor,rounded corners=.15cm,inner sep=3pt] {$\boldsymbol{g_{1}^{2}(c,s)}$};
  \node (EDB_T_LOS_c_s_) at (265bp,95bp) [draw=black,fill=fillcolor,ellipse,label=below:{\normalsize $(r)$}] {$\boldsymbol{T(c,s)}$};

  \draw [->] (RULE_0_LOS_n_s_c_) -> (GOAL_0_0_WON_n_c_);
  \draw [->] (REL_Q_WON_n_s_) -> (RULE_0_LOS_n_s_c_);
  \draw [->] (RULE_0_LOS_n_s_c_) -> (GOAL_0_1_WON_c_s_);

  \draw [->] (GOAL_0_0_WON_n_c_) -> (EDB_T_LOS_n_c_);

  \draw [->] (REL_Q_WON_n_s_) -> (RULE_0_LOS_n_s_w_);
  \draw [->] (RULE_0_LOS_n_s_w_) -> (GOAL_0_1_WON_w_s_);
  \draw [->] (RULE_0_LOS_n_s_c_) -> (GOAL_0_2_WON_n_s_);
  \draw [->] (RULE_0_LOS_n_s_w_) -> (GOAL_0_2_WON_n_s_);
  \draw [->] (GOAL_0_2_WON_n_s_) -> (EDB_T_LOS_n_s_);

  \draw [->] (GOAL_0_1_WON_c_s_) -> (EDB_T_LOS_c_s_);

  \draw [->] (GOAL_0_0_WON_n_w_) -> (EDB_T_LOS_n_w_);

  \draw [->] (GOAL_0_1_WON_w_s_) -> (EDB_T_LOS_w_s_);

  \draw [->] (RULE_0_LOS_n_s_w_) -> (GOAL_0_0_WON_n_w_);
\end{tikzpicture}
 }
  $\,$\\[-4mm]
  \caption{Provenance graph for $\whyq \rel{Q}(n,s)$ when $T(n,s)$ is left undetermined}\label{fig:exam-whynot-instantiate}

\end{figure}

\begin{figure*}[t]
  \centering $\,$\\[-4mm]
\resizebox{1\linewidth}{!}{
  \begin{minipage}{1.03\linewidth}
  \begin{minipage}{0.36\linewidth}
  \begin{equation}
    \label{eq:translate-exists}
    \frac{\aForm \defas \exists x: \aForm_1}
    {
       \formQ{}(\freeOf{\aForm}) \dlImp \domQ(x), \formQ{1}(\freeOf{\aForm_1})
  }
\end{equation}
\end{minipage}
\hspace{0.1cm}
  \begin{minipage}{0.4\linewidth}
  \begin{equation}
    \label{eq:translate-R}
    \frac{\aForm \defas \dlNeg R(\varVec), \freeOf{\aForm} = \{x_1, \ldots, x_n\} }
    {
       \formQ{}(\freeOf{\aForm}) \dlImp \domQ(x_1), \ldots, \domQ(x_n), \dlNeg R(\varVec)
  }
\end{equation}
\end{minipage}
\hspace{0.3cm}
  \begin{minipage}{0.2\linewidth}
  \begin{equation}
    \label{eq:translate-R}
    \frac{\aForm \defas R(\varVec)}
    {
       \formQ{}(\freeOf{\aForm}) \dlImp R(\varVec)
  }
\end{equation}
\end{minipage}

  \begin{minipage}{0.68\linewidth}
  \begin{equation}
    \label{eq:translate-or}
    \frac{\aForm \defas \aForm_1 \vee \aForm_2, \freeOf{\aForm_1} = \{x_1, \ldots, x_n, y_1, \ldots, y_m\}, \freeOf{\aForm_2} = \{x_1, \ldots, x_n, z_1, \ldots, z_l\}}
    {
      \splitfrac{
         \formQ{}(\freeOf{\aForm}) \dlImp \domQ(z_1), \ldots, \domQ(z_k), \formQ{1}(\freeOf{\aForm_1})}
      {\formQ{}(\freeOf{\aForm}) \dlImp \domQ(y_1), \ldots, \domQ(y_m), \formQ{2}(\freeOf{\aForm_2})}
  }
  \end{equation}
\end{minipage}
\hfill
  \begin{minipage}{0.31\linewidth}
  \begin{equation}
    \label{eq:translate-R}
    \frac{\aForm \defas x \formOp y}
    {
       \formQ{}(x,y) \dlImp \domQ(x), \domQ(y), x \formOp y
  }
\end{equation}
\end{minipage}

\begin{minipage}{0.56\linewidth}
  \begin{equation}
    \label{eq:translate-forall}
    \frac{\aForm \defas \forall x: \aForm_1, \freeOf{\aForm} = \{x_1, \ldots, x_n\} }
    {
      \splitfrac{
         \formQ{}(\freeOf{\aForm}) \dlImp \domQ(x_1), \ldots, \domQ(x_k), \dlNeg \formQ{{}'}(\freeOf{\aForm})}
      {\formQ{{}'}(\freeOf{\aForm}) \dlImp \domQ(x), \domQ(x_1), \ldots, \domQ(x_n), \neg \formQ{1}(\freeOf{\aForm_1})}
  }
  \end{equation}
\end{minipage}
\hfill
  \begin{minipage}{0.4\linewidth}
  \begin{equation}
    \label{eq:translate-and}
    \frac{\aForm \defas \aForm_1 \wedge \aForm_2}{
     \formQ{}(\freeOf{\aForm}) \dlImp \formQ{1}(\freeOf{\aForm_1}), \formQ{2}(\freeOf{\aForm_2})
    }
  \end{equation}
\end{minipage}

\end{minipage}
}

  $\,$\\[-3mm]
  \caption{Translating a first-order formula $\aForm$ into a first-order query $\formQ{}$}\label{fig:translation-formula-to-query}
\end{figure*}

\subsection{Supporting Undeterminism in  Provenance Graphs}\label{sec:supp-undet-liter}

Supporting undetermined facts in our provenance model is surprisingly straightforward. We introduce a new label $\yellowD$ which is used to label nodes whose success/failure (existence/absence) is undetermined. To account for this new label, we amend the rules for determining connectivity and node labeling as follows:

\begin{itemize}
\item For a goal node $v_g$ (no matter whether positive or negative) that is connected to a tuple node $v_t$ with $\successLabel(v_t) = \yellowD$, we set $\successLabel(v_g) = \yellowD$ (goals corresponding to undetermined tuples are undetermined).
\item A rule node is succesful ($\greenT$) if all its goals are successful, a rule node is failed if some of its goals are failed ($\redF$), and finally a rule node is undetermined ($\yellowD$) if at least one of its goals is undetermined and none of its goals are failed. Successful rule nodes are connected to all goals, failed rule nodes to failed and undetermined goals (these may provide further justification for the failure), and undetermined rule nodes to all goals (successful and undetermined goals will determine success of the rule nodes under choices).
\item An IDB tuple exists ($\greenT$) if at least one of its rule derivation is successful. It is connected to all successful and undetermined rule derivations (these may provide additional justifications under certain choices for undetermined facts). An IDB tuple is absent $(\redF$) if all of its rule derivations fail. Finally, an IDB tuple is undetermined ($\yellowD$) if at least one of its rule derivations is undetermined and none is successful. Undetermined tuple nodes are connected to all their rule derivations (failed ones may be additional justifications for absence while undetermined ones may justify either existence or absence).
\end{itemize}

\begin{Example}\label{ex:dual-poly-PG}
Consider Example\,\ref{ex:dual-poly} in our extended provenance graph model. Let $v_{n,s}$ be the node corresponding to $\rel{T}(n,s)$. If we set $\successLabel(v_{n,s}) = \yellowD$ to indicate that $\rel{T}(n,s)$ should be considered as undetermined, then we get the provenance graph in Fig.\,\ref{fig:exam-whynot-instantiate}. Our approach correctly determines that under this assumption the truth of $\rel{Q}(n,s)$ is undetermined and that there are two potential derivations of this result which also are undetermined, because they depend on existing tuples as well as the undetermined tuple $\rel{T}(n,s)$. To evaluate the effect of choosing $\rel{T}(n,s)$ to be true or false, we would set $\successLabel(v_{n,s}) = \greenT$ or $\successLabel(v_{n,s}) = \redF$ and propagate the effect of this change bottom-up through-out the provenance graph.
\end{Example}

Importantly, the provenance graph captured for an instance with undetermined facts is sufficient for evaluating the effect of setting any of these undetermined facts to false or true. That is, just like it is not necessary to reevaluate the semiring annotation of a formula to evaluate the impact of such a choice,  it is also not necessary to recapture provenance in our model to evaluate a choice.
For lack of space, we are not discussing the details of a corresponding extension for provenance games, but still would like to remark that undetermined facts correspond to draws in the game (neither player has a winning strategy). In the type of two-player games employed in provenance games, draws are caused by cycles in the game graph. To leave the existence of an EDB tuple undetermined, we introduce an EDB fact node for the tuple and add a self-edge to this node which causes the tuple node to be a draw in the game.

\begin{figure*}[t]
  \centering $\,$\\[-2mm]
  \resizebox{0.93\linewidth}{!}{\begin{tikzpicture}[>=latex',line join=bevel,line width=0.3mm,scale=0.4]
\definecolor{uncolor}{rgb}{1.0,0.98,0.865} \definecolor{failcolor}{rgb}{0.63,0,0}
\definecolor{succcolor}{rgb}{0.83,1.0,0.8}

  \node (RULE_3_LOST_b_b_) at (843.37bp,18.0bp) [draw=black,fill=uncolor,rectangle,label=below:{\small $(\cdot)$}] {$\boldsymbol{r_3(b,a)}$};
  \node (TUPLE_R_LOST_a_b_) at (1537.3bp,77.0bp) [draw=black,fill=succcolor,ellipse] {$\boldsymbol{R(a,a)}$};

  \node (GOAL_3_2_LOST_b_) at (927.87bp,-65.0bp) [draw=black,rectangle,rounded corners=.15cm,fill=succcolor,label=below:{\small $(+)$}] {$\boldsymbol{g_{3}^{2}(b)}$};
  \node (GOAL_3_1_LOST_b_b_) at (977.87bp,18.0bp) [draw=black,rectangle,rounded corners=.15cm,fill=uncolor,label=below:{\small $(+)$}] {$\boldsymbol{g_{3}^{1}(b,a)}$};
  \node (TUPLE_QPHITWO_LOST_b_b_) at (1133.1bp,18.0bp) [draw=black,fill=uncolor,ellipse,label=below:{\small $(+)$}] {$\boldsymbol{Q_{\varphi_2}(b,a)}$};
\node (DOM_b_) at (1073.1bp,-65.0bp) [draw=black,fill=succcolor,ellipse,label=below:{\small $(1)$}] {$\boldsymbol{\domQ(b)}$};
  
  \node (TUPLE_QPHIONE_LOST_b_) at (694.83bp,18.0bp) [draw=black,fill=uncolor,ellipse,label=below:{\small $(+)$}] {$\boldsymbol{Q_{\varphi_1}(b)}$};
  \node (RULE_4_LOST_a_b_) at (1283.8bp,77.0bp) [draw=black,fill=succcolor,rectangle,label=above:{\small $(\cdot)$}] {$\boldsymbol{r_4(a,a)}$};
  \node (GOAL_4_0_LOST_a_b_) at (1403.3bp,77.0bp) [draw=black,rectangle,rounded corners=.15cm,fill=succcolor,label=above:{\small $(x)$}] {$\boldsymbol{g_{4}^{1}(a,a)}$};
  \node (RULE_1_WON_1_) at (139.69bp,45.0bp) [draw=black,fill=uncolor,rectangle,label=above:{\small $(\cdot)$}] {$\boldsymbol{r_1()}$};
  \node (TUPLE_QPHITWO_LOST_a_b_) at (1133.1bp,77.0bp) [draw=black,fill=succcolor,ellipse,label=above:{\small $(+)$}] {$\boldsymbol{Q_{\varphi_2}(a,a)}$};
  \node (DOM_a_) at (1073.1bp,160.0bp) [draw=black,fill=succcolor,ellipse,label=above:{\small $(1)$}] {$\boldsymbol{\domQ(a)}$};
  
  \node (GOAL_1_1_WON_1_) at (233.69bp,45.0bp) [draw=black,rectangle,rounded corners=.15cm,fill=uncolor,label=above:{\small $(+)$}] {$\boldsymbol{g_{1}^{1}()}$};
  \node (TUPLE_R_LOST_b_b_) at (1537.3bp,18.0bp) [draw=black,fill=uncolor,ellipse] {$\boldsymbol{R(b,a)}$};
  \node (TUPLE_QPHI_WON_1_) at (38.347bp,45.0bp) [draw=black,fill=uncolor,ellipse,label=above:{\small $(+)$}] {$\boldsymbol{Q_{\varphi}()}$};

  \node (GOAL_3_2_LOST_a_) at (927.87bp,160.0bp) [draw=black,rectangle,rounded corners=.15cm,fill=succcolor,label=above:{\small $(+)$}] {$\boldsymbol{g_{3}^{2}(a)}$};
  \node (GOAL_3_1_LOST_a_b_) at (977.87bp,77.0bp) [draw=black,rectangle,rounded corners=.15cm,fill=succcolor,label=above:{\small $(+)$}] {$\boldsymbol{g_{3}^{1}(a,a)}$};
  \node (RULE_3_LOST_a_b_) at (843.37bp,77.0bp) [draw=black,fill=succcolor,rectangle,label=above:{\small $(\cdot)$}] {$\boldsymbol{r_3(a,a)}$};
  \node (RULE_4_LOST_b_b_) at (1283.8bp,18.0bp) [draw=black,fill=uncolor,rectangle,label=below:{\small $(\cdot)$}] {$\boldsymbol{r_4(b,a)}$};
  \node (TUPLE_QPHIONE_LOST_a_) at (694.83bp,77.0bp) [draw=black,fill=succcolor,ellipse,label=above:{\small $(+)$}] {$\boldsymbol{Q_{\varphi_1}(a)}$};
  \node (GOAL_4_0_LOST_b_b_) at (1403.3bp,18.0bp) [draw=black,rectangle,rounded corners=.15cm,fill=uncolor,label=below:{\small $(y)$}] {$\boldsymbol{g_{4}^{1}(b,a)}$};
  
  \node (RULE_2_LOST_1_b_) at (463.78bp,18.0bp) [draw=black,fill=uncolor,rectangle,label=below:{\small $(\cdot)$}] {$\boldsymbol{r_2(b)}$};
  \node (GOAL_2_2_LOST_a_) at (561.28bp,77.0bp) [draw=black,rectangle,rounded corners=.15cm,fill=failcolor,label=above:{\small $(+)$}, text=white] {$\boldsymbol{g_{2}^{2}(a)}$};
  \node (RULE_2_LOST_1_a_) at (463.78bp,77.0bp) [draw=black,fill=failcolor,rectangle,label=above:{\small $(\cdot)$}, text=white] {$\boldsymbol{r_2(a)}$};
  \node (GOAL_2_2_LOST_b_) at (561.28bp,18.0bp) [draw=black,rectangle,rounded corners=.15cm,fill=uncolor,label=below:{\small $(+)$}] {$\boldsymbol{g_{2}^{2}(b)}$};
  \node (TUPLE_QPHIPR_LOST_1_) at (350.74bp,45.0bp) [draw=black,fill=uncolor,ellipse,label=above:{\small $(\cdot)$}] {$\boldsymbol{Q_{\varphi'}()}$};

  \node (GOAL_2_1_LOST_b_) at (561.28bp,-65.0bp) [draw=black,rectangle,rounded corners=.15cm,fill=succcolor,label=below:{\small $(+)$}] {$\boldsymbol{g_{1}^{2}(b)}$};  

  \draw [->] (GOAL_3_1_LOST_a_b_) -- (TUPLE_QPHITWO_LOST_a_b_);
  \draw [->] (RULE_2_LOST_1_a_) -- (GOAL_2_2_LOST_a_);
  \draw [->] (RULE_4_LOST_a_b_) -- (GOAL_4_0_LOST_a_b_);
  \draw [->] (RULE_3_LOST_b_b_) -- (GOAL_3_1_LOST_b_b_);
\draw [->] (RULE_3_LOST_b_b_) -- (GOAL_3_2_LOST_b_);
  
  \draw [->] (GOAL_4_0_LOST_b_b_) -- (TUPLE_R_LOST_b_b_);
  \draw [->] (GOAL_2_2_LOST_a_) -- (TUPLE_QPHIONE_LOST_a_);
  \draw [->] (RULE_2_LOST_1_b_) -- (GOAL_2_2_LOST_b_);
  \draw [->] (TUPLE_QPHITWO_LOST_a_b_) -- (RULE_4_LOST_a_b_);
  \draw [->] (RULE_3_LOST_a_b_) -- (GOAL_3_1_LOST_a_b_);
\draw [->] (RULE_3_LOST_a_b_) -- (GOAL_3_2_LOST_a_);
  
  \draw [->] (TUPLE_QPHIPR_LOST_1_) -- (RULE_2_LOST_1_b_);
  \draw [->] (TUPLE_QPHITWO_LOST_b_b_) -- (RULE_4_LOST_b_b_);
  \draw [->] (TUPLE_QPHIPR_LOST_1_) -- (RULE_2_LOST_1_a_);
  \draw [->] (TUPLE_QPHIONE_LOST_b_) -- (RULE_3_LOST_b_b_);
  \draw [->] (GOAL_1_1_WON_1_) -- (TUPLE_QPHIPR_LOST_1_);
  \draw [->] (TUPLE_QPHI_WON_1_) -- (RULE_1_WON_1_);
  \draw [->] (GOAL_2_2_LOST_b_) -- (TUPLE_QPHIONE_LOST_b_);
  \draw [->] (RULE_1_WON_1_) -- (GOAL_1_1_WON_1_);
  \draw [->] (GOAL_3_1_LOST_b_b_) -- (TUPLE_QPHITWO_LOST_b_b_);
\draw [->] (GOAL_3_2_LOST_b_) -- (DOM_b_);
\draw [->] (GOAL_3_2_LOST_a_) -- (DOM_a_);

  \draw [->] (TUPLE_QPHIONE_LOST_a_) -- (RULE_3_LOST_a_b_);
  \draw [->] (RULE_4_LOST_b_b_) -- (GOAL_4_0_LOST_b_b_);
  \draw [->] (GOAL_4_0_LOST_a_b_) -- (TUPLE_R_LOST_a_b_);

  \draw [->] (RULE_2_LOST_1_b_) -- (GOAL_2_1_LOST_b_);
  \draw [->] (GOAL_2_1_LOST_b_) to [out=330, in=200] (DOM_b_);
\end{tikzpicture}

 }
  $\,$\\[-6mm]
  \caption{Provenance graph for query $\formQ{}$ based on $\aForm \defas \forall x\, \exists y\, \rel{R}(x,y)$ when $R(a,a)$ is true ($\kInter(\rel{R}(a,a)) = x$ and $\kInter(\dlNeg \rel{R}(a,a)) = 0$) and $R(b,a)$ is left undetermined ($\kInter(\rel{R}(b,a)) = y$ and $\kInter(\dlNeg \rel{R}(b,a)) = \dualOf{y}$). Note that $\aForm_1 \defas \exists y\, R(x,y)$ and $\aForm_2 \defas R(x,y)$. The result of $\provToFO$ shown besides the nodes encodes the dual polynominal $\kInter(\aForm) = x \cdot y$.}\label{fig:exam-form-pgraph}
\end{figure*}

\subsection{From First-order Formulas to Datalog}\label{sec:from-first-order}

We now present a translation $\formToQ$ from FO formulas $\aForm$ to boolean Datalog queries $\formQ{}$. The query generated based on a formula $\aForm$ is equivalent to the formula in the following sense: if $\aForm$ evaluates to true for a model, then $\formQ{}(I)$ returns true. Here, $I$ is the instance that contains precisely the tuples corresponding to literals that are true in the model.
We assume that the free variables of a formula (the variables not bound by any quantifier) are distinct from variable names bound by quantifiers and that no two quantifiers bind a variable of the same name.
This can be achieved by renaming variables in a formula that does not fulfill this condition. For example,
$(\forall x\, R (x,y)) \wedge (\exists y\, S(y))$ does not fulfill this condition, but the equivalent formula $(\forall x\, R(x,y)) \wedge (\exists z\, S(z))$ does.
We also assume an arbitrary, but fixed, total order $\varOrder$ over variables that appear in formulas.
We  use $\freeOf{\aForm}$ to denote the list of free variables of a formula $\aForm$ ordered increasingly by $\varOrder$.  For instance, for $\aForm \defas \forall x: R(x,y)$ we have $\freeOf{\aForm} = \{y\}$. Our translation $\formToQ$ takes as input a formula $\aForm$ and outputs a Datalog program with an answer predicate $\formQ{}$.
The translation rules are shown in Fig.\,\ref{fig:translation-formula-to-query}. Each rule translates one construct (e.g., a quantifier) and outputs one or more Datalog rules. The Datalog program generated by the translation for an input $\aForm$ is the set of Datalog rules generated by applying the rules from Fig.\,\ref{fig:translation-formula-to-query} to all sub-formulas of $\aForm$. Here, we assume the existence of a unary predicate $\domQ$ whose extension is the domain $\fDom$. Most translation rules are straightforward and standard. Logical operators are translated into their obvious counterpart in Datalog, e.g., a conjunction $\aForm_1 \wedge \aForm_2$ is translated into a rule with two body atoms $\formQ{1}$ and $\formQ{2}$. The rules generated for a formula $\aForm$ return the formula's free variables to make them available to formulas that use $\aForm$. For instance, since Datalog does not support universal quantification directly, we have to simulate it using double negation ($\forall x\, \aForm$ is rewritten as $\dlNeg \exists x\, \dlNeg \aForm$). Disjunctions are turned into unions. The complexity of the rule for disjunction stems from the fact that, in $\aForm_1 \vee \aForm_2$, the sets of free variables for $\aForm_1$ and $\aForm_2$ may not be the same. To make them union compatible, use $\freeOf{\aForm}$ as the arguments of the heads of the rules for both $\aForm_1$ and $\aForm_2$, and add additional goals $\domQ$ to ensure that these rules are safe.

\begin{Example}\label{ex:form-to-q}
  Consider a directed graph encoded as its edge relation $\rel{R}$. The formula $\aForm \defas \forall x\, \exists y\, \rel{R}(x,y)$
checks whether all nodes in the graph have outgoing edges.
  Let $\aForm_1 = \exists y\,\rel{R}(x,y)$ and $\aForm_2 = \rel{R}(x,y)$. Translating this formula,
we get:
\vspace{-2mm}
  \begin{align*}
    \formQ{}() &\dlImp \neg \formQ{{}'}() &\formQ{{}'}() &\dlImp \domQ(x), \neg \formQ{1}(x)\\[-9.5mm]
  \end{align*}
      \begin{align*}
    \formQ{1}(x) &\dlImp \domQ(y), \formQ{2}(x,y)
                                     &\formQ{2}(x,y) &\dlImp \rel{R}(x,y)
  \end{align*}
\end{Example}

\subsection{From Graphs to FO Semiring Provenance}\label{sec:from-first-order-extraction}

Given a formula $\aForm$ in negation normal form (nnf) and a $\ProvPolyDual$-interpretation $\kInter$, we now demonstrate how to extract $\kInterOf{\aForm}$ from the subgraph of the provenance graph generated based on $\kInter$ over $\formToQ(\aForm)$ rooted at the tuple node $\formQ{}(\fVal(\freeOf{\aForm}))$.
First, we apply $\formToQ(\aForm)$ to compute $\formQ{}$.
Then, we generate an instance $\Iinter$ where the existence of a tuple corresponding to a literal $\rel{R}(\valVec)$ is determined based on the truth value of this literal encoded by its annotation $\kInter(\rel{R}(\valVec))$. A tuple $\rel{R}(\valVec)$ exists in $\Iinter$ if $\kInter(\rel{R}(\valVec)) = x$ or $\kInter(\rel{R}(\valVec)) = 1$ and $\kInter(\dlNeg \rel{R}(\valVec)) = 0$, the tuple is missing if $\kInter(\dlNeg \rel{R}(\valVec)) = \dualOf{x}$ or $\kInter(\dlNeg \rel{R}(\valVec)) = 1$ and $\kInter(\rel{R}(\valVec)) = 0$, and the tuple's existence is undetermined if $\kInter(\rel{R}(\valVec)) = x$ and $\kInter(\dlNeg \rel{R}(\valVec)) = \dualOf{x}$. Note that this corresponds to the truth value according to the 5 cases we have discerned in Sec.\,\ref{sec:k-interpr-dual}. 

Next, we generate the provenance graph $\provGraph(\formQ{},\Iinter)$. If the formula has free variables, then the provenance graph will contain multiple tuple nodes $\formQ{}(\fVal(\freeOf{\aForm}))$, one for each valuation $\fVal$ of the free variables, and the subgraph rooted at one such tuple node encodes $\kInterOf{\aForm}$. By applying a function $\provToFO$ (defined in the following), we  translate the subgraph rooted at tuple $\formQ{}(\fVal(\freeOf{\aForm}))$ in $\provGraph(\formQ{},\Iinter)$ into $\kInterOf{\aForm}$.

The function $\provToFO$ replaces nodes in the provenance graph with nodes labeled as ``$+$'',  ``$\cdot$'', and annotations of literals. The polynomial $\kInterOf{\aForm}$ can then be read from the graph generated by $\provToFO$ through a top-down traversal. Intuitively, the translation can be explained as follows. Datalog rules are a conjunction of atoms and, thus, are replaced with multiplication. There may exist multiple ways that derive an IDB tuple through the rules of query. That is, IDB tuple nodes represent addition. The exception is IDB tuples that are used in a negated fashion which are replaced with multiplication, because, for the goal to succeed, all derivations of the tuple have to fail. Note that a tuple is used in a negated way if there is an odd number of negated goals on the path between the root of the provenance graph and IDB tuple node. In a program produced by our translation rules, this can only be the case for tuples that correspond to head predicate of a rule computing the $\dlNeg \exists$ part of the translation of a universal quantification.

$\provToFO$ consists of the following steps:
\begin{enumerate}
\item Replace tuple nodes $\domQ(\varVec)$ with $1$.
\item A goal node connected to an EDB tuple node representing a literal $\rel{R}(\valVec)$ is replaced by $\kInter(\rel{R}(\valVec))$ if the goal is positive and $\kInter(\dlNeg \rel{R}(\valVec))$ otherwise.
\item Next, all EDB tuple nodes are removed leaving the goal nodes formerly connected to EDB tuple nodes to be the new leaves of the graph.
\item Rule nodes are replaced with multiplication ($\cdot$).
\item Next all remaining goal nodes are replaced with addition ($+$).
\item Finally, nodes $v_t$ corresponding to IDB tuples are replaced with addition with the exception of IDB tuples corresponding to the head predicate ($\formQ{'}$) of the second rule  of a translated universal quantification  which are replaced with multiplication ($\cdot$).
\end{enumerate}

\begin{Example}
Consider the formula and query from Example~\ref{ex:form-to-q}. Assume that $\fDom = \{a,b\}$ and consider that interpretation $\kInter$ which tracks provenance for edge $(a,a)$, keeps $\rel{R}(b,a)$ undetermined, and sets all other positive literals to false without provenance tracking, i.e.,  $\kInter(\rel{R}(a,a)) \\= x$, $\kInter(\dlNeg \rel{R}(a,a)) = 0$,  $\kInter(\rel{R}(b,a)) = y$, $\kInter(\dlNeg \rel{R}(b,a)) = \dualOf{y}$,
and for all other $\rel{R}(\valVec)$ we have $\kInter(\rel{R}(\valVec)) = 0$ and $\kInter(\neg \rel{R}(\valVec)) = 1$. That is, in $\Iinter$, tuple $\rel{R}(a,a)$ exists, tuple $\rel{R}(b,a)$'s existence is undetermined, and all other tuples are missing. Fig.\,\ref{fig:exam-form-pgraph} shows the provenance graph $\provGraph(\formQ{}, \Iinter)$. The truth of the universal quantification in $\aForm$ is undetermined, because while there exists no $a \in \fDom$ such that $\neg \aForm_1$ for $\fVal \defas (x = a)$ is true (there is an outgoing edge starting at $a$), the truth of $\neg \aForm_1$ is undetermined for $\fVal \defas (x = b)$ (the existence of edge $\rel{R}(b,a)$ is undetermined). The truth of $\exists y\, R(a,y)$ and $\exists y\, \rel{R}(b,y)$ is justified by the existing tuples $\rel{R}(a,a)$ and $\rel{R}(b,a)$, respectively. Applying the translation $\provToFO$, we get graph with node labels shown in Fig.\,\ref{fig:exam-form-pgraph} which corresponds to the polynomial $1 \cdot x \cdot 1 \cdot 1 \cdot y = x \cdot y = \kInter(\aForm)$.
\end{Example}

We are now ready to state the main result of this section: our provenance graphs extended for undetermined facts can encode semiring provenance for first-order (FO) model checking. For simplicity, we only consider sentences, i.e., formulas $\aForm$ without free variables, but the result also holds for formulas with free variables by only translating a subgraph of the provenance rooted at the IDB tuple node $\formQ{}(\fVal(\freeOf{\aForm}))$ which corresponds to the formula $\fVal(\aForm)$.

\begin{Theorem}
  Let $\aForm$ be a formula, $\kInter$ a $\ProvPolyDual$\--in\-ter\-pre\-ta\-tion, $Q \defas \formToQ(\aForm)$,  and $\Iinter$ the instance corresponding to $\kInter$ as defined above. Then \vspace{-2mm}
  $$\provToFO(\provGraph(Q, \Iinter)) = \kInterOf{\aForm}$$
\end{Theorem}
\begin{proof}
\ifnottechreport{We prove the theorem by induction over the structure of the input formula $\aForm$ for a given valuation $\fVal$ of $\freeOf{\aForm}$.
For the full proof, see~\cite{LL18}.}
\iftechreport{
We prove the theorem by induction over the structure of the input formula $\aForm$ for a given valuation $\fVal$ of $\freeOf{\aForm}$. Recall that in $\Iinter$ the existence of a tuple is determined by the annotations assigned by $\kInter$ to the positive and negated literals corresponding to the tuple. In the following, we will often consider the three cases separately (tuple exists, tuples is missing, tuple's existence is undetermined), because our the graph produced by our provenance model may differ based on which case applies. Furthermore, it will be beneficial to also prove that tuple $\formQ{}(\fVal(\freeOf{\aForm}))$ exists/is absent/undetermined iff $\fVal(\aForm)$ is true/false/undetermined, because it allows us to assume that the truth value of a sub-formula corresponds to $\fVal(\aForm)$. In the following we will use $\fTup{\aForm}{\fVal}$ to denote $\formQ{}(\fVal(\freeOf{\aForm}))$, i.e., the IDB tuple corresponding to applying the valuation $\fVal$ to formula $\aForm$. Furthermore, for IDB tuple $\fTup{\aForm}{\fVal}$ use $\fG{\aForm}{\fVal}$ to denote the subgraph of the provenance graph rooted at $\fTup{\aForm}{\fVal}$ and $\fP{\aForm}{\fVal}$ to denote $\provToFO(\fG{\aForm}{\fVal})$.

\myproofpar{Base case: $\aForm \defas \rel{R}(\varVec)$}
Such formulas are translated into programs with a single rule. The head predicate $\formQ{}$ for this rule has as arguments the free variables of $\varVec$ which for a literal $\rel{R}(\varVec)$ are all variables of the literal. For any valuation $\fVal$, there is one corresponding tuple $\formQ{}(\fVal(\varVec))$ in the provenance graph which as we have mentioned above we denote as $\fTup{\aForm}{\fVal}$. First consider the case that the tuple $\fTup{\aForm}{\fVal}$ exists which is the case when literal $\rel{R}(\fVal(\varVec))$ is true. Applying $\provToFO$ to the subgraph corresponding to such a tuple, the tuple node $\formQ{}(\fVal(\varVec))$ is mapped to $+$, the rule node to $\cdot$, the goal node connected to the EDB node is replaced with $\kInter(\rel{R}(\fVal(\varVec)))$. The resulting expression tree is $+(\cdot(\kInter(\rel{R}(\fVal(\varVec)))))$ which evaluates to $\kInter(\rel{R}(\fVal(\varVec))) = \kInterOf{\aForm}$. If the existence of $\rel{R}(\fVal(\varVec))$ is undetermined, then the rule deriving the result tuple $\fTup{\aForm}{\fVal}$ are also undetermined. However, the resulting expression is still the same. If tuple  $\rel{R}(\fVal(\varVec))$ is missing we still get a graph with the same structure. 

\myproofpar{Base case: $\dlNeg \rel{R}(\varVec)$} A formula
$\dlNeg \rel{R}(\varVec)$ is translated into a single rule that ``joins'' the
atom $\dlNeg \rel{R}(\varVec)$ with multiple $\domQ$ atoms (one for each
variable in $\varVec$). The arguments of the head predicate of the rule are all
variables in $\varVec$. For a given valuation $\fVal$, there exists one IDB
tuple node $\fTup{\aForm}{\fVal}$. We consider three cases based on the truth of
literal $\rel{R}(\fVal(\varVec))$. If the tuple exists then there exists one
successful binding of the rule connected through its goals to $\card{\varVec}$
copies of $\domQ$ --- one for each value in $\fVal(\varVec)$ and to a tuple node
$\rel{R}(\fVal(\varVec))$. Applying the translation and ignoring redundant
additions and multiplications (nodes with a single child) we get an expression
$\underbrace{1 \cdot \ldots \cdot 1}_{\card{\varVec}} \cdot \kInter(\dlNeg
\rel{R}(\fVal(\varVec))) = \kInter(\dlNeg \rel{R}(\fVal(\varVec))
=\kInterOf{\aForm}$. If $\rel{R}(\fVal(\varVec))$ is undetermined, then the
existence of tuple $\fTup{\aForm}{\fVal}$ and the status of the rule deriving it
are also undetermined. Thus, the structure of the graph is the same as in the
first case and the translation returns the same dual provenance polynomial and
we have $\kInter(\dlNeg \rel{R}(\fVal(\varVec))) =\kInterOf{\aForm}$.  Finally,
if literal $\rel{R}(\varVec)$ is false then tuple $\fTup{\aForm}{\fVal}$ is
absent from the query result and the derivation producing it failed. Since
failed derivations are only connected to failed goals, the expression created by
$\provToFO$ is
$\cdot ( \kInter(\dlNeg \rel{R}(\fVal(\varVec)))) = \kInter(\dlNeg
\rel{R}(\fVal(\varVec))) =\kInterOf{\aForm}$.

\myproofpar{Inductive step}
We assume that the statement holds for formulas $\aForm_1$ and/or $\aForm_2$ and show that it also holds for a formula $\aForm$.

\myproofpar{$\aForm \defas \aForm_1 \wedge \aForm_2$}
We have $\kInterOf{\aForm} = \kInterOf{\aForm_1} \cdot \kInterOf{\aForm_2}$ for valuation $\fVal$.  The program $\formToQ(\aForm)$ consists of single rule $\formQ{}(\freeOf{\aForm}) \dlImp \formQ{1}(\freeOf{\aForm_1}), \formQ{1}(\freeOf{\aForm_2})$. For a valuation $\fVal$, $\fTup{\aForm}{\fVal}$ exists if $\aForm_1 \wedge \aForm_2$ evaluates to true. In this case, $\fTup{\aForm}{\fVal}$ is connected through a rule and two goal nodes to the subgraphs  $\fG{\aForm_1}{\fVal}$ and $\fG{\aForm_2}{\fVal}$.
Let $p_i = \fP{\aForm_i}{\fVal}$  for $i \in \{1,2\}$ and $p = \fP{\aForm}{\fVal}$.
We have $p = +((+(p_1)) \cdot (+(p_2))) = p_1 \cdot p_2$. By the induction hypothesis $p_i = \kInterOf{\aForm_i}$ and we get $p = \kInterOf{\aForm_1} \cdot \kInterOf{\aForm_2} = \kInterOf{\aForm}$. Node $\fTup{\aForm}{\fVal}$ is undetermined, if one input is undetermined and the other is either undetermined or true. In this case we get the same provenance expression. Finally, $\fTup{\aForm}{\fVal}$ is missing if at least one of the inputs is missing. In this case the provenance graph only contains a subgraph for missing or undetermined inputs. That is, there are three possible provenance expressions extracted from the graph $p_1 = \kInterOf{\aForm_1}$, $p_2 = \kInterOf{\aForm_2}$, or $p_1 \cdot p_2 = \kInterOf{\aForm_1} \cdot \kInterOf{\aForm_2}$. We claim that in all three cases $p = 0 = \kInterOf{\aForm}$. Consider a failed input $\kInterOf{\aForm_i}$ for $i \in \{1,2\}$. If $\aForm_i$ evaluates to false, then $\kInterOf{\aForm_i} = 0$. Then $p_i = 0$, the subgraph $g_i$ is connected to node $\fTup{\aForm}{\fVal}$, and it follows that $p = 0$.

\myproofpar{$\aForm \defas \aForm_1 \vee \aForm_2$} We have
$\kInterOf{\aForm_1 \vee \aForm_2} = \kInterOf{\aForm_1} + \kInterOf{\aForm_2}$.
A disjunction is translated into two rules (a union) corresponding to joining
$\formQ{i}$ for $i \in \{1,2\}$ with a number of copies of $\domQ$ such that
both rules return $\freeOf{\aForm}$. Tuple $\fTup{\aForm}{\fVal}$ exists as long
as one tuple $\fTup{\aForm_i}{\fVal}$ for $i \in \{1,2\}$ exists. In this case the
tuple will be connected through successful or undetermined derivations to nodes
$\fTup{\aForm_i}{\fVal}$. Let $p = \fP{\aForm}{\fVal}$
and $p_i = \fP{\aForm_i}{\fVal}$. We have either $p = p_1$,
$p = p_2$, or $p = p_1 + p_2$. In the last case, we have
$p = \kInterOf{\aForm}$. Now consider the case where $p=p_1$ (the case for
$p=p_2$ is symmetric). Recall that $\fTup{\aForm_1}{\fVal}$ is missing iff
$\fVal(\aForm_1)$ evaluates to false in which case $\kInterOf{\aForm_1} =
0$. Then we get $p = p_2 = 0 + p_2 = \kInterOf{\aForm_1} +
\kInterOf{\aForm_2}$. Formula $\aForm_1 \vee \aForm_2$ evaluates to false or
undetermined if both $\aForm_i$ are either undetermined or false. In both cases,
the provenance graph will contain both $g_1$ and $g_2$ and
$p = \kInterOf{\aForm}$.

\myproofpar{$\aForm \defas \exists x\,\aForm_1$}
We have $\kInterOf{\aForm} = \sum_{a \in \fDom} \kInterOfVal{\aForm_1}{\fVal[x \mapsto a]}$.
There is one rule deriving $\formQ{}$ which binds $x$, the variable bound by the existential quantifier, to $\domQ$ (all values from $\fDom$) and $\aForm_1$. Tuple $\fTup{\aForm}{\fVal}$ exists iff there exists at least one $a \in \fDom$ for which $\fVal[x \mapsto a](\aForm_1)$ is true which based on the induction hypothesis implies that tuple $\fTup{\aForm_1}{\fVal[x \mapsto a]}$
exists. Recall that an existing tuple is connected to successful and undetermined rule derivations.
Let $p_a = \fP{\aForm_1}{\fVal[x \mapsto a]}$.
Since the tuple node $\fTup{\aForm}{\fVal}$ will be replaced with $+$ by $\provToFO$, the expression $p$ is a sum $p = \sum_{a \in \fDom \wedge \fVal[x \mapsto a](\freeOf{\aForm_1})} p_a$. By the induction hypothesis $p_a = \kInterOfVal{\aForm_1}{\fVal[x \mapsto a]}$. Furthermore, if $\fVal[x \mapsto a](\freeOf{\aForm_1})$ is false then $\kInterOfVal{\aForm_1}{\fVal[x \mapsto a]} = 0$. Thus, $p = \kInterOf{\aForm}$. Tuple $\fTup{\aForm}{\fVal}$ is missing if all its derivations fail and undetermined if no derivation is successful and at least one if undetermined. In both cases the provenance graph connects this tuple to all possible derivations and, thus, $p = \kInterOf{\aForm}$.

\myproofpar{$\aForm \defas \forall x\,\aForm_1$} We have
$\kInterOf{\aForm} = \prod_{a \in \fDom} \kInterOfVal{\aForm_1}{\fVal[x \mapsto
  a]}$.  The translation uses a standard way of encoding universal
quantification in Datalog as double negation
($\dlNeg \exists x\, \dlNeg \aForm_1$). Note that there is a single rule
deriving $\formQ{}$ with a single negated atom $\formQ{'}$ in its body.  Let
$t' = \formQ{'}(\fVal(\freeOf{\aForm}))$, $g'$ be the subgraph rooted at $t'$,
and $p' = \provToFO(g')$.  Furthermore, let
$t_a = \fTup{\aForm_1}{\fVal[x \mapsto a]}$ and
$p_a = \fP{\aForm_1}{\fVal[x \mapsto a]}$.  Based on step 6 of the translation,
$t'$ is replaced with $\cdot$ and we get an expression $p' = \prod_{a \in A'} p_a$
for some $A' \subseteq \fDom$. Now the existence of $t'$ and of the individual
$t_a$ determine the set $A'$ (i.e., all $t_a$ that are connected to $t'$ in the
provenance graph). Tuple $\fTup{\aForm}{\fVal}$ exists (is missing, is
undetermined) if $t'$ is missing (exists, is undetermined). We have
$\fP{\aForm}{\fVal} = p'$. Now, the body of the rule deriving $\formQ{'}$
consists of one negated atom $\dlNeg \formQ{1}(\freeOf{\aForm_1})$. Thus, $t'$
exists if for at least one $a \in \fDom$ we have that
$\fVal[x \mapsto a](\aForm_1)$ is false, i.e., the universal quantification
evaluates to false. In this case, $t'$ is connected to all tuples $t_a$ which are
existing or undetermined. If we can prove that for every $a \in (A - A')$ we
have $\kInterOfVal{\aForm_1}{\fVal[x \mapsto a]} = 0$, then
$\kInterOf{\aForm} = p' = p$. Since for every $a \in (A - A')$ we know that
$t_a$ is missing, by the induction hypothesis we have
$\kInterOfVal{\aForm_1}{\fVal[x \mapsto a]} = 0$. If $t'$ is missing, then all
its derivations are failed and all tuples $t_a$ exist. In this case $A = A'$ and
$\kInterOf{\aForm} = p' = p$.  Finally, $t'$ is undetermined if all its
derivations are either failed or undetermined and at least one derivation is
undetermined. This is the case when all tuples $t_a$ are either existing or
undetermined and at least one $t_a$ is undetermined. In this, case $t'$ is still
connected to all tuples $t_a$ and we have again $A=A'$.

\myproofpar{$\aForm \defas x \formOp y$}
The translation creates a single rule consisting of the comparison and atoms $\domQ(x)$ and $\domQ(y)$ which ensure safety. There is a single rule derivation for each valuation $\fVal$ that assigns constants to $x$ and $y$. We have $\kInterOf{x \formOp y} = 1$ if $\fVal(x) \formOp \fVal(y)$ and $\kInterOf{x \formOp y} = 0$ otherwise.
Tuple $t = \formQ{}(\fVal(x), \fVal(y))$ exists if $\fVal(x) \formOp \fVal(y)$.
In this case $\fP{\aForm}{\fVal} = 1 \cdot 1 \cdot 1 = 1 = \kInterOf{x \formOp y}$. Otherwise, $\fP{\aForm}{\fVal} = 1 \cdot 1 \cdot 0 = \kInterOf{x \formOp y}$.
}
\end{proof}

 \section{Computing Explanations} \label{sec:compute-gp}

We now present our approach for computing explanations using Datalog.  Our approach generates a Datalog program $\GPProg{P}{\aProvQ}{}$ by rewriting a given query (input program) $P$
to return the edge relation of the explanation
$\explainq(P,\aProvQ,I)$
for a provenance question (\provQ{}) $\aProvQ$.
Recall that a \provQ{} is a pattern describing existing/missing outputs of interest and that an explanation for a \provQ{} is a subgraph of the provenance which contains the provenance of all tuples described by the pattern.

Our approach for computing $\GPProg{P}{\aProvQ}{}$ consists of the following steps that we describe in detail in the following subsections: 1) we unify the input program $P$ with the \provQ{} $\aProvQ$ by propagating constants from $\aProvQ$ top-down to prune derivations of outputs that do not match the \provQ{}; 2) we determine for each IDB predicate whether the explanation may contain existing, missing, or both types of tuples from this predicate. Similarly, for each rule we determine whether successful, failed, all, or no derivations of this rule may occur in the provenance graph; 3) based on restricted and annotated version of the input program produced by the first two steps, we then generate \textit{firing rules} which capture the variable bindings of successful and failed derivations of the input program's rules;
4) The result of the firing rules is a superset of the set of relevant provenance fragments. We introduce additional rules that enforce connectivity to remove spurious fragments;
5) finally, we create rules that generate the edge relation of the explanation. This is the only step that depends on what provenance type (e.g., Fig.\,\ref{fig:mig-exam-why}) is requested.

In the following, we will illustrate our approach using the
provenance question $\aProvQ_{n,s} = \whyq \rel{Q}(n,s)$ from Example\,\ref{ex:example1}, i.e.,
why New York is connected to Seattle via train with one intermediate stop, but there is no direct connection.

\subsection{Unifying the Program with the PQ}
\label{sec:unify-program-with}

The node $\rel{Q}(n,s)$ in the provenance graph (Fig.\,\ref{fig:exam-pg-why-NY-seattle})
is only connected to derivations which return
$\rel{Q}(n,s)$. For instance, if variable $X$ is bound to another city $x$ (e.g., Chicago) in a
derivation of the rule $r_1$, then this rule cannot return the tuple $(n,s)$.
This reasoning can be applied recursively to replace variables in rules with
constants. That is, we unify the rules in the program top-down with the \provQ{}. This process corresponds to selection push-down for relational algebra expressions.  We may create multiple partially
unified versions of a rule or predicate.
For example, to explore successful derivations of $\rel{Q}(n,s)$, we are interested in both train connections from New York to some city ($\rel{T}(n,Z)$) and from any city to Seattle ($\rel{T}(Z,s)$). Furthermore, we need to know whether there is a direct connection from New York to Seattle ($\rel{T}(n,s)$).
We store variable bindings as superscripts to distinguish multiple copies of a rule generated based on different bindings. 

\begin{Example}
Given the question $\aProvQ_{n,s}$, we unify the single rule
$r_1$ using the assignment $(X{=}n,Y{=}s)$: \vspace{-2mm}
$$r_1^{(X=n,Y=s)}:  \rel{Q}(n,s) \dlImp \rel{T}(n,Z), \rel{T}(Z,s), \dlNeg \rel{T}(n,s)$$
\end{Example}

This approach is correct because if we bind a variable in the head of rule, then only rule derivations that agree with this binding can derive tuples that agree with this binding.
Based on this unification step, we know which bindings may produce fragments of $\provGraph(P,I)$ that are relevant
for explaining  the \provQ{} (the pseudocode for the algorithm is presented
\ifnottechreport{in~\cite{LS16a}}\iftechreport{in Algorithm\,\ref{alg:unify-program}}).
For an input $P$, we use $\unProg{P}$ to denote the result of this unification. \iftechreport{
\begin{algorithm}[t]
{\small
  \begin{algorithmic}[1]
    \Procedure{UnifyProgram}{$P$, $\aProvQ)$}
      \State $Q(t) \gets \qPattern(\aProvQ)$
      \State $todo \gets [Q(t)]$
      \State $done \gets \{\}$
      \State $\unProg{P} = []$
      \While {$todo \neq []$}
        \State $a \gets \Call{pop}{todo}$
        \State $\Call{insert}{done,a}$ 
        \State $rules \gets \Call{getRulesForAtom}{P,a}$
        \ForAll {$r \in rules$}
          \State $unRule \gets \Call{unifyRule}{r,a}$
          \State $\unProg{P} \gets \unProg{P} \listconcat unRule$
          \ForAll {$g \in \bodyOf{unRule}$}
            \If {$g \not\in done$}
              $todo \gets todo \listconcat g$
            \EndIf
          \EndFor
        \EndFor
      \EndWhile
      \State \Return $\unProg{P}$
    \EndProcedure
  \end{algorithmic}
}
  \caption{Unify Program With $\provQ$}\label{alg:unify-program}
\end{algorithm}

 }
\subsection{Add Annotations based on Success/Failure}\label{sec:part-inst-game}

For $\whyq Q(t)$ ($\whynotq Q(t)$), we are only interested in subgraphs of the provenance rooted at existing (missing) tuple nodes for $\rel{Q}$.
With this information, we can infer restrictions for the success/failure state of nodes in
the provenance graph
that are directly or indirectly connected to \provQ{} node(s) (belong to the explanation).
We store these restrictions as annotations $\greenT$, $\redF$, and $\redF/\greenT$
on heads and goals of rules and use these annotations to guide the generation of rules that capture derivations in step 3.
Here, $\greenT$ ($\redF$) indicates that we are only interested in successful (failed) nodes, and $\redF/\greenT$ that we are interested in both. 

\begin{Example} \label{ex:example3}
Continuing with our running example question $\aProvQ_{n,s}$, we know that $\rel{Q}(n,s)$  is in the result (Fig.\,\ref{fig:running-example-db}). This implies that only successful rule nodes and their successful goal nodes can be connected to this tuple node.
Note that this annotation
only indicates that it is sufficient to focus on successful rule derivations
since failed ones cannot be connected to $\rel{Q}(n,s)$. \begin{align*}
r_1^{(X=n,Y=s), \greenT}: \rel{Q}(n,s)^{\greenT} \dlImp \rel{T}(n,Z)^{\greenT}, \rel{T}(Z,s)^{\greenT}, \dlNeg \rel{T}(n,s)^{\greenT}
\end{align*}
We now propagate the annotations of the goals in $r_1$ throughout the program.
That is, for any goal that is an IDB predicate, we propagate its annotation to the head of all rules deriving the goal's predicate
and, then, propagate these annotations  to the corresponding rule bodies.
Note that the inverted annotation is propagated for negated goals (e.g., $\dlNeg \rel{T}(n,s)^{\greenT}$). For instance, if $\rel{T}$ would be an IDB predicate, then we would  annotate the head of all rules deriving $\rel{T}(n,s)$ with $\redF$, because $\rel{Q}(n,s)$ can only exist if $\rel{T}(n,s)$ does not exist.
\end{Example}

Partially unified atoms such as $\rel{T}(n,Z)$ may occur in both negative and positive goals.
We annotate such atoms   with $\redF/\greenT$. The algorithm generating the annotation consists of the steps shown below (the pseudocode is presented
\ifnottechreport{in~\cite{LS16a}}\iftechreport{in Algorithm\,\ref{alg:annot-program}}).
We use $\adProg{P}$ to denote the result of this algorithm for $\unProg{P}$ (input to this step). \smallskip
\begin{compactenum}
\item Annotate the head of all rules deriving tuples matching
the question with $\greenT$ (why) or $\redF$ (why-not).
\item Repeat the following steps until a fixpoint is reached:
  \begin{compactenum}
  \item Propagate the annotation of a rule head to goals in the rule body as follows: propagate $\greenT$ for $\greenT$ annotated heads and $\redF/\greenT$ for $\redF$ annotated heads.
  \item For each annotated positive goal in the rule body, we propagate its annotation ($\redF$, $\greenT$, or $\redF/\greenT$) to all rules that have this atom in the head.     For negated goals, we propagate the inverted annotation (e.g., $\redF$ for $\greenT$) unless the annotation is $\redF/\greenT$ in which case we propagate $\redF/\greenT$.
  \end{compactenum}
\end{compactenum}
\smallskip
\iftechreport{
\begin{algorithm}[t]
  
{\small
  \begin{algorithmic}[1]
    \Procedure{AnnotProgram}{$\unProg{P}$, $\aProvQ$}
   \State $state \gets \qType(\aProvQ)$
    \State $Q(t) \gets \qPattern(\aProvQ)$
      \State $todo \gets [Q(t)^{state}]$
     \State $done \gets \{\}$
      \State $\adProg{P} = []$
      \While {$todo \neq []$}
        \State $a \gets \Call{pop}{todo}$
       \State $state \gets \qType(a)$
        \State $\Call{insert}{done,a}$ 
        \State $rules \gets \Call{getRulesForAtom}{\unProg{P},a}$
        \ForAll {$r \in rules$}
          \State $annotRule \gets \Call{annotRule}{r,state}$
          \State $\adProg{P} \gets \adProg{P} \listconcat annotRule$
          \ForAll {$g \in \bodyOf{annotRule}$}
              \If {state = \redF}
                  \State $state \gets \redF/\greenT$
              \EndIf
             \If {\Call{isNegated}{g}}
              \State $state \gets \Call{switchState}{state}$
              \EndIf
            \If {$g^{state} \not\in done \wedge \Call{isIDB}{g}$}
              \State $todo \gets todo \listconcat g^{state}$
            \EndIf
          \EndFor
        \EndFor
      \EndWhile
     \ForAll {$r \in \adProg{P}$}
         \If {$\qType(r) = \redF/\greenT$}
             \State $\adProg{P} \gets \Call{removeAnnotatedRules}{\adProg{P},r,\{\redF,\greenT\}}$
         \EndIf
      \EndFor
      \State \Return $\adProg{P}$
\EndProcedure
  \end{algorithmic}
}
  \caption{Success/Failure Annotations}\label{alg:annot-program}
\end{algorithm}

 }

\subsection{Creating Firing Rules}\label{sec:creat-firer-rules}

To compute the relevant subgraph of $\provGraph(P,I)$ (the explanation) for a \provQ{}, we need to determine successful and/or failed rule derivations.
Each derivation paired with the information whether it is successful
over the given database (and which goals are failed in case it is not
successful) is sufficient for generating a fragment of $\provGraph(P,I)$.
Successful derivations are always part of $\provGraph(P,I)$ for a given query (input program) $P$ whereas failed
rule derivations only appear
if the tuple in the head failed,
i.e., there are no successful derivations of any rule with
this head. To capture the variable bindings of successful/failed rule derivations, we create
\textit{``firing rules''}.
For successful rule derivations, a firing rule consists of the body of the rule
(but using the firing version of each predicate in the body) and
a new head predicate that contains all variables used in the rule.
In this way, the firing rule captures all the variable bindings of a rule derivation.
\iftechreport{
The rationale behind this is that rule derivations apply a valuation to assign constants  to the variables of a rule. A rule derivation is successful if after applying the valuation to the body, all goals are successful. For any successful rule derivation, the result of applying the valuation to the head is part of the result of the program containing the rule. Thus, if for a rule $r$ we create a firing rule that has the same body as $r$, but has all body variables in the head, then the IDB predicate computed by the firing rule stores all successful derivations of $r$.
}
Furthermore, for each IDB predicate $R$ that occurs as a head of a rule $r$, we create a firing rule that has the firing version of predicate $R$ in the head and firing version of the rules $r$ deriving the predicate in the body.
\iftechreport{
These rules capture existing IDB tuples.
}
For EDB predicates, we create firing rules that have the firing version of the predicate in the head and the EDB predicate in the body.
\iftechreport{
These rules create copies of EDB relations. Note that for the positive case these rules are not strictly necessary, but as we will show they are necessary once negation and missing answers are introduced.
}
\begin{figure}[t]\centering
 $\,$\\[-4mm]
  \begin{minipage}{.78\linewidth}
   \centering
    \begin{align*}
      \fire{Q}{}{\greenT}(n,s) &\dlImp \fire{r_1}{Q}{\greenT}(n,s,Z)\\[1mm]
\fire{r_1}{Q}{\greenT}(n,s,Z) &\dlImp \fire{T}{}{\greenT}(n,Z), \fire{T}{}{\greenT}(Z,s), \fire{T}{}{\redF}(n,s)\\[1mm]
\fire{T}{}{\greenT}(n,Z) &\dlImp \rel{T}(n,Z)\\
\fire{T}{}{\greenT}(Z,s) &\dlImp \rel{T}(Z,s)\\
\fire{T}{}{\redF}(n,s) &\dlImp \dlNeg \rel{T}(n,s)
    \end{align*}\\[-6mm]
    \caption{Example firing rules for $\whyq \rel{Q}(n,s)$}\label{fig:exam-fire-rules}
  \end{minipage}

\end{figure}

\begin{Example}\label{ex:example4}
Consider the annotated program in Example\,\ref{ex:example3} for the question $\aProvQ_{n,s} = \whyq \rel{Q}(n,s)$.
We generate the firing rules shown in Fig.\,\ref{fig:exam-fire-rules}.
The firing rule for $r_1^{(X=n,Y=s), \greenT}$ (the second rule from the top) is derived
from the rule $r_1$ by adding $Z$
(the only existential variable) to the head,
renaming the head predicate as $\fire{r_1}{Q}{\greenT}$,
and replacing each goal with its firing version
(e.g., $\fire{T}{}{\greenT}$ for the two positive goals and $\fire{T}{}{\redF}$ for the negated goal).
Note that negated goals are replaced with firing rules that have inverted annotations
(e.g., the goal $\dlNeg \rel{T}(n,s)^{\greenT}$ is replaced with $\fire{T}{}{\redF}(n,s)$).
\iftechreport{In the following, we will define rules with annotation $\redF$ to capture missing tuples of an EDB or IDB predicate and failed rule derivations.}
Furthermore, we introduce firing rules for EDB tuples (three rules at the bottom). \iftechreport{
  Evaluated over the example instance from Fig.\,\ref{fig:running-example-db} these rules produce the following instance:
  \begin{align*}
   & \fire{Q}{}{\greenT}(n,s) \\[3mm]
& \fire{r_1}{Q}{\greenT}(n,s,c) && \fire{r_1}{Q}{\greenT}(n,s,w) \\[3mm]
&\fire{T}{}{\greenT}(n,c) &&\fire{T}{}{\greenT}(n,w) \\
& \fire{T}{}{\greenT}(c,s) && \fire{T}{}{\greenT}(s,s) && \fire{T}{}{\greenT}(w,s)\\[3mm]
&\fire{T}{}{\redF}(n,s)
  \end{align*}
}
\end{Example}

We, now, extend firing rules to support queries with negation and capture missing answers.
To construct a $\provGraph(P,I)$ fragment corresponding to a missing tuple,
we need to find failed rule derivations with the tuple in the head
and ensure that no successful derivations  with this head exist (otherwise, we may capture irrelevant failed derivations of existing tuples).
In addition, we need to determine which goals are failed for a failed rule derivation because only failed goals are connected to the node representing the failed rule derivation in the
provenance graph.
To capture this information, we add additional boolean variables --- $V_i$ for goal $g^i$ --- to the head of a firing rule that record for each goal whether it failed or not.
The body of a firing rule for failed rule derivations is created by
replacing every goal in the body with its $\redF/\greenT$ firing version, and
adding the firing version of the negated head to the body (to ensure that only bindings for missing tuples are captured). Firing rules capturing failed derivations use the $\redF/\greenT$ firing versions of their goals because not all goals of a failed derivation have to be failed and the failure status determines whether the corresponding goal node is part of the explanation.
\iftechreport{
  A $\redF/\greenT$ firing rule for a predicate $\rel{R}$ captures all tuples in $\tupDom(\rel{R})$ no matter whether they exist or not. An additional boolean attribute is used to store for each such tuple whether it exists or not.
  \begin{Example}
    Consider an EDB relation $\rel{R(A,B)}$, domain assignment $\domA(\rel{R.A}) = \domA(\rel{R.B}) = \{a,b\}$ and instance $\{\rel{R}(a,a), \rel{R}(b,b)\}$. The firing rules for $\fire{R}{}{\redF/\greenT}$ using the queries $\domA_{\rel{R.A}}$ and $\domA_{\rel{R.B}}$ provided by the user to compute the domain assignment are:
    \begin{align*}
      \fire{R}{}{\redF/\greenT}(X,Y,\boolT) &\dlImp \fire{R}{}{\greenT}(X,Y)\\
\fire{R}{}{\redF/\greenT}(X,Y,\boolF) &\dlImp \fire{R}{}{\redF}(X,Y)\\
\fire{R}{}{\greenT}(X,Y) &\dlImp \rel{R}(X,Y)\\
\fire{R}{}{\redF}(X,Y) &\dlImp \domA_{\rel{R.A}}(X), \domA_{\rel{R.B}}(X), \dlNeg \rel{R}(X,Y)
    \end{align*}
The third rule computes existing tuples creating a copy of relation $\rel{R}$. The fourth rule, enumerates all tuples in $\tupDom(\rel{R})$ using the domain assignment and only returns tuples that do not exist. The first and the second rule then combine the results of the last two rules and store whether a tuple exists as a boolean value in an additional attribute.
Evaluating these rules we generate the following instance:
    \begin{align*}
      &\fire{R}{}{\redF/\greenT}(a,a,\boolT) &      &\fire{R}{}{\redF/\greenT}(b,b,\boolT)\\
            &\fire{R}{}{\redF/\greenT}(a,b,\boolF) &      &\fire{R}{}{\redF/\greenT}(b,a,\boolF)
    \end{align*}

  \end{Example}

}
A firing rule capturing missing tuples may not be safe, i.e., it may contain variables that only occur in negated goals. These variables should be restricted to the associated domains for the attributes
the variables are bound to. Recall that associated domain $\domA(\rel{R.A})$ for an attribute $\rel{R.A}$ is given as an unary query $\domA_\rel{R.A}$. We use these queries in firing rules to restrict the values
a variable is bound to. Thus, we ensure that only missing answers formed from the associated domains are considered and that firing rules are safe.

\begin{figure}[t]\centering
 $\,$\\[-4mm]
\begin{minipage}{1\linewidth}
    \begin{align*}
\fire{Q}{}{\redF}(s,n) &\dlImp \dlNeg \fire{Q}{}{\greenT}(s,n)\\
\fire{Q}{}{\greenT}(s,n) &\dlImp \fire{r_1}{Q}{\greenT}(s,n,Z)\\[1mm]
\fire{r_1}{Q}{\redF}(s,n,Z,V_1,V_2, \dlNeg V_3) &\dlImp \fire{Q}{}{\redF}(s,n), \fire{T}{}{\redF/\greenT}(s,Z,V_1),\\
                                             &\hspace{5mm}\fire{T}{}{\redF/\greenT}(Z,n,V_2),  \fire{T}{}{\redF/\greenT}(s,n,V_3)\\
\fire{r_1}{Q}{\greenT}(s,n,Z) &\dlImp \fire{T}{}{\greenT}(s,Z), \fire{T}{}{\greenT}(Z,n),					\fire{T}{}{\redF}(s,n)\\[1mm]
\fire{T}{}{\redF/\greenT}(s,Z,\boolT) &\dlImp \fire{T}{}{\greenT}(s,Z)\\
\fire{T}{}{\redF/\greenT}(s,Z,\boolF) &\dlImp \fire{T}{}{\redF}(s,Z)\\
\fire{T}{}{\greenT}(s,Z) &\dlImp \rel{T}(s,Z)\\
\fire{T}{}{\redF}(s,Z) &\dlImp \domA_{\rel{T.toCity}}(Z), \dlNeg \rel{T}(s,Z)
    \end{align*}\\[-12mm]
    \caption{Example firing rules for $\whynotq \rel{Q}(s,n)$}\label{fig:exam-fire-negated}
  \end{minipage}

\end{figure}

\begin{Example}\label{ex:example5}
Consider the question $\whynotq \rel{Q}(s,n)$  from Example\,\ref{ex:example1}.
The firing rules generated for this question are in Fig.\,\ref{fig:exam-fire-negated}.
We exclude the rules for the second goal
$\rel{T}(Z,n)$ and the negated goal $\dlNeg \rel{T}(s,n)$ which are analogous to the rules
for the first goal $\rel{T}(s,Z)$.
New York cannot be reached from Seattle with exactly one transfer, i.e.,  $\rel{Q}(s,n)$ is not in the result. Thus, we are only interested in failed derivations of rule $r_1$ with $X{=}s$ and $Y{=}n$.
Furthermore, each rule node in the provenance graph corresponding to such a derivation
will only be connected to failed subgoals.
Thus, we need to capture which goals are successful or failed
for each such failed derivation.
We model this using boolean variables $V_1$, $V_2$, and $V_3$
(one for each goal) that are set to true iff the tuple corresponding to the goal  exists.
The firing version $\fire{r_1}{Q}{\redF}(s,n,Z,V_1,V_2,\dlNeg V_3)$ of $r_1$  returns
all variable bindings for derivations of $r_1$ such that $\rel{Q}(s,n)$ is the head
(i.e., guaranteed by adding $\fire{Q}{}{\redF}(s,n)$ to the body),
the rule derivations are failed, and the tuple corresponding to the $i^{th}$ goal exists for this binding iff $V_i$ is true. The failure status of the $i^{th}$ goal is $V_i$  for positive goals and  $\neg V_i$ for negated goals. To produce all these bindings, we need rules capturing successful and failed tuple nodes
for each subgoal of the rule $r_1$.
We annotate such rules with $\redF/\greenT$ and use a boolean variable (true or false) to record
whether a tuple exists
(e.g., $\fire{T}{}{\redF/\greenT}(s,Z,\boolT) \dlImp \fire{T}{}{\greenT}(s,Z)$ is one of these rules).
Similarly, $\fire{T}{}{\redF/\greenT}(s,n,false)$
represents the fact that tuple $\rel{T}(s,n)$ (connection from Seattle to New York)
is missing. This causes the third goal of $r_1$ to succeed for any derivation where $X{=}s$ and $Y{=}n$.
For each unified EDB atom annotated with $\redF/\greenT$, we create four rules: one for existing tuples
(e.g., $\fire{T}{}{\greenT}(s,Z) \dlImp \rel{T}(s,Z)$),
 one for the failure case (e.g., $\fire{T}{}{\redF}(s,Z) \dlImp \domA_{\rel{T.toCity}}(Z), \dlNeg \rel{T}(s,Z)$), and two for the $\redF/\greenT$ version.
For the failure case,
we use predicate $\domA_{\rel{T.toCity}}$ to only consider missing tuples $(s,Z)$ where $Z$ is a value from the associated domain. \end{Example}

Algorithm\,\ref{alg:firing-rules} takes as input the program $\adProg{P}$ produced by step 2  and outputs a program $\fireProg{P}$ containing firing rules.  The pseudocode for the subprocedures is presented
\ifnottechreport{in~\cite{LS16a}}\iftechreport{in Algorithm\,\ref{alg:firing-rules-sub}}. The algorithm maintains a queue $todo$ of annotated atoms that have to be processed which is initialized with $\qPattern(\aProvQ)$, i.e., the provenance question atom. Furthermore, we maintain a set $done$ of atoms that have been processed already. Variables $todo$, $done$, and $\fireProg{P}$ are global variables that are shared with the subprocedures of this algorithm.  For each atom
$R(t)^{\adornment}$  (line 8) from the queue (here $\adornment$ is the annotation of the atom, e.g., $\redF$),
we mark the atom as done (line 9). We then consider two cases: $R$ is an EDB atom or an IDB atom in which case we have to create firing rules for the predicate (relation) and the rules deriving it.
\iftechreport{
Recall that an EDB atom is a relation in the schema over the input Datalog program $P$ and
an IDB atom is the head atom of the rule(s) in $P$.
}
The firing rules for EDB predicates check whether the tuples do or do not exist.
These rules allow us to determine the success or failure of goals corresponding EDB predicates in rule derivations. For IDB predicates, we create firing rules that determine their existance based on successful or failed rule derivations captured by firing rules for the rules of the program. Consider a given program $P$ with two rules:
1) $r_1: \rel{Q}(X) \dlImp \rel{R}(X,Y), \rel{Q_1}(Y)$ and 2) $r_2: \rel{Q_1}(Y) \dlImp \rel{S}(Y,Z)$
where $\rel{R}$ and $\rel{S}$ are EDB relations and $\rel{Q}$ and $\rel{Q_1}$ are IDB predicates.
To capture provenance for the predicate $\rel{Q}(X)$,
we create firing rules for $\rel{R}$ and $\rel{S}$ to check existence or absence of tuples matching $t$ in $\rel{R}$ and $\rel{S}$.
Moreover, we also generate firing rules for rules $r_1$ and $r_2$
to explain how derivations of $\rel{Q}(X)$ through these rules have succeeded or failed. The firing rule for $r_1$ uses the firing rule for IDB predicate $\rel{Q_1}$
which in turn uses the firing rule for $r_2$ since $\headOf{r_2} = \rel{Q_1}$.
We describe these two cases in the following.

\mypartitle{EDB atoms (line 13)}
For an EDB atom $R(t)^{\greenT}$, we use procedure \textsc{createEDBFiringRule} to create one rule $\fire{R}{}{\greenT}(t) \dlImp R(t)$ that returns tuples from relation $R$ that match $t$. For missing tuples ($R(t)^{\redF}$), we extract all variables from $t$ (some arguments may be constants propagated during unification) and create a rule that returns all tuples that can be formed from values of the associated domains of the attributes these variables are bound to and do not exist in $R$. This is achieved by adding goals $\domA{(X_i)}$
as explained in Example\,\ref{ex:example5}.

\begin{algorithm}[t]
{\small
  \begin{algorithmic}[1]
    \Procedure{CreateFiringRules}{$\adProg{P}$, $\aProvQ$}
    \State $\fireProg{P} \gets []$
    \State $state \gets typeof(\aProvQ)$
    \State $Q(t) \gets \qPattern(\aProvQ)$
    \State $todo \gets [Q(t)^{state}]$
    \State $done \gets \{\}$
    \While {$todo \neq []$} \Comment{create rules for a predicate}
       \State $R(t)^{\adornment} \gets \Call{pop}{todo}$
       \State $\Call{insert}{done,R(t)^{\adornment}}$
       \If {\Call{isEDB}{$R$}}
       \State \Call{CreateEDBFiringRule}{$\fireProg{P},R(t)^{\adornment}$}
       \Else
       \State \Call{CreateIDBNegRule}{$\fireProg{P},R(t)^{\adornment}$}
       \State $rules \gets \Call{getRules}{R(t)^{\adornment}}$
       \ForAll {$r \in rules$} \Comment{create firing rule for $r$}
         \State $args \gets \argsOf{\headOf{r}}$
         \State $args \gets args \listconcat (\argsOf{\bodyOf{r}} - \argsOf{\headOf{r}})$
         \State \Call{CreateIDBPosRule}{$\fireProg{P},R(t)^{\adornment},r,args$}
       \State \Call{CreateIDBFiringRule}{$\fireProg{P},R(t)^{\adornment},r,args$}
       \EndFor
       \EndIf
    \EndWhile
    \State \Return $\fireProg{P}$
\EndProcedure
  \end{algorithmic}
}
  \caption{Create Firing Rules}\label{alg:firing-rules}
\end{algorithm}

\mypartitle{IDB atoms (lines 13-19)}
IDB atoms with $\redF$ or $\redF/\greenT$ annotations are handled in the same way as EDB atoms with these annotations.
If the atom is $R(t)^{\redF}$ (line 13), we create a rule with $\dlNeg \fire{R}{}{\greenT}(t)$ in the body using the associated domain queries to restrict variable bindings. Similarly, for $R(t)^{\redF/\greenT}$, the procedure called in line 13 adds two additional rules as shown in Fig.\,\ref{fig:exam-fire-negated} ($5^{th}$ and $6^{th}$ rule) for EDB atoms.
Both types of rules only use the positive firing version for $R(t)$ and domain predicates in their body. Thus, these rules are independent of which rules derive $R$.
Now, for any $R$, we create positive firing rules that correspond to the derivation of $R$ through one particular rule. For that, we iterate over the annotated versions of all rules deriving $R$ (lines 14+15).
For each rule $r$ with head $\rel{R}(t)$, we create a rule $\fire{R}{}{\greenT}(t) \dlImp \fire{r}{R}{\greenT}(\vec{X})$ where $\vec{X}$ is the concatenation of $t$ with all existential variables from the body of $r$.

\mypartitle{Rules (line 15-19)}
Consider a rule $r: R(t) \dlImp g_1(\vec{X_1}), \\\ldots, g_n(\vec{X_n})$.
If the head of $r$ is annotated with $\greenT$, then we create a rule with head $\fire{r}{R}{\greenT}(\vec{X})$ where $\vec{X} = \varsOf{r}$ (stored in variable $args$, lines 16+17) and the same body as $r$ except that each goal is replaced with its firing version with appropriate annotation (e.g., $\greenT$ for positive goals). For rules annotated with $\redF$ or $\redF/\greenT$,
we create one additional rule with head $\fire{r}{R}{\redF}(\vec{X},\vec{V})$ where $\vec{X}$ is defined as above, and $\vec{V}$ contains $V_i$ if the $i^{th}$ goal of $r$ is positive and $\dlNeg V_i$ otherwise. The body of this rule contains the $\redF/\greenT$ version of every goal in $r$'s body plus an additional goal $\fire{R}{}{\redF}$ to ensure that the head atom is failed. As an example for this type of rule, consider the third rule from the top in Fig.\,\ref{fig:exam-fire-negated}.

\iftechreport{
\begin{algorithm}[t]

{\small

  \begin{algorithmic}[1]
    \Procedure{CreateEDBFiringRule}{$\fireProg{P}$, $R(t)^{\adornment}$}
           \State $[X_1,\ldots,X_n] \gets \varsOf{t}$
              \State $r_{\greenT}  \gets \fire{R}{}{\greenT}(t) \dlImp R(t)$
              \State $r_{\redF}  \gets \fire{R}{}{\redF}(t) \dlImp \domA({X_1}), \ldots, \domA({X_n}), \dlNeg R(t)$
              \State $r_{\redF/\greenT-1}  \gets \fire{R}{}{\redF/\greenT}(t,\boolT) \dlImp \fire{R}{}{\greenT}(t)$
              \State $r_{\redF/\greenT-2}  \gets \fire{R}{}{\redF/\greenT}(t,\boolF) \dlImp \fire{R}{}{\redF}(t)$
           \If {$\adornment = \greenT$}
              \State $\fireProg{P} \gets \fireProg{P} \listconcat r_{\greenT}$      
           \ElsIf {$\adornment = \redF$}

              \State $\fireProg{P} \gets \fireProg{P} \listconcat r_{\redF}$      
           \Else
              \State $\fireProg{P} \gets \fireProg{P} \listconcat r_{\greenT} \listconcat r_{\redF} \listconcat r_{\redF/\greenT-1} \listconcat r_{\redF/\greenT-2}$      
           \EndIf
    \EndProcedure
  \end{algorithmic}

  \begin{algorithmic}[1]
    \Procedure{CreateIDBNegRule}{$\fireProg{P}$, $R(t)^{\adornment}$}
       \State $[X_1,\ldots,X_n] \gets \varsOf{t}$
       \If {$\adornment \neq \greenT$}   
           \State $r_{new}  \gets \fire{R}{}{\redF}(t) \dlImp \domA({X_1}), \ldots, \domA({X_n}), \dlNeg \fire{R}{}{\greenT}(t)$
           \State $\fireProg{P} \gets \fireProg{P} \listconcat r_{new}$      
       \EndIf
       \If {$\adornment = \redF/\greenT$}
           \State $r_{\greenT}  \gets \fire{R}{}{\redF/\greenT}(t,true) \dlImp \fire{R}{}{\greenT}(t)$
           \State $r_{\redF}  \gets \fire{R}{}{\redF/\greenT}(t,false) \dlImp \fire{R}{}{\redF}(t)$
           \State $\fireProg{P} \gets \fireProg{P} \listconcat r_{\greenT} \listconcat r_{\redF}$
       \EndIf
    \EndProcedure
  \end{algorithmic}

  \begin{algorithmic}[1]
    \Procedure{CreateIDBPosRule}{$\fireProg{P}$, $R(t)^{\adornment}$, $r$, $args$}
            \State $r_{pred}  \gets \fire{R}{}{\greenT}(t) \dlImp \fire{r}{R}{\greenT}(args)$
          \State $\fireProg{P} \gets \fireProg{P} \listconcat r_{pred}$      
    \EndProcedure
  \end{algorithmic}

  \begin{algorithmic}[1]
    \Procedure{CreateIDBFiringRule}{$\fireProg{P}$, $R(t)^{\adornment}$, $r$, $args$}
         \State $body_{new} \gets []$
         \ForAll {$g_i(\vec{X}) \in \bodyOf{r}$}
            \State $\adornment_{goal} \gets \greenT$
            \If {$\Call{isNegated}{g_i}$}
               \State $\adornment_{goal} \gets \redF$
            \EndIf
              \State $g_{new} \gets \fire{\predOf{g_i}}{}{\adornment_{goal}}(\vec{X})$
            \State $body_{new} \gets body_{new} \listconcat g_{new}$
            \If { $ g_i(\vec{X})^{\greenT} \not\in done \wedge \adornment = \greenT$ }
               \State $todo \gets todo :: g_i(\vec{X})^{\adornment_{goal}}$
            \EndIf
         \EndFor
       \State $r_{new} \gets \fire{r}{R}{\greenT}(args) \dlImp body_{new}$
       \State $\fireProg{P} \gets \fireProg{P} \listconcat r_{new}$ 
         \If {$\adornment \neq \greenT$}
           \State $body_{new} \gets []$
           \ForAll {$g_i(\vec{X}) \in \bodyOf{r}$}
             \State $g_{new} \gets \fire{\predOf{g_i}}{}{\redF/\greenT}(\vec{X}, V_i)$
             \State $body_{new} \gets body_{new} \listconcat g_{new}$
             \If { $ g_i(\vec{X})^{\redF/\greenT} \not\in done $ }
               \State $todo \gets todo :: g_i(\vec{X})^{\adornment_{goal}}$
            \EndIf
             \If {\Call{isNegated}{$g_i$}}
               \State $args \gets args \listconcat \neg V_i$
             \Else
               \State $args \gets args \listconcat V_i$
             \EndIf
         \EndFor
       \State $r_{new} \gets \fire{r}{R}{\adornment}(args) \dlImp body_{new}$
       \State $\fireProg{P} \gets \fireProg{P} \listconcat r_{new}$ 
         \EndIf
    \EndProcedure
  \end{algorithmic}
}
\caption{Create Firing Rules Subprocedures}\label{alg:firing-rules-sub}
\end{algorithm}

 }

\begin{Theorem}[Correctness of Firing Rules]\label{theo:alg-firingcorrect}
 Let $P$ be an input program, $r$ denote a rule of $P$ with $m$ goals, and $\fireProg{P}$ be the firing version of $P$. We use $r(t) \isSucc P(I)$ to denote that the rule derivation $r(t)$ is successful in the evaluation of program $P$ over $I$. The firing rules for $P$ correctly determine existence of tuples, successful  derivations, and failed derivations for missing answers: \vspace{-5mm}
  \begin{itemize}
  \item $\fire{R}{}{\greenT}(t) \in \fireProg{P}(I) \leftrightarrow R(t) \in P(I)$
  \item $\fire{R}{}{\redF}(t) \in \fireProg{P}(I) \leftrightarrow R(t) \not\in P(I)$
  \item $\fire{r}{}{\greenT}(t) \in \fireProg{P}(I) \leftrightarrow r(t) \isSucc P(I)$
  \item $\fire{r}{}{\redF}(t,\vec{V}) \in \fireProg{P}(I) \leftrightarrow r(t) \isFailed P(I) \wedge \headOf{r(t)} \not\in P(I)$ and for $i \in \{1,\ldots,m\}$ we have that $V_i$ is false iff $i^{th}$ goal fails in $r(t)$.
  \end{itemize}
\end{Theorem}

\begin{proof}
\ifnottechreport{
  We prove Theorem\,\ref{theo:alg-firingcorrect} by induction over the structure of a program. For the proof, see~\cite{LS17} or~\cite{LL18}.}
\iftechreport{
We prove Theorem\,\ref{theo:alg-firingcorrect} by induction over the ``depth'' of a program.
  We define the depth $\aDepth$ of predicates, rules, and programs as follows: 1) for all EDB predicates $R$, we define $\depthP{R} = 0$; 2) for an IDB predicate $R$, we define $\depthP{R} = \max_{\headOf{r} = R}\depthP{r}$, i.e., the maximal depth among all rules $r$ with $\headOf{r} = R$; 3) the depth of a rule $r$ is $\depthP{r} = \max_{R \in \bodyOf{r}} \depthP{R} + 1$, i.e., the maximal depth of all predicates in its body plus one; 4) the depth of a program $P$ is the maximum depth of its rules: $\depthP{P} = \max_{r \in P} \depthP{r}$.

\mypara{1) Base Case}
Assume that program $P$ has depth $1$,
e.g., $r: \rel{Q}(\vec{X}) \dlImp \rel{R}(\vec{X_1}), \ldots, \rel{R}(\vec{X_n})$. We first prove that firing rules for EDB atoms are correct,
because only these rules are used for the rules of depth $1$ programs.
A positive version of EDB firing rule $\fire{R}{}{\greenT}$ creates a copy of the input relation $\rel{R}$ and, thus, $\forall t: t \in R \Leftrightarrow t \in \fire{R}{}{\greenT}$.
For the negative version $\fire{R}{}{\redF}$, all variables are bound to associated domains $\domA$
and it is explicitly checked that $\dlNeg R(\vec{X})$ holds.
Finally, $\fire{R}{}{\redF/\greenT}$ uses $\fire{R}{}{\greenT}$ and $\fire{R}{}{\redF}$ to determine whether the tuple exists in \rel{R}.
Since these rules are correct, it follows that $\fire{R}{}{\redF/\greenT}$ is correct.
The positive firing rule for the rule $r$ ($\fire{r}{}{\greenT}$) is correct since its body only contains positive and negative EDB firing rules ($\fire{R}{}{\greenT}$ and $\fire{R}{}{\redF}$, respectively) which are already known to be correct.
The correctness of the positive firing version of  a rule's head predicate ($\fire{Q}{}{\greenT}$) follows naturally from the correctness of $\fire{r}{}{\greenT}$.
The negative version of the rule $\fire{r}{}{\redF}(\vec{X},\vec{V})$ contains an additional goal (i.e., $\dlNeg \rel{Q}(\vec{X})$)
and uses the firing version $\fire{R}{}{\redF/\greenT}$ to return only bindings for failed derivations.
For a head predicate with annotation $\redF$, we create two firing rules ($\fire{Q}{r}{\greenT}$ and $\fire{Q}{r}{\redF}$).
The rule $\fire{Q}{}{\greenT}$ was already proven to be correct.
$\fire{Q}{}{\redF}$ is also correct, because it contains only $\fire{Q}{}{\greenT}$ and domain queries in the body which were already proven to be correct.

\mypara{2) Inductive Step}
It remains to be shown that firing rules for programs of depth $n+1$ are correct.
Assume that firing rules for programs of depth up to $n$ are correct.
Let $r$ be a firing rule of depth $n+1$ in a program of depth $n+1$.
It follows that $\max_{R \in \bodyOf{r}} \depthP{R} \leq n$,
otherwise $r$ would be of a depth larger than $n+1$. Based on the induction hypothesis, it is guaranteed that the firing rules for all these predicates are correct.
Using the same argument as in the base case, it follows that the firing rule for $r$ is correct.
}
\end{proof}

\iftechreport{
\begin{algorithm}[t]
{\small
  \begin{algorithmic}[1]
    \Procedure{AddConnectivityRules}{$\fireProg{P}$, $\aProvQ$}
    \State $\fireCProg{P} \gets []$
    \State $Q(t) \gets \qPattern(\aProvQ)$
    \State $paths \gets \Call{pathStartingIn}{\fireProg{P}, Q(t)}$
    \ForAll {$p \in paths$}
       \ForAll {$e =(r_i(\vec{X_1})^{\adornment_1},r_j(\vec{X_2})^{\adornment_2}) \in p$}
           \State $goals \gets \Call{getMatchingGoals}{e}$
           \ForAll {$g_k \in goals$}
             \State $g_{new} \gets \Call{unifyHead}{\fire{r_i}{R_1}{\adornment_1}(t_1), g_k, \fire{r_j}{R_2}{\adornment_2}(t_2)}$
             \State $r_{new} \gets \fireC{r_j}{}{\adornment_2}{r_i^k}(t_2) \dlImp \bodyOf{\fire{r_j}{R_2}{\adornment_2}(t_2)}, g_{new}$
             \State $\fireCProg{P} \gets \fireCProg{P} \listconcat r_{new}$
           \EndFor
       \EndFor
    \EndFor
    \State \Return $\fireCProg{P}$
\EndProcedure
  \end{algorithmic}
}
  \caption{Add Connectivity Joins}\label{alg:connectivity-joins}
\end{algorithm}

 }

\begin{figure}[t]
  $\,$\\[-5mm]
  \centering
  \begin{minipage}{.8\linewidth}\centering
  \begin{align*}
\fire{Q}{}{\greenT}(n,s) &\dlImp \fire{r_1}{Q}{\greenT}(n,s,Z)\\
\fire{r_1}{Q}{\greenT}(n,s,Z) &\dlImp \fire{T}{}{\greenT}(n,Z), \fire{T}{}{\greenT}(Z,s), \fire{T}{}{\redF}(n,s)\\
\fireC{r_2}{}{\greenT}{r_1^1}(n,Z) &\dlImp \rel{T}(n,Z), \fire{r_1}{Q}{\greenT}(n,s,Z)\\
\fireC{r_2}{}{\greenT}{r_1^2}(Z,s) &\dlImp \rel{T}(Z,s), \fire{r_1}{Q}{\greenT}(n,s,Z)\\
\fireC{r_2}{}{\redF}{r_1^3}(n,s) &\dlImp \dlNeg \rel{T}(n,s), \fire{r_1}{Q}{\greenT}(n,s,Z)
  \end{align*}
  \end{minipage}
  \begin{minipage}{1\linewidth}
  \caption{Example firing rules with connectivity checks}\label{fig:exam-connect-joins}
 \end{minipage}
\end{figure}

\subsection{Connectivity Joins}\label{sec:connectivity-joins}

To be in the result of a firing rule is a necessary,
but not sufficient, condition for the corresponding rule node to be
connected to a node $Q(t') \in \qMatch(\aProvQ)$ in the explanation. Thus, we have to check connectivity of intermediate results explicitly.

\begin{Example}\label{ex:example6}
Consider the firing rules for
$\aProvQ_{n,s}$ shown in Fig.\,\ref{fig:exam-fire-rules}.
The corresponding rules with connectivity checks are shown in Fig.\,\ref{fig:exam-connect-joins}.
All rule nodes corresponding to
$\fire{r_1}{Q}{\greenT}(n,s,Z)$ are guaranteed to be connected to the node $\rel{Q}(n,s)$ (corresponding to the only atom in $\qMatch(\aProvQ_{n,s})$).
Note that connectivity joins are also required for negative firing rules
(e.g., $\fire{r_1}{Q}{\redF}(s,n,Z,V_1,V_2, \dlNeg V_3)$ in Fig.\,\ref{fig:exam-fire-negated}
is used for $\whynotq$). For sake of example, assume that instead of using $\rel{T}$, rule $r_1$ uses  an IDB relation $R$ which is computed using a rule $r_2: \rel{R}(X,Y) \dlImp \rel{T}(X,Y)$.
Consider the firing rule $\fire{r_2}{T}{\greenT}(n,Z) \dlImp \rel{T}(n,Z)$ created based on the $1^{st}$ goal of $r_1$.
Some provenance fragments computed by this rule may not be connected to $\rel{Q}(n,s)$.
A tuple node $\rel{R}(n,c)$ for a constant $c$ is only connected to the node $\rel{Q}(n,s)$ iff it is
part of a successful binding of $r_1$.
That is, for the node $\rel{R}(n,c)$, there has to exist a tuple $\rel{R}(c,s)$. Connectivity is achieved by adding the head of the firing rule for $r_1$ to the body of the firing rule for $r_2$ as shown in Fig.\,\ref{fig:exam-connect-joins}
(the $3^{rd}$ and $4^{th}$ rule).
\end{Example}

 Our algorithm
 traverses the query's rules starting from \provQ{} atom(s) to find all combinations of rules $r_i$ and $r_j$ such that the head of $r_j$ can be unified with a goal in $r_i$'s body.
 For each such pair $(r_i,r_j)$ where the head of $r_j$ corresponds to the $k^{th}$ goal in the body of $r_i$, we create a rule
 $\fireC{r_j}{}{\greenT}{r_i^k}(\vec{X})$ as follows. We unify the variables of the $k^{th}$
goal in the firing rule for $r_i$ with the head variables of the firing rule for $r_j$. All remaining variables of $r_i$ are renamed to avoid name clashes. We add the unified head of $r_i$ to the body of $r_j$.
These rules check whether rule nodes in the provenance graph are connected to nodes in $\qMatch(\aProvQ)$.

\iftechreport{
\begin{algorithm}[t]
{\small
  \begin{algorithmic}[1]
    \Procedure{CreateEdgeRelation}{$\fireCProg{P}$, $\aProvQ$}
      \State $\moveProg{P} \gets []$
      \State $Q(t) \gets \qPattern(\aProvQ)$
      \State $todo \gets [Q(t)]$
      \State $done \gets \{\}$
      \While {$todo \neq []$}
        \State $R(t)^{\adornment} \gets \Call{pop}{todo}$
        \State $done \gets \Call{insert}{done, R(t)^{\adornment}}$
        \State $rules \gets \Call{getRules}{R(t)^{\adornment}}$
        \ForAll {$r \in rules$}
          \State $args \gets \argsOf{\headOf{r}}$
          \If {$\Call{isEDB}{R}$}
            \If {$\adornment = \greenT$}
              \If {$\Call{isNegated}{g}$}
              \State $r_{g \to R} \gets \rel{edge}(\nodeSk{rel}{g}{\greenT}(t), \nodeSk{rel}{R}{\redF}(t)) \dlImp \fire{r}{R}{\greenT}(args)$
              \Else
              \State $r_{g \to R} \gets \rel{edge}(\nodeSk{rel}{g}{\greenT}(t), \nodeSk{rel}{R}{\greenT}(t)) \dlImp \fire{r}{R}{\greenT}(args)$
              \EndIf
           \Else
              \If {$\Call{isNegated}{g}$}
              \State $r_{g \to R} \gets \rel{edge}(\nodeSk{rel}{g}{\redF}(t), \nodeSk{rel}{R}{\greenT}(t)) \dlImp \fire{r}{R}{\redF}(args)$
              \Else
              \State $r_{g \to R} \gets \rel{edge}(\nodeSk{rel}{g}{\redF}(t), \nodeSk{rel}{R}{\redF}(t)) \dlImp \fire{r}{R}{\redF}(args)$
              \EndIf
	   \EndIf
           \State $\moveProg{P} \gets \moveProg{P} \listconcat r_{g \to R}$
          \Else
            \State $r_{new} \gets \rel{edge}(\nodeSk{rel}{\predOf{r}}{\adornment}(t), \nodeSk{rel}{r}{\adornment_{r}}(t,\ldots)) \dlImp \fire{r}{R}{\adornment}(t,\ldots)$
            \State $\moveProg{P} \gets \moveProg{P} \listconcat r_{new}$
            \ForAll {$g(t) \in \bodyOf{r}$}
              \If {$\Call{isNegated}{g}$}
              \State $\adornment' \gets \Call{switchState}{\adornment}$
              \Else
              \State $\adornment' \gets \adornment$
              \EndIf
              \State $todo \gets todo \listconcat g(t)^{\adornment'}$
              \If {$\adornment' = \greenT$}
              \State $r_{r \to g} \gets \rel{edge}(\nodeSk{rel}{r}{\greenT}(args), \nodeSk{rel}{g}{\greenT}(t)) \dlImp \fire{r}{R}{\greenT}(args)$
              \Else
              \State $r_{r \to g} \gets \rel{edge}(\nodeSk{rel}{r}{\redF}(args), \nodeSk{rel}{g}{\redF}(t)) \dlImp \fire{r}{R}{\redF}(args)$
              \EndIf
              \State $\moveProg{P} \gets \moveProg{P} \listconcat r_{r \to g}$             \EndFor
          \EndIf
        \EndFor
      \EndWhile
      \State \Return $\moveProg{P}$
    \EndProcedure
  \end{algorithmic}
}
  \caption{Create Edge Relation}\label{alg:create-move-relations}
\end{algorithm}

 }

\subsection{Computing the Edge Relation}\label{sec:comp-edge-relat}

The program created so far captures sufficient information for generating the edge
relation of the explanation for a $\provQ$ (which is used when rendering graphs). We make this step part of the program to offload this work to database backend.
To compute the edge relation, we use Skolem functions to create node
identifiers.
An identifier records the type of the node (tuple, rule, or
goal),  variables assignments, and the success/failure status of the
node, e.g., a tuple node $\rel{T}(n,s)$ that is successful would be represented as
$\nodeSk{rel}{T}{\greenT}(n,s)$. Each rule firing corresponds to a fragment of $\provGraph(P,I)$. For example, one such fragment is shown in Fig.\,\ref{fig:examp-graph-edge} (left).  Such a substructure is created through a set of rules: \vspace{-2mm}
\begin{itemize}\item One rule creating edges between tuple nodes for the head predicate and
      rule nodes
\item One rule for each goal connecting a rule node to       that goal node (only failed goals for
      failed rules)
\item One rule creating edges between each goal node and the corresponding EDB tuple node
\end{itemize}

\begin{figure}[t]
  \begin{minipage}{1.1\linewidth}
    $\,$\\[-6mm]
    \begin{minipage}{0.4\linewidth}
\resizebox{0.9\textwidth}{!}{\begin{tikzpicture}[>=latex',line join=bevel,line width=0.3mm]
  \definecolor{fillcolor}{rgb}{0.83,1.0,0.8};
  \definecolor{failcolor}{rgb}{0.63,0,0};

  \node (REL_Q_WON_a_c_) at (150bp,284bp) [draw=black,fill=fillcolor,ellipse] {$\boldsymbol{Q(n,s)}$};
  \node (RULE_0_LOST_a_c_b_) at (150bp,253bp) [draw=black,fill=fillcolor,rectangle] {$\boldsymbol{r_1(n,s,Z)}$};

  \node (GOAL_0_0_WON_a_b_) at (100bp,223bp) [draw=black,fill=fillcolor,rounded corners=.15cm,inner sep=3pt] {$\boldsymbol{g_{1}^{1}(n,Z)}$};
  \node (EDB_T_LOST_a_b_) at (93bp,193bp) [draw=black,fill=fillcolor,ellipse] {$\boldsymbol{T(n,Z)}$};

  \node (GOAL_0_1_WON_b_c_) at (150bp,223bp) [draw=black,fill=fillcolor,rounded corners=.15cm,inner sep=3pt] {$\boldsymbol{g_{1}^{2}(Z,s)}$};
  \node (EDB_T_LOST_b_c_) at (150bp,193bp) [draw=black,fill=fillcolor,ellipse] {$\boldsymbol{T(Z,s)}$};

  \node (GOAL_0_2_WON_a_c_) at (200bp,223bp) [draw=black,fill=fillcolor,rounded corners=.15cm,inner sep=3pt] {$\boldsymbol{g_{1}^{3}(n,s)}$};
  \node (EDB_T_LOST_a_c_) at (205bp,193bp) [draw=black,fill=failcolor,ellipse, text=white] {$\boldsymbol{T(n,s)}$};

  \draw [->] (RULE_0_LOST_a_c_b_) -> (GOAL_0_2_WON_a_c_);
  \draw [->] (GOAL_0_0_WON_a_b_) -> (EDB_T_LOST_a_b_);

  \draw [->] (RULE_0_LOST_a_c_b_) -> (GOAL_0_0_WON_a_b_);
 
  \draw [->] (GOAL_0_2_WON_a_c_) -> (EDB_T_LOST_a_c_);

  \draw[->] (GOAL_0_1_WON_b_c_) -> (EDB_T_LOST_b_c_);

  \draw [->] (REL_Q_WON_a_c_) -> (RULE_0_LOST_a_c_b_);
  \draw [->] (RULE_0_LOST_a_c_b_) -> (GOAL_0_1_WON_b_c_);
\end{tikzpicture}
 }
  \end{minipage} \hspace{-10mm}
  \begin{minipage}{0.38\linewidth}
   \small \vspace{3mm}
    \begin{align*}
      \rel{edge}(\nodeSk{rel}{Q}{\greenT}(n,s), \nodeSk{rule}{r_1}{\greenT}(n,s,Z)) &\dlImp \fire{r_1}{only2hop}{\greenT}(n,s,Z)\\       \rel{edge}(\nodeSk{rule}{r_1}{\greenT}(n,s,Z), \nodeSk{rel}{g_1^1}{\greenT}(n,Z)) &\dlImp \fire{r_1}{only2hop}{\greenT}(n,s,Z)\\       \rel{edge}(\nodeSk{rel}{g_1^1}{\greenT}(n,Z), \nodeSk{rel}{T}{\greenT}(n,Z)) &\dlImp \fire{r_1}{only2hop}{\greenT}(n,s,Z)\\       \rel{edge}(\nodeSk{rel}{g_1^3}{\greenT}(n,s), \nodeSk{rel}{T}{\redF}(n,s)) &\dlImp \fire{r_1}{only2hop}{\greenT}(n,s,Z)\\     \end{align*}  \end{minipage}
 \end{minipage}
 $\,$\\[-12mm]
 \caption{Fragment of an explanation corresponding to a derivation of rule $r_1$ (left) and the rules generating the edge relation for such a fragment (right)}\label{fig:examp-graph-edge}
\end{figure}
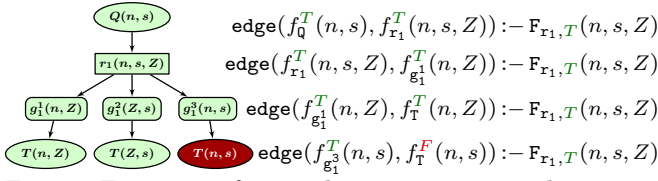

\begin{Example}
  Consider the firing rules with connectivity joins from Example\,\ref{ex:example6}.   Some of the rules for creating the edge relation of the explanation sought by
  the user are shown in Fig.\,\ref{fig:examp-graph-edge} (right).
  For example, each edge connecting the tuple node $\rel{Q}(n,s)$   to a successful rule node $r_1(n,s,Z)$ is created by the
  top-most rule,
  and the $2^{nd}$ rule creates an edge
  between   $r_1(n,s,Z)$ and   $g_1^1(n,Z)$.   Edges for failed derivations are created by considering the corresponding node identifiers   and a failure pattern (e.g., $\fire{r_1}{Q}{\redF}(s,n,Z,V_1,V_2, \dlNeg V_3)$).
\end{Example}

\iftechreport{
\begin{algorithm}[t]
{\small
  \begin{algorithmic}[1]
    \Procedure{CreateLinEdgeRelation}{$\fireCProg{P}$, $\aProvQ$}
    \State $\moveProg{P} \gets []$
    \State $Q(t) \gets \qPattern(\aProvQ)$
      \State $todo \gets [Q(t)]$
      \State $done \gets \{\}$
      \While {$todo \neq []$}
        \State $R(t)^{\adornment} \gets \Call{pop}{todo}$
        \If {$R(t)^{\adornment} \in done$}
          \State continue
        \EndIf
        \State $done \gets \Call{insert}{done, R(t)^{\adornment}}$
        \State $rules \gets \Call{getRules}{R(t)^{\adornment}}$
        \ForAll {$r \in rules$}
          \State $args \gets \argsOf{\headOf{r}}$
            \If {$\adornment = \greenT$}
              \If {$\Call{isNegated}{g}$}
              \State $r_{Q \to R} \gets \rel{edge}(\nodeSk{rel}{Q}{\greenT}(t), \nodeSk{rel}{R}{\redF}(t)) \dlImp \fire{r}{R}{\greenT}(args)$
              \Else
              \State $r_{Q \to R} \gets \rel{edge}(\nodeSk{rel}{Q}{\greenT}(t), \nodeSk{rel}{R}{\greenT}(t)) \dlImp \fire{r}{R}{\greenT}(args)$
              \EndIf
           \Else
              \If {$\Call{isNegated}{g}$}
              \State $r_{Q \to R} \gets \rel{edge}(\nodeSk{rel}{Q}{\redF}(t), \nodeSk{rel}{R}{\greenT}(t)) \dlImp \fire{r}{R}{\redF}(args)$
              \Else
              \State $r_{Q \to R} \gets \rel{edge}(\nodeSk{rel}{Q}{\redF}(t), \nodeSk{rel}{R}{\redF}(t)) \dlImp \fire{r}{R}{\redF}(args)$
              \EndIf
	   \EndIf
           \State $\moveProg{P} \gets \moveProg{P} \listconcat r_{Q \to R}$
        \EndFor
      \EndWhile
      \State \Return $\moveProg{P}$
    \EndProcedure
  \end{algorithmic}
}
  \caption{Create Lineage Edge Relation}\label{alg:create-move-relations-lin}
\end{algorithm}
 }
\subsection{$\SomeK$-Explanations}\label{sec:somek-explanations}

To compute one of the $\SomeK$-explanation types introduced in Sec.\,\ref{sec:expl-types}, we only have to
adapt the rules generating the edge relation.  As an example, we present the modifications for computing $\explainq_{\WhichProv}$ (e.g., Fig.\,\ref{fig:lin-model}). Recall that semiring $\WhichProv$ models provenance as a set of contributing tuples and we encode this as a graph by connecting a head of a rule derivation to the atoms in its body.
That is, for the $\explainq_{\WhichProv}$, we create only one type of rule that connects tuple nodes for the head predicate to EDB tuple nodes.
We use $\GPProg{P}{\aProvQ}{\WhichProv}$ to denote the program generated in this way for an input program $P$, and a \provQ{} $\aProvQ$.

\begin{Example}
Consider the graph fragment for $r_1$ in Fig.\,\ref{fig:examp-graph-edge} (left) without rule and goal nodes.
The rule that creates the edge between $\rel{Q}(n,s)$ and $\rel{T}(n,Z)$ is \vspace{-1mm}
$$\rel{edge}(\nodeSk{rel}{Q}{\greenT}(n,s), \nodeSk{rel}{T}{\greenT}(n,Z)) \dlImp \fire{r_1}{only2hop}{\greenT}(n,s,Z)$$
For each successful derivation of result $\rel{Q}(n,s)$ using rule $r_1$, a subgraph replacing $Z$ with bindings from the derivation is included in $\explainq_{\WhichProv}$.
\end{Example}

\subsection{Correctness}\label{sec:corr-proof-compl}

 We now prove that our approach is correct.

\begin{Theorem}\label{theo:alg-correctness}
Let $P$ be a program, $I$ be a instance, and $\aProvQ$ a \provQ. Program $\GPProg{P}{\aProvQ}{}$ evaluated over $I$ returns the edge relation of $\explainq(P,\aProvQ,I)$. \end{Theorem}
\begin{proof}
  To prove Theorem\,\ref{theo:alg-correctness}, we have to show that 1) only edges from $\provGraph(P,I)$ are in $\GPProg{P}{\aProvQ}{}(I)$ and 2) the program returns precisely the set of edges of
 explanation $\explainq(P,\aProvQ,I)$. \ifnottechreport{
The full proof is presented in~\cite{LS17} and in our accompanying report~\cite{LL18}.
}
\iftechreport{
\mypara{1}
The constants used as variable binding by the rules creating edges in $\GPProg{P}{\aProvQ}{}$ are either constants that occur in the \provQ~$\aProvQ$ or the result of rules which are evaluated over the instance $I$.
Since only the rules for creating the edge relation create new values (through Skolem functions), it follows that any constant used in constructing a node argument exists in the associated domain. Recall that the $\provGraph(P,I)$ only  contains nodes with arguments from the associated domain. Any edge returned by $\GPProg{P}{\aProvQ}{}$ is strictly based on the structure of the input program and connects nodes that agree on variable bindings.
Thus, each edge produced by $\GPProg{P}{\aProvQ}{}$ will be contained in $\provGraph(P,I)$.

\mypara{2}
We now prove that the program $\GPProg{P}{\aProvQ}{}$ returns precisely the set of edges of $\explainq(P,\aProvQ,I)$.
Assume that the \provQ{}~$\aProvQ$ only uses constants (the extension to \provQ{}s which contain variables is immediate).
Consider a rule of an input program of depth $1$ (i.e., only EDB predicates in the rule body). For such a rule node to be connected to an atom $Q(t) \in \qMatch(\aProvQ)$, its head variables have to be bound to $t$ (guaranteed by the unification step in Sec.\,\ref{sec:unify-program-with}). Since the firing rules are known to be correct, this guarantees that exactly the rule nodes connected to the \provQ{} node are generated. The propagation of this unification to the firing rules for EDB predicates is correct, because only EDB nodes agreeing with this binding can be connected to such a rule node. However, propagating constants is not sufficient since the firing rule for an EDB predicate (e.g., $R$) may return irrelevant tuples, i.e., tuples that are not part of any rule derivations for $Q(t)$ (e.g., there may not exist EDB tuples for other goals in the rule which share variables with the particular goal using predicate $R$). This is checked by the connectivity joins (Sec.\,\ref{sec:connectivity-joins}). If a tuple is returned by a connected firing rule, then the corresponding node is guaranteed to be connected to at least one rule node deriving \provQ{}.
Note that this argument does not rely on the fact that predicates in the body of a rule are EDB predicates.  Thus, we can apply this argument in a proof by induction to show that, given that rules of depth up to $n$ only produce connected rule derivations, the same holds for rules of depth $n+1$.
}
\end{proof}

\begin{Theorem}\label{theo:simpl-correctness}
  Let $P$ be a positive program, $I$ be a database instance, and $\aProvQ$ a \provQ.   The result of program $\GPProg{P}{\aProvQ}{\WhichProv}$ is the edge relation of $\explainq_{\WhichProv}(P,\aProvQ,I)$. \end{Theorem}
\begin{proof}
We prove Theorem\,\ref{theo:simpl-correctness} by induction over the structure of a program as in the proof of Theorem\,\ref{theo:alg-firingcorrect}. \ifnottechreport{
The full proof is presented in our technical report~\cite{LL18}.
}
\iftechreport{

\mypara{1) Base Case}
Consider a program $P$ with depth $1$ and $P$ has a single IDB predicate $\rel{Q}$, i.e., only containing rules of the form $r_i: \rel{Q}(\vec{X}) \dlImp \rel{R_{1}^i}(\vec{X_{1}^i}), \ldots, \rel{R_{n_i}^i}(\vec{X_{n_i}^i})$, where each $\rel{R_j}$ is an EDB relation and all goals are positive.
According to~\cite{green2007provenance}, the semiring annotation of a tuple $\rel{Q}(t)$ in the result of a positive Datalog program is computed as a sum of products. This sum contains one monomial per successful rule derivation with head $\rel{Q}(t)$. Such a monomial is constructed by multiplying the annotations of the grounded goals in the rule derivation. Let $\nu$ denote a variable assignment corresponding to a rule derivation and $Val(r,\rel{Q}(t))$ the set of all variable assignment for the rule $r$ that yield $\rel{Q}(t)$. Since addition and multiplication are idempotent in $\WhichProv$,
the annotation of a result $\rel{Q}(t)$ is computed as below:
$$
\bigcup_i \hspace{2mm} \bigcup_{\nu \in Val(r_i,\rel{Q}(t))} \hspace{2mm} \bigcup_{j = 1}^{n_i} \nu(R_{j}^i(X_j^i))
$$
That is, the $\WhichProv$ expression of $\rel{Q}(t)$ contains the set of annotations of all tuples that appear in at least one successful derivation of $\rel{Q}(t)$. In $\explainq_{\WhichProv}$, the fact that a tuple $\rel{R}(t')$ belongs to the provenance of $\rel{Q}(t)$ is recorded as an edge from $\rel{Q}(t)$ to $\rel{R}(t')$. Thus, to prove that the generated $\explainq_{\WhichProv}$ graph correctly encodes $\WhichProv$, we have to show that such an edge exists for every goal of a successful rule derivation with $\rel{Q}(t)$ as the head.
In Theorem~\ref{theo:alg-firingcorrect}, we have proven that firing rules correctly determine existence of tuples and successful/failed derivations for the user's provenance question. The adapted algorithm creates one rule for every goal of a rule which returns an edge if the goal is part of a successful rule derivation (this is ensured by using firing rules).
Thus, an edge exists for every goal of a successful derivation of $\rel{Q}(t)$.

\mypara{2) Inductive Step}
Assume that the algorithm is correct for any program of depth less than $n$.
Consider the program $P$ with depth $n$ and a derivation of a rule $r$ of depth $n$ in this program.
Based on the induction hypothesis, we know that the $\WhichProv$ annotation for each atom in the body of the derivation is recorded correctly.
From Theorem~\ref{theo:alg-firingcorrect} and using the same argument as in the base case, the rules created for rule $r$ will generate edges that link the result tuple of rule $r$ to each atom in its body (the claim holds). }
\end{proof}

\begin{figure*}[t]
\centering\vspace{-4mm}
\begin{minipage}{.65\linewidth}
\subfloat[2hop queries ($r_3$ and $r_4$), rewriting ($r_5$, $r_5'$, $r_5''$) according to d-tree $\fTree_1$, and example database (graph)]{\label{fig:input_q_d}
  \begin{minipage}{1.0\linewidth}
    \centering
    \begin{minipage}{0.45\linewidth}
      \centering
  \begin{align*}
    &r_3 : \rel{Q_{2hop}}(X) \dlImp \rel{H}(Y,L_1,Z), \rel{H}(Z,L_2,X)  \end{align*}\\[-2mm]
  \hrule
  $\,$\\[-6mm]
  \begin{align*}
  &r_{4} : \rel{Q_{2hop-d}}() \dlImp \rel{H}(Y,L_1,Z), \rel{H}(Z,L_2,d)
  \end{align*}\\[-2mm]
  \hrule
  $\,$\\[-6mm]
  \begin{align*}
    r_5 : &\,\rel{Q_{2hop}}() \dlImp \rel{Q_{L_1}}(Z),\rel{Q_{L_2}}(Z) \\
     r_{5'} : &\,\rel{Q_{L_1}}(Z) \dlImp \rel{H}(Y,L_1,Z) \\
     r_{5''} : &\,\rel{Q_{L_2}}(Z) \dlImp \rel{H}(Z,L_2,d) 
  \end{align*}\\[-4mm]
\end{minipage}
\begin{minipage}{.5\linewidth}
  \begin{minipage}{0.5\linewidth}
    \centering
    \resizebox{1.1\columnwidth}{!}{\begin{minipage}{1.4\linewidth}
\begin{tikzpicture}[>=latex',
line join=bevel,
line width=0.4mm,
every node/.style={ellipse},
minimum height=4mm]

  \definecolor{fillcolor}{rgb}{0.0,0.0,0.0};
  \node (a) at (0bp,13bp) [draw=black,circle] {$a$};
  \node (b) at (0bp,-13bp) [draw=black,circle] {$b$};
  \node (c) at (35bp,0bp) [draw=black,circle] {$c$};
  \node (d) at (70bp,0bp) [draw=blue,circle] {$d$};

   \path[]
		(a) edge [in=120,out=30,->] node [above] {$l_1$} (c)
		(a) edge [in=170,out=320,->] node [above] {$l_2$} (c)
		(b) edge [in=180,out=40,->] node [below] {$l_3$} (c)
		(b) edge [in=230,out=340,->] node [below] {$l_4$} (c)
		(c) edge [in=160,out=30,->] node [above] {$l_5$} (d)
		(c) edge [in=200,out=330,->] node [below] {$l_6$} (d);

\end{tikzpicture}
\end{minipage}

 }
   \end{minipage} 
\begin{minipage}{.35\linewidth}
    \centering
    \scriptsize\centering
    \textbf{\normalsize Relation $\rel{H}$}\\[2mm]
    \begin{tabular}{|ccc|l}
       \thead {S} & \thead {L} & \thead{E} & \\       a & $l_1$ & c & $s_1$ \\
      a & $l_2$ & c & $s_2$ \\
      b & $l_3$ & c & $t_1$ \\
      b & $l_4$ & c & $t_2$ \\
      c & $l_5$ & d & $u_1$ \\
      c & $l_6$ & d & $u_2$ \\
      \cline{1-3}
    \end{tabular}
  \end{minipage}
\end{minipage}
    \end{minipage}
}
 \end{minipage}
\begin{minipage}{.33\linewidth}\vspace{3mm}
\subfloat[Factorized representation ($r_5$, $r_{5'}$, $r_{5''}$)]{\label{fig:pg_expr}
\begin{minipage}{1\linewidth}
  \centering
  \begin{minipage}{1\linewidth}
    \centering
    \resizebox{0.61\columnwidth}{!}{\begin{tikzpicture}[>=latex',line join=bevel,line width=0.3mm]

  \definecolor{fillcolor}{rgb}{0.83,1.0,0.8};

\node (head) at (0.5,1) [draw=black,fill=fillcolor,circle] {$\boldsymbol{+}$};
\node (times) at (0.5,0) [draw=black,fill=fillcolor,rectangle] {$\boldsymbol{\cdot}$};

\node (r) at (-1,-1) [draw=black,fill=fillcolor,ellipse] {$\boldsymbol{+}$};
\node (t) at (2,-1) [draw=black,fill=fillcolor,ellipse] {$\boldsymbol{+}$};

\node (r1t) at (-2.5,-2) [draw=black,fill=fillcolor,rectangle] {$\boldsymbol{\cdot}$};
\node (r2t) at (-1.5,-2) [draw=black,fill=fillcolor,rectangle] {$\boldsymbol{\cdot}$};
\node (s1t) at (-0.5,-2) [draw=black,fill=fillcolor,rectangle] {$\boldsymbol{\cdot}$};
\node (s2t) at (0.5,-2) [draw=black,fill=fillcolor,rectangle] {$\boldsymbol{\cdot}$};
\node (t1t) at (1.5,-2) [draw=black,fill=fillcolor,rectangle] {$\boldsymbol{\cdot}$};
\node (t2t) at (2.5,-2) [draw=black,fill=fillcolor,rectangle] {$\boldsymbol{\cdot}$};

\node (r1) at (-2.5,-3) [draw=black,fill=fillcolor,circle] {$\boldsymbol{s_1}$};
\node (r2) at (-1.5,-3) [draw=black,fill=fillcolor,circle] {$\boldsymbol{s_2}$};
\node (s1) at (-0.5,-3) [draw=black,fill=fillcolor,circle] {$\boldsymbol{t_1}$};
\node (s2) at (0.5,-3) [draw=black,fill=fillcolor,circle] {$\boldsymbol{t_2}$};
\node (t1) at (1.5,-3) [draw=black,fill=fillcolor,circle] {$\boldsymbol{u_1}$};
\node (t2) at (2.5,-3) [draw=black,fill=fillcolor,circle] {$\boldsymbol{u_2}$};

\draw [->] (head) -> (times);
\draw [->] (times) -> (r);
\draw [->] (times) -> (t);

\draw [->] (r) -> (r1t);
\draw [->] (r) -> (r2t);

\draw [->] (r) -> (s1t);
\draw [->] (r) -> (s2t);

\draw [->] (t) -> (t1t);
\draw [->] (t) -> (t2t);

\draw [->] (r1t) -> (r1);
\draw [->] (r2t) -> (r2);

\draw [->] (s1t) -> (s1);
\draw [->] (s2t) -> (s2);

\draw [->] (t1t) -> (t1);
\draw [->] (t2t) -> (t2);

\end{tikzpicture}

 }
  \end{minipage}\\[-3mm]
    \begin{minipage}{0.4\linewidth}
      \centering
   \begin{align*}
     (s_1 + s_2 + t_1 + t_2) \cdot (u_1 + u_2)
   \end{align*}
 \end{minipage}

 \end{minipage}
 }
\end{minipage}\\[2mm]
\begin{minipage}{.38\linewidth}
\centering\vspace{-1mm}
\subfloat[Two d-trees of $r_{4}$: $\Tau_1$ (left) and $\Tau_2$ (right) ]{\label{fig:f-trees}
  \begin{minipage}{1\linewidth}
    \centering
    \begin{minipage}{.45\linewidth}
      \centering
      \resizebox{1.2\columnwidth}{!}{\begin{minipage}{1.4\linewidth}
\begin{tikzpicture}[>=latex',
line join=bevel,
line width=0.4mm,
every node/.style={ellipse},
minimum height=4mm]

  \definecolor{fillcolor}{rgb}{0.0,0.0,0.0};
  \node (y) at (-20bp,-50bp) [label={[xshift=-0.7cm, yshift=-0.6cm] \scriptsize $\{Z,L_1\}$}] {$Y$};
  \node (z) at (0bp,0bp) [label={[xshift=-0.4cm, yshift=-0.6cm] \scriptsize $\{\}$}] {$Z$};
  \node (l1) at (-10bp,-25bp) [label={[xshift=-0.5cm, yshift=-0.6cm] \scriptsize $\{Z\}$}] {$L_1$};
  \node (l2) at (10bp,-25bp) [label={[xshift=0.6cm, yshift=-0.6cm] \scriptsize $\{Z\}$}] {$L_2$};

   \path[]
		(z) edge [-] (l1)
		(l1) edge [-] (y)
		(z) edge [-] (l2);

\end{tikzpicture}
\end{minipage}

 }
    \end{minipage}
    \begin{minipage}{.4\linewidth}
      \centering
      \resizebox{1.2\columnwidth}{!}{\begin{minipage}{1.4\linewidth}
\begin{tikzpicture}[>=latex',
line join=bevel,
line width=0.4mm,
every node/.style={ellipse},
minimum height=4mm]

  \definecolor{fillcolor}{rgb}{0.0,0.0,0.0};
  \node (y) at (0bp,0bp) [label={[xshift=-0.5cm, yshift=-0.6cm] \scriptsize $\{\}$}] {$Y$};
  \node (z) at (0bp,-25bp) [label={[xshift=-0.5cm, yshift=-0.6cm] \scriptsize $\{Y\}$}] {$Z$};
  \node (l1) at (-10bp,-50bp) [label={[xshift=-0.7cm, yshift=-0.6cm] \scriptsize $\{Y,Z\}$}] {$L_1$};
  \node (l2) at (10bp,-50bp) [label={[xshift=0.6cm, yshift=-0.6cm] \scriptsize $\{Z\}$}] {$L_2$};

   \path[]
		(z) edge [-] (l2)
		(y) edge [-] (z)
		(z) edge [-] (l1);

\end{tikzpicture}
\end{minipage}

 }
    \end{minipage}
  \end{minipage}
}
\end{minipage}
\begin{minipage}{.59\linewidth}\vspace{-1mm}
\subfloat[Flat representation ($r_4$)]{\label{fig:flat_expr}
\begin{minipage}{1\linewidth}
  \centering
   \resizebox{.93\columnwidth}{!}{\begin{tikzpicture}[>=latex',line join=bevel,line width=0.3mm]

  \definecolor{fillcolor}{rgb}{0.83,1.0,0.8};
  \node (pl) at (-3,0) [draw=black,fill=fillcolor,circle] {$\boldsymbol{+}$};

\node (s1u1) at (-10,-1) [draw=black,fill=fillcolor,rectangle] {$\boldsymbol{\cdot}$};
\node (s1u2) at (-8,-1) [draw=black,fill=fillcolor,rectangle] {$\boldsymbol{\cdot}$};
\node (s2u1) at (-6,-1) [draw=black,fill=fillcolor,rectangle] {$\boldsymbol{\cdot}$};
\node (s2u2) at (-4,-1) [draw=black,fill=fillcolor,rectangle] {$\boldsymbol{\cdot}$};
\node (t1u1) at (-2,-1) [draw=black,fill=fillcolor,rectangle] {$\boldsymbol{\cdot}$};
\node (t1u2) at (0,-1) [draw=black,fill=fillcolor,rectangle] {$\boldsymbol{\cdot}$};
\node (t2u1) at (2,-1) [draw=black,fill=fillcolor,rectangle] {$\boldsymbol{\cdot}$};
\node (t2u2) at (4,-1) [draw=black,fill=fillcolor,rectangle] {$\boldsymbol{\cdot}$};

\node (s11) at (-10.5,-2) [draw=black,fill=fillcolor,circle] {$\boldsymbol{s_1}$};
\node (u1s1) at (-9.5,-2) [draw=black,fill=fillcolor,circle] {$\boldsymbol{u_1}$};

\node (s12) at (-8.5,-2) [draw=black,fill=fillcolor,circle] {$\boldsymbol{s_1}$};
\node (u2s1) at (-7.5,-2) [draw=black,fill=fillcolor,circle] {$\boldsymbol{u_2}$};

\node (s21) at (-6.5,-2) [draw=black,fill=fillcolor,circle] {$\boldsymbol{s_2}$};
\node (u1s2) at (-5.5,-2) [draw=black,fill=fillcolor,circle] {$\boldsymbol{u_1}$};

\node (s22) at (-4.5,-2) [draw=black,fill=fillcolor,circle] {$\boldsymbol{s_2}$};
\node (u2s2) at (-3.5,-2) [draw=black,fill=fillcolor,circle] {$\boldsymbol{u_2}$};

\node (t11) at (-2.5,-2) [draw=black,fill=fillcolor,circle] {$\boldsymbol{t_1}$};
\node (u1t1) at (-1.5,-2) [draw=black,fill=fillcolor,circle] {$\boldsymbol{u_1}$};

\node (t12) at (-0.5,-2) [draw=black,fill=fillcolor,circle] {$\boldsymbol{t_1}$};
\node (u2t1) at (0.5,-2) [draw=black,fill=fillcolor,circle] {$\boldsymbol{u_2}$};

\node (t21) at (1.5,-2) [draw=black,fill=fillcolor,circle] {$\boldsymbol{t_2}$};
\node (u1t2) at (2.5,-2) [draw=black,fill=fillcolor,circle] {$\boldsymbol{u_1}$};

\node (t22) at (3.5,-2) [draw=black,fill=fillcolor,circle] {$\boldsymbol{t_2}$};
\node (u2t2) at (4.5,-2) [draw=black,fill=fillcolor,circle] {$\boldsymbol{u_2}$};

\draw [->] (pl) -> (s1u1);
\draw [->] (pl) -> (s1u2);
\draw [->] (pl) -> (s2u1);
\draw [->] (pl) -> (s2u2);
\draw [->] (pl) -> (t1u1);
\draw [->] (pl) -> (t1u2);
\draw [->] (pl) -> (t2u1);
\draw [->] (pl) -> (t2u2);

\draw [->] (s1u1) -> (s11);
\draw [->] (s1u1) -> (u1s1);
\draw [->] (s1u2) -> (s12);
\draw [->] (s1u2) -> (u2s1);

\draw [->] (s2u1) -> (s21);
\draw [->] (s2u1) -> (u1s2);
\draw [->] (s2u2) -> (s22);
\draw [->] (s2u2) -> (u2s2);

\draw [->] (t1u1) -> (t11);
\draw [->] (t1u1) -> (u1t1);
\draw [->] (t1u2) -> (t12);
\draw [->] (t1u2) -> (u2t1);

\draw [->] (t2u1) -> (t21);
\draw [->] (t2u1) -> (u1t2);
\draw [->] (t2u2) -> (t22);
\draw [->] (t2u2) -> (u2t2);

\end{tikzpicture}

 }
   \begin{minipage}{.5\linewidth}
   \vspace{-5mm}
   \centering
   \begin{align*}
     s_1 \cdot u_1 + s_1 \cdot u_2 + s_2 \cdot u_1 + s_2 \cdot u_2 + t_1 \cdot u_1 + t_1 \cdot u_2 + t_2 \cdot u_1 + t_2 \cdot u_2
   \end{align*}
   \end{minipage}
 \end{minipage}
 }
\end{minipage}
   $\,$\\[-3mm]
   \caption{Factorized and flat provenance graphs ($\ProvPoly$)  explaining $\whyq \rel{Q_{2hop}}(d)$ and two d-trees for $r_4$.}
   \label{fig:exam-pg-flat}
\end{figure*}

\section{Factorization}
\label{sec:factorize}

For provenance polynomials, we can exploit the distributivity law of semirings to generate 
  factorizations  of provenance~\cite{OZ12} which are exponentially more concise in the best case.
For instance, consider a query $r_3$ returning the end points of paths of length 2 evaluated over the edge-labelled graph in Fig.\,\ref{fig:input_q_d}. 
The  provenance polynomial for the query result $\rel{Q_{2hop}}(d)$ using the annotations from Fig.\,\ref{fig:input_q_d} is shown in Fig.\,\ref{fig:flat_expr}. Each monomial in the polynomial corresponds to one of the derivations of the result using $r_3$. Each of these $2 \cdot (2^2)$ 
(we have two options as starting points and, for each hop, we have two options)  derivations corresponds to one path of length $2$ ending in $d$.
When generating provenance graphs for provenance polynomials, we create ``$\cdot$'' nodes for rule derivations and ``$+$'' nodes for IDB tuples.
Fig.\,\ref{fig:pg_expr} is the factorized representation of this polynomial. We can exploit the fact that our approach shares common subexpressions to produce a particular factorization. 
This is achieved by rewriting the input program to partition a query by materializing joins and projections as new IDB relations which can then be shared. 
We first review f-trees and d-trees as introduced in~\cite{OZ15}
which encode possible nesting ``schemas'' for factorized representations of provenance (or query results), the size bounds for factorized representations based on d-trees proven in~\cite{OZ15}, and how to chose a d-tree for a query that results in the optimal worst-case size bound for the factorized representation of the provenance according to this d-tree.
Then, we introduce a query transformation for conjunctive queries which, given an input query and the d-tree for this query,  generates a rewritten query which returns a provenance graph factorized corresponding to this d-tree. We employ this rewriting to produce more concise provenance in PUG (experiments are shown in Sec.\,\ref{sec:experiments-fact}).

\mypara{Factorized Representations} \label{sec:f-rep}
In~\cite{OZ12,OZ15}, a \textit{factorized representation} (f-rep for short) of a relation is defined as an algebraic expression constructed using singleton relations (one tuple with one value) and the relational operators union and product. Any f-rep over a set of attributes from a schema $S$ can be interpreted as a relation over $S$ by evaluating the algebraic expression, e.g., $\{(a)\} \times (\{(b)\} \union \{(c)\})$ is a factorized representation of the relation $\{(a,b), (a,c)\}$. Following the convention from~\cite{OZ12}, we denote a singleton $\{(a)\}$ as $a$.
Factorization can be applied to compactly represent relations and query results as well as provenance (e.g., Fig.\,\ref{fig:pg_expr}). We will factorize representations of provenance 
which  encode variables of provenance polynomials as the tuples annotated by these variables and show how to extract provenance polynomials from provenance graphs generated in this way.

\mypara{F-trees for F-reps}
\label{sec:f-tree}
Olteanu et al.~\cite{OZ15} introduce \textit{f-trees} to encode the nesting
structure of f-reps.  At first, let us consider only f-trees which encode the
nesting structure of a boolean query~\cite{OZ12}.  An f-tree for a boolean query $Q$ (e.g., $r_4$ in Fig.\,\ref{fig:input_q_d}) is a rooted
forest with one node for every variable of $Q$.\footnote{In~\cite{OZ15},
  relational algebra is used to express queries and nodes of f-trees represent
  equivalence classes of attributes which in Datalog correspond to query
  variables.}  An f-rep according to an f-tree $\fTree$ nests values according
to $\fTree$: a node labelled with $X$ corresponds to a union of values from the
attributes bound to $X$ by the query. The values of attributes bound to children
of a node $X$ corresponding to a single value $x$ bound to $X$ are grouped under
$x$. If a node has multiple children, then their f-reps are connected via
$\times$. For example, consider an f-tree $\fTree$ with root $X$ and a single
child $Y$ for a query $\rel{Q}() \dlImp \rel{R}(X,Y)$. An f-rep
of $\rel{Q}$ according to $\fTree$ would be of the form
$x_1 \times (y_{1_1} \cup \ldots \cup y_{n_1}) \cup \ldots \cup x_m \times
(y_{1_m} \cup \ldots \cup y_{n_m})$, i.e., the $Y$ values co-occurring with a
given $X$ value $x$ are grouped as a union and then paired with $x$. An f-tree
encodes (conditional) independence of the variables of a query in the sense that the
values of one variable do not depend on the values of another variable. For
instance, two siblings $X$ and $Y$ in an f-tree have to be independent since a
union of $X$ values is paired (cross-product) with a union of $Y$ values. This is only correct if the values of $X$ and $Y$ are independent.
The independence assumptions encoded in an f-tree may not hold for every possible query with the same schema as the f-tree.
Thus, only some f-trees with a particular schema may be applicable for a query with this schema.
It was shown in~\cite{OZ15}, that a query has an f-rep over
an f-tree $\fTree$ for any database iff for each relation in $Q$ the variables assigned to attributes of this relation
(these variables are called dependent) are on the same root-to-leaf path in the
f-tree. This is called the \textit{path condition}.  Note that multiple references to the
same relation in a query are considered as separate relations when checking this
condition.  For instance, consider the boolean query $r_{4}$ in Fig.\,\ref{fig:input_q_d} which checks if there are paths of length 2
ending in the node $d$.  Fig.\,\ref{fig:f-trees} shows two f-trees $\fTree_1$
and $\fTree_2$ for this query (ignore the sets on the side of nodes for now). An f-rep according to $\fTree_2$ for $r_4$ would
encode a union of $Y$ values paired ($\times$) with a union of $Z$ values for this $Y$
value. Each $Z$ value nested under a $Y$ value is then paired with a
cross-product of $L_1$ and $L_2$ values.  

\mypara{D-trees for D-reps}
\label{sec:d-tree}
The size of a factorized representation can be further reduced by allowing subexpressions to be shared through definitions, i.e., using algebra graphs instead of trees. In~\cite{OZ15}, such representations are called \textit{d-representations} (d-rep).
Analogous to how f-trees define the
structure of f-reps, d-trees were introduced to define the structure of d-reps.
A d-tree is an f-tree where each node $X$ is annotated with a set $key(X)$, a subset of its ancestors in the f-tree on which the node and any of its dependents depend on. The f-rep of the subtree rooted in $X$ is unique for each combination of values from $key(X)$. That is, if $key(X)$ is a strict subset of the ancestors of $X$, then the same d-rep for the subtree at $X$ can be shared by multiple ancestors, reducing the size of the representation. In Fig.\,\ref{fig:f-trees}, the set $key$ is shown beside each node, e.g., in $\fTree_2$, the variable $L_2$ depends only on $Z$, but not on $Y$. 
An important result proven in~\cite{OZ15} is that, for a given d-tree $\fTree$ for a query $Q$, the size of d-rep of $Q$ over a database $I$ is bound by $\card{I}^{s^{\uparrow}(\fTree)}$ where $s^{\uparrow}(\fTree)$ is a rational number computed based on $\fTree$ alone (see~\cite{OZ15} for details of how to compute $s^{\uparrow}(\fTree)$). This bound can be used to determine the d-tree for a query $Q$ which will yield the d-rep of  worst-case optimal size by enumerating the valid d-trees for $Q$ and, then, chosing the one with the lowest value of $s^{\uparrow}$.

\begin{Example}
  Consider the d-rep for $r_4$ (Fig.\,\ref{fig:input_q_d}) 
over the example instance of relation $\rel{H}$ (Fig.\,\ref{fig:input_q_d}) according to d-tree $\fTree_2$ (Fig.\,\ref{fig:f-trees}). 
Variable $Y$ at the root of $\fTree_2$ is bound to the attribute $S$ from the first reference of $\rel{H}$, i.e., the starting point of paths of length $2$ ending in $d$. There are two such starting points $a$ and $b$. Now each of these are paired with the only valid intermediate node $c$ on these paths (variable $Z$). Finally, for this node, we compute the cross-product of the $L_1$ and $L_2$ values connected to $c$. Since the $L_2$ values only depend on $Z$, we share these values when the same $Z$ value is paired with multiple $Y$ values. The final result is $(a \times c \times (l_1 \cup l_2) \times l^{\uparrow}) + (b \times c \times (l_3 \cup l_4) \times l^{\uparrow})$ where $l^{\uparrow} \defas (l_5 \cup l_6)$.
\end{Example}

\mypara{Factorization of Provenance}
\label{sec:f-rep-prov}
For the provenance of a conjunctive query $Q$ that is not a boolean query, i.e., it has one or more variables in the head (e.g., $r_3$ in Fig.\,\ref{fig:input_q_d}), we have to compute a provenance polynomial for each result of $Q$.
We would like the factorization of the provenance of $Q$ to clearly associate the provenance polynomial of a result $t$ with the tuple $t$. That is, we want to avoid factorizations where head variables of $Q$ are nested below variables that store provenance (appear only in the body) since reconstructing the provenance polynomial for $t$ would require enumeration of the full provenance from the factorized representation in the worst case. For example, consider a query with head variable
$X$ and body variable $Y$. If $Y$ is the root of a d-tree $\fTree$, then the d-rep of $Q$ according to $\fTree$ would be of the form $y_1 \times (x_{1_1} + \ldots + x_{n_1}) + \ldots + y_m \times (x_{1_m} + \ldots + x_{n_m})$. To extract the provenance polynomial for a result $x_i$, we may have to traverse all $y$ values since there is no indication, for which $y$ values, $x_i$ appears in the sum $x_{1_i} + \ldots + x_{n_i}$. 
We ensure this by constructing d-trees which do not include the head variables, but treat those as ancestors of every node in the d-tree when computing $key$ for the nodes. 
For instance, to make  $\fTree_1$ (Fig.\,\ref{fig:f-trees}) a valid d-tree for capturing the provenance of $r_3$ (Fig.\,\ref{fig:input_q_d}), we treat the head variable $X$ as a virtual ancestor of all nodes and get $key(Z) = \{X\}$ and $key(L_2) = \{Z,X\}$.
Furthermore, if we are computing an explanation to a provenance question (PQ) $\aProvQ$ that binds one or more head variables to constants, then we can propagate these bindings before constructing a d-tree for the query. For example, to explain $\rel{Q_{2hop}}(d)$, we would propagate the binding $X=d$ resulting in rule $r_4$ (Fig.\,\ref{fig:input_q_d}). Thus, any d-tree for $r_4$ can be used to create a factorized $\explainq_{\ProvPoly}$ graph for the user question $\whyq(\rel{Q_{2hop}}(d))$.

\mypara{Rewriting Queries for Factorization}
\label{sec:trans-query}
We now explain how, given a d-tree $\fTree$ for a conjunctive query $Q$ and positive PQ $\aProvQ \defas \whyq Q(t)$, to generate a Datalog query $Q_{rewr}$ such that, for any database $I$, we have that $\explainq_{\ProvPoly}(Q_{rewr}, \aProvQ, I)$ encodes $\ProvPoly(Q_{rewr}, I, t)$ for each $t \in \qMatch(\aProvQ)$ factorized according to $\fTree$. 
We first unify the query with the PQ as described in Sec.\,\ref{sec:unify-program-with}.
Given a unified input query $Q$ and a d-tree $\fTree$, we compute $Q_{rewr}$ as follows.

\begin{enumerate}
\item Assume a total order among the variables of $Q$ (e.g., the lexicographical order). For every node $X$ with children $Y_1$, \ldots, $Y_n$ in the d-tree $\fTree$, we generate $$r_X: \rel{Q_{X}}(key(X)) \dlImp \rel{Q_{Y_1}}(key(Y_1)), \ldots, \rel{Q_{Y_n}}(key(Y_n))$$
\item Now for every atom $\rel{R}(Z_1, \ldots, Z_m)$ in the body of $Q$, we find the shortest path starting in a root node that contains all nodes $Z_1$ to $Z_m$. Let $Y = Z_i$ for some $i$ be the last node on this path. Then, we add atom $\rel{R}(Z_1, \ldots, Z_m)$ to the body of rule $r_Y$ created in the previous step.
\item Let $X_1$, \ldots, $X_n$ be the roots of the d-tree $\fTree$ (being a forest, a d-tree may have multiple roots). Furthermore, let $Y_1$, \ldots, $Y_m$ denote the head variables of the unified input query $Q$ with the PQ. We create $$r_Q: \rel{Q}(Y_1, \ldots, Y_m) \dlImp \rel{Q_{X_1}}(key(X_1)), \ldots, \rel{Q_{X_n}}(key(X_n))$$
\end{enumerate}

The rewriting above creates a factorization according to a d-tree $\fTree$. However, it may contain rules which cannot potentially lead to reuse and, thus, result in overhead that could be avoided if we were able to identify such rules. We now present an optimization that removes such rules to further reduce the size of the generated provenance graphs.  
  Consider two nodes $X$ and $Y$ in a d-tree where $Y$ is the only child of $X$, i.e., $key(Y) = key(X) \cup \{X\}$. We would generate rules

\begin{minipage}{1\linewidth}
\centering\vspace{-2mm}
\begin{align*}
&r_X: \rel{Q_X}(key(X)) \dlImp \rel{Q_Y}(X \cup key(X))\\ &r_Y: \rel{Q_Y}(X \cup key(X)) \dlImp \ldots\\[-2mm]
\end{align*}
\end{minipage}

In this case, the intermediate result $Q_Y$ does not lead to further factorization (we have a union of unions). Thus, we can merge the rules by substituting the atom $\rel{Q_Y}(X \cup key(X))$ in $r_X$ with the body of $r_Y$. A similar situation may arise with the rule $r_Q$ deriving the final query result. In general, we can merge any rule of the form $\rel{Q_1}(X_1, \ldots, X_n) \dlImp \rel{Q_2}(X_1, \ldots, X_n)$ with the rule deriving $\rel{Q_2}$ (in our translation, there will be exactly one rule with head $\rel{Q_2}$).

\begin{Example}\label{ex:d-tree-rewrite}
Consider the question $\whyq \rel{Q_{2hop}}(d)$ over the query $r_3$ from Fig.\,\ref{fig:input_q_d}. Unifying the query with this question yields $r_4$ (below $r_3$ in the same figure). To rewrite the query according to
the d-tree $\Tau_1$ from Fig.\,\ref{fig:f-trees}, we apply the above algorithm to create rules:

\begin{minipage}{1\linewidth}
\vspace{-2mm}\hspace{-3mm}
\begin{minipage}{.9\linewidth}
\centering\small
\begin{align*}
  &r_{Q_{2hop}}: \rel{Q_{2hop}}() \dlImp \rel{Q_{Z}()} &
  &r_{Z}: \rel{Q_{Z}()} \dlImp \rel{Q_{L_1}}(Z), \rel{Q_{L_2}}(Z)\\
  &r_{L_1}: \rel{Q_{L_1}}(Z) \dlImp \rel{Q_{Y}}(Z,L_1) &
  &r_{Y}: \rel{Q_{Y}}(Z,L_1) \dlImp \rel{H}(Y,L_1,Z)\\
  &r_{L_2}: \rel{Q_{L_2}}(Z) \dlImp \rel{H}(Z,L_2,d)\\[-2mm]
\end{align*}
\end{minipage}
\end{minipage}
Applying the optimizations introduced above, we merge the rules $r_{Q_{2hop}}$ with $r_{Z}$ (the head $\rel{Q}_Z$ is the body of $r_{Q_{2hop}}$). Since $key(Y) = key(L_1) \cup \{L_1\}$ and $L_1$ has only one child, we merge $r_Y$ into $r_{L_1}$. The resulting program is shown as rules $r_5$, $r_{5'}$ and $r_{5''}$ in Fig.\,\ref{fig:input_q_d}.

\end{Example}

\mypara{Factorized Explanations}
To generate a concise factorization of provenance for a PQ $\aProvQ$ over a conjunctive query $Q$, we first find a d-tree $\fTree$ with minimal $s^{\uparrow}$ among all d-trees for $Q$ (such a d-tree $\fTree$ guarantees worst-case optimal size bounds for the generated factorization). Then, we rewrite the input query according to $\fTree$ (explained above) and use the approach in Sec.\,\ref{sec:compute-gp} to generate $\explainq_{\ProvPoly}(Q_{rewr}, \aProvQ, I)$ encoding the d-rep of $\ProvPoly(Q_{rewr}, I, t)$ for each $t \in \qMatch(\aProvQ)$.

\begin{Example}
Continuing with Example~\ref{ex:d-tree-rewrite}, assume we compute the $\ProvPoly$ explanation using the rewritten query ($r_5$, $r_{5'}$, and $r_{5''}$). The result over the example database is shown in Fig.\,\ref{fig:pg_expr}. The top-most addition and multiplication correspond to the successful derivation using rule $r_5$ (using $c$ as an intermediate hop from some node to $d$). The left branch below the multiplication encodes the four possible derivations of $\rel{Q_{L_1}}(c)$ ($s_1 + s_2 + t_1 + t_2$) and the right branch corresponds to the two derivations of $\rel{Q_{L_2}}(c)$ ($u_1 + u_2$). The polynomial captured by this graph is $(s_1 + s_2 + t_1 + t_2) \cdot (u_1 + u_2)$. That is, there are 4 ways to reach $c$ from any starting node and two ways of reaching $d$ from $c$ leading to a total of $4 \cdot 2 = 8$ paths of length two ending in the node $d$. 
\end{Example}

 \section{Implementation}
\label{sec:transl-into-relat}

We have implemented the approach presented in this paper in a system called \textit{PUG} (Provenance Unification through Graphs).
PUG is an extension of GProM~\cite{AG14}, a middleware that executes provenance requests using a relational database backend (shown in Fig.\,\ref{fig:gprom-gp}). 
We have extended the system to support Datalog enriched with syntax for stating provenance questions. The user provides a why or why-not question and the corresponding Datalog query as an input. 
Our system  parses  and semantically analyzes this input. Schema
information is gathered by querying the catalog of the backend database (e.g.,
to determine whether an EDB predicate exists). Modules
for accessing schema information are already part of the GProM system, but a new
semantic analysis component had to be developed to support Datalog. 
The algorithms presented in
Sec.\,\ref{sec:compute-gp} are applied to create the program 
$\GPProg{P}{\aProvQ}{}$ for the input program $P$ and the provenance question $\psi$
which computes $\explainq(P,\aProvQ,I)$ (analogously, $\explainq_{\SomeK}(P,\aProvQ,I)$ for $\GPProg{P}{\aProvQ}{\SomeK}$). This program is then translated into  relational algebra ($\cal RA$). The
resulting algebra expression is translated into SQL and sent to the backend database to compute the edge relation of 
 the explanation for the question. Based on this edge relation, we render a provenance graph. For examples and installation guidelines see: \url{https://github.com/IITDBGroup/PUG}.
While it would certainly be possible to directly translate the
Datalog program into SQL without the intermediate translation into $\cal RA$, we choose to introduce this step to be able to leverage the existing
heuristic and cost-based optimizations for $\cal RA$ expressions provided by 
GProM~\cite{NK17} and use its library of $\cal RA$ to SQL translators. 
\begin{figure}[t]
  $\,$\\[-2mm]
  \centering
  \includegraphics[width=0.8\linewidth]{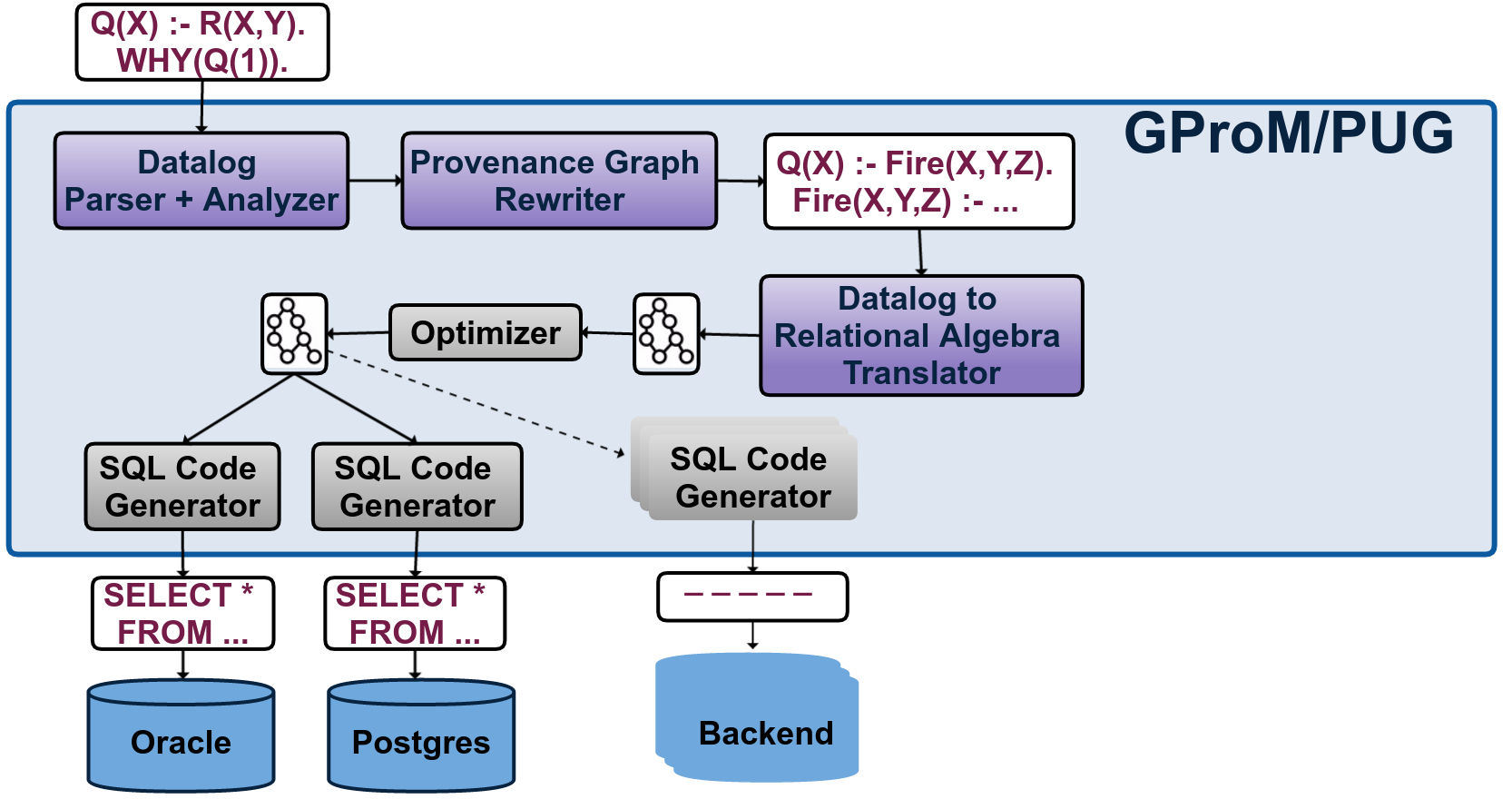}
  $\,$\\[-3mm]
  \caption{PUG implementation in GProM}
  \label{fig:gprom-gp}
\end{figure}
Our translation of first-order  queries (a program with a distinguished answer relation) to $\cal RA$ 
is mostly standard. See~\cite{LS17} for details and an example.

\section{Experiments}
\label{sec:experiments}

We evaluate the performance of our solution
over a co-author graph relation extracted from DBLP (\url{http://www.dblp.org})
as well as over the TPC-H benchmark dataset (\url{http://www.tpc.org/tpch/default.asp}).
We mainly evaluate three aspects;
1) we compare our approach for computing explanations ($\explainq$) with the approach introduced for provenance games~\cite{KL13}.  We call the provenance game approach \texttt{Direct Method (DM)}, because it
directly constructs the full provenance graph; 2) we compare our approach for  Lineage ($\explainq_{\WhichProv}$) to the language-integrated approach developed for the \textit{Links} programming language~\cite{FS17}; 3) we evaluate the performance impact of rewriting queries to produce factorized provenance (Sec.\,\ref{sec:factorize}).
We have created subsets of the DBLP dataset with 100, 1K, 10K, 100K, 1M,
and 8M co-author pairs (tuples).
For the TPC-H benchmark, we used database sizes 10MB, 100MB, 1GB, and 10GB.
All experiments were run on a machine with 2 x 3.3Ghz AMD Opteron 4238 CPUs (12 cores in total) and 128GB RAM running Oracle Linux 6.4.
We use the commercial DBMS X (name omitted due to licensing restrictions) and Postgres as a backend (DBMS X is the default). Unless stated otherwise, each experiment was repeated 100 times (we stopped executions that ran longer than 10 minutes) and we report the median runtime.
Computations that did
not finish within the allocated time are omitted from the graphs.

\begin{figure}[t]
 \centering
  \begin{minipage}{0.98\linewidth}
    \centering
    $\,$\\[-5mm]
    \begin{minipage}{0.82\linewidth}
    \centering\scriptsize
    \begin{align*}
      r_1: \rel{only2hop}(X,Y) &\dlImp \rel{DBLP}(X,Z), \rel{DBLP}(Z,Y), \dlNeg \rel{DBLP}(X,Y)\\[1mm] \hline \\[-3mm]
      r_2: \rel{XwithYnotZ}(X,Y) &\dlImp \rel{DBLP}(X,Y), \dlNeg \rel{Q_1}(X)\\
      r_{2'}: \rel{Q_1}(X) &\dlImp \rel{DBLP}(X,\text{`Svein Johannessen'})\\[1mm] \hline \\[-3mm]
      r_3: \rel{only3hop}(X,Y) &\dlImp \rel{DBLP}(X,A), \rel{DBLP}(A,B), \rel{DBLP}(B,Y),\\ & \mathtab\mathtab \dlNeg \rel{E_1}(X), \dlNeg \rel{E_2}(X) \\
      r_{3'}: \rel{E_1}(X) &\dlImp \rel{DBLP}(X,Y)\\       r_{3''}: \rel{E_2}(X) &\dlImp \rel{DBLP}(X,A), \rel{DBLP}(A,Y)       \\[1mm] \hline \\[-3mm]
      r_4: \rel{ordPriority}(X,Y) &\dlImp \rel{CUSTOMER}(A,X,B,C,D,E,F,G), \\ & \mathtab\mathtab \rel{ORDERS}(H,A,I,J,K,Y,M,N,O)\\[1mm] \hline
    \end{align*}\\[-10mm]
    \begin{align*}
      r_5: \rel{ordDisc}(X,Y) &\dlImp \rel{CUSTOMER}(A,X,B,C,D,E,F,G), \\
	& \hspace{-6mm} \rel{ORDERS} (H,A,I,J,K,L,M,O,P), \\
	& \hspace{-13mm} \rel{LINEITEM} (H,Q,R,S,T,U,V,Y,W,Z,A',B',C',D',E',F')\\[1mm] \hline \\[-3mm]
      r_6: \rel{partNotAsia}(X) &\dlImp \rel{PART}(A,X,B,C,D,E,F,G,H), \\
	& \hspace{-6mm} \rel{PARTSUPP} (A,I,J,K,L), \rel{SUPPLIER}(I,M,N,O,P,Q,R), \\
	& \hspace{-6mm} \rel{NATION} (O,S,T,U), \dlNeg \rel{R_1}(T,\text{`ASIA'})\\       r_{6'}: \rel{R_1}(T,Z) &\dlImp \rel{REGION}(T,Z,V)
      \\[1mm] \hline \\[-3mm]
      r_7: \rel{suppCust}(N) \dlImp &\rel{SUPPLIER}(A,B,C,N,D,E,F),\\
                               & \rel{CUSTOMER}(G,H,I,N,J,K,L,M) \\
    \end{align*}
    \end{minipage}
  \end{minipage}
  $\,$\\[-7mm]
  \caption{DBLP and TPC-H queries for experiments}
  \label{fig:experi-queries}
\end{figure}
\mypartitle{Workloads}
We compute explanations for the queries in Fig.\,\ref{fig:experi-queries}. For DBLP datasets, we consider: \rel{only2hop} ($r_1$) which is our running example query in this paper;
\rel{XwithYnotZ} ($r_2$) that returns authors that are direct co-authors of a certain person $Y$,
but not of ``Svein Johannessen'';
\rel{only3hop} ($r_3$) that returns pairs of authors $(X,Y)$ that are connected via a path of length 3 in the co-author graph where $X$ is not a co-author or indirect co-author (2 hops) of $Y$. For TPC-H, we consider: \rel{ordPriority} ($r_4$) which returns for each customer the priorities of her/his orders;
\rel{ordDisc} ($r_5$) which returns customers and the discount rates of items in their orders;
\rel{partNotAsia} ($r_6$) which finds parts that can be supplied from a country that is not in Asia;
\rel{suppCust} ($r_7$) returns nations having both suppliers and customers.

\mypartitle{Implementing DM}
\texttt{DM} has to instantiate a graph with ${\cal O}(\card{\adom{I}}^n)$ nodes
where $n$ is the maximal number of variables in a rule.
We do not have a full implementation of \texttt{DM}, but compute a conservative
lower bound for the runtime of the step constructing the game graph by executing  a query ($n$-way cross-product over the active domain).
Note that the actual runtime will be much higher
because 1) several edges are created for each rule binding (we underestimate the number of nodes of the constructed graph) and 2) recursive Datalog queries have to be evaluated over this graph using the well-founded semantics.
The results for different instance sizes and number of variables are shown in
Fig.\,\ref{tab:baseline}.
Even for only 2 variables, \texttt{DM} did not finish
for datasets of more than 10K tuples
within the allocated 10 min timeslot.
For queries with more than $4$ variables, \texttt{DM} did not even finish for the smallest dataset. \begin{figure}
  \begin{center}
	\scriptsize \centering $\,$\\[-1mm]
    \begin{tabular}{|l|c|c|c|c|} \hline
	\thead {DBLP (\#tuples)} & \thead {100} & \thead {1K} & \thead {10K} & \thead {100K}
		\\ \hline
		2 Variables ($r_2$) & 0.043 & 0.171 & 14.016 & -
    		\\ \hline
	        3 Variables ($r_1$) & 0.294 & 285.524 & - & -
    		\\ \hline
	        4 Variables ($r_3$) & 56.070 & - & - & -
		\\ \hline
	\thead {TPC-H (Size)} & \thead {10MB} & \thead {100MB} & \thead {1GB} & \thead {10GB}
    		\\ \hline
		$>$ 10 Variables  & \multirow{ 2}{*}{-} & \multirow{ 2}{*}{-} & \multirow{ 2}{*}{-} & \multirow{ 2}{*}{-} \\
                \hspace{2mm}($r_4, r_5, r_6, r_7$) & & & &
    		\\ \hline
    \end{tabular}
  \end{center}
  $\,$\\[-9mm]
  \caption{Runtime of \texttt{DM} in seconds. For entries with `-', the computation did not finish within 10 min.}
  \label{tab:baseline}
\end{figure}

\begin{figure}[t]
\begin{minipage}{0.98\linewidth}
$\,$\\[-11mm]
\scriptsize\centering
\subfloat[\small Runtime of \rel{only2hop}]
{
  \includegraphics[width=0.50\columnwidth,trim=0 60 0 0, clip]{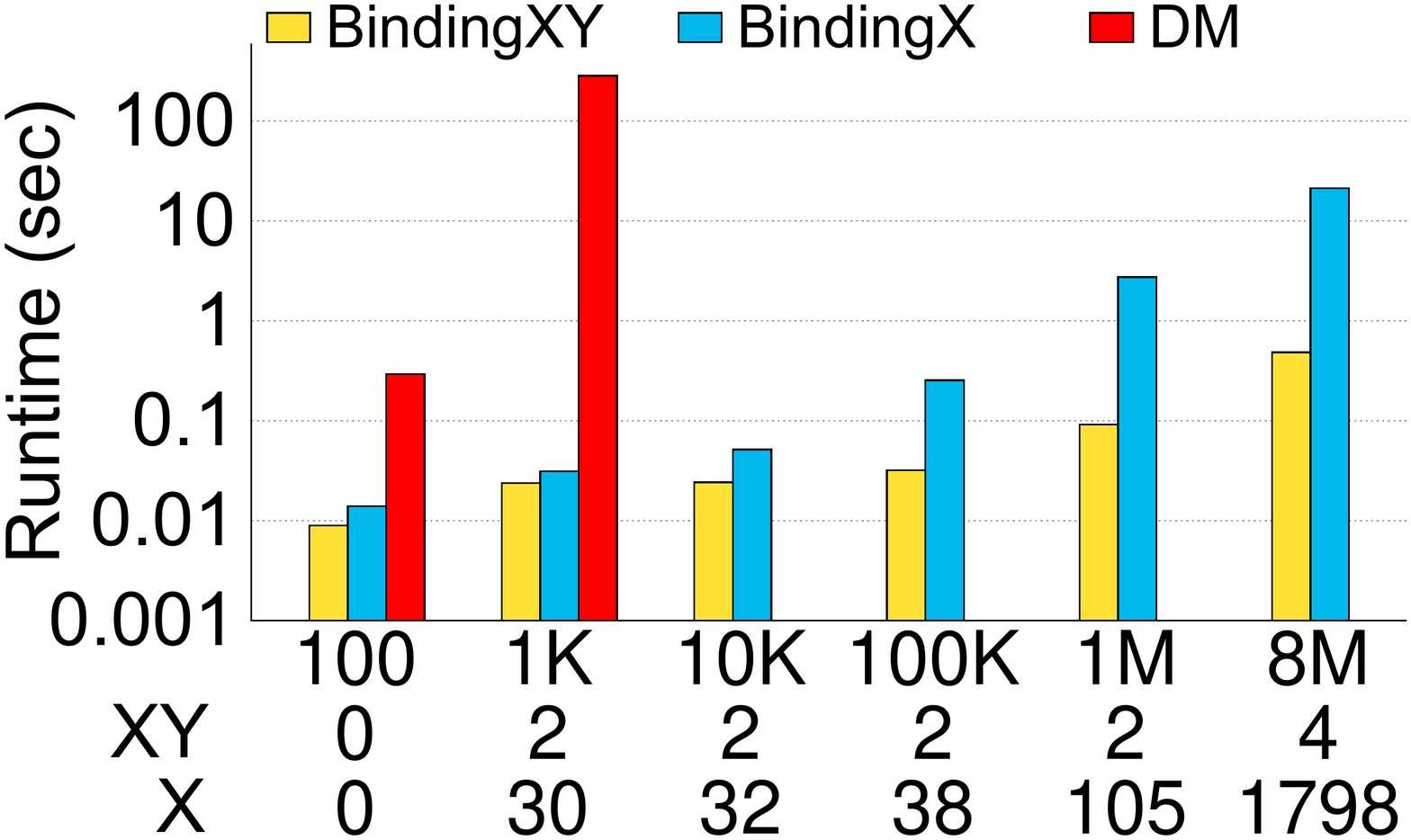}
\label{fig:only2hop-why}}
\subfloat[\small Runtime of \rel{XwithYnotZ}]
{
  \includegraphics[width=0.50\columnwidth,trim=0 60 0 0, clip]{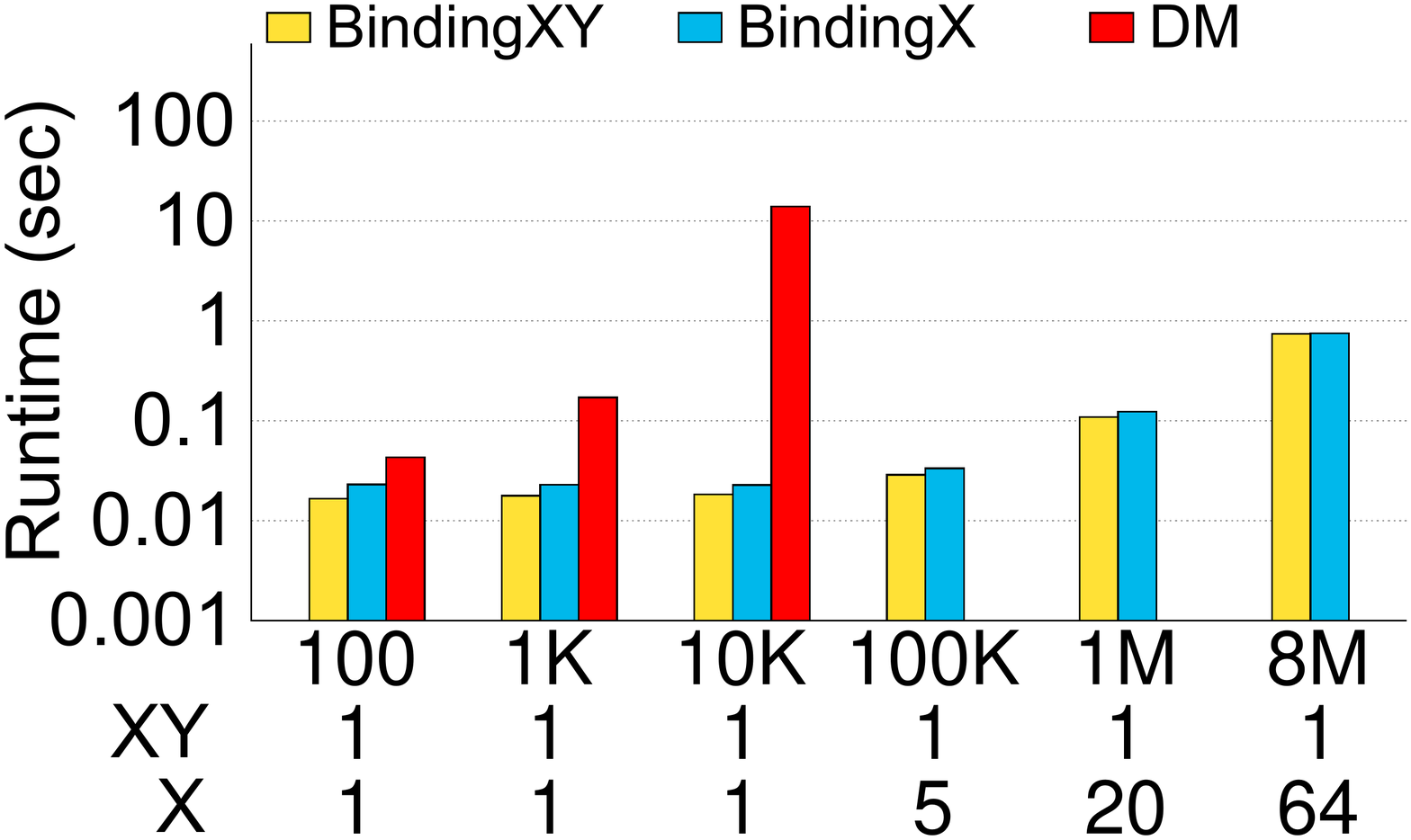}
\label{fig:xynotz-why}}
\\[2mm]
\begin{minipage}{0.98\linewidth}
\scriptsize\centering
\subfloat[\small Variable bindings for DBLP $\provQ$s]{
  \begin{minipage}{0.98\linewidth}
  \centering$\,$\\[-1mm]
  \begin{tabular}{|c|cc|}
    \thead{Query \textbackslash\, Binding}&\thead{X}&\thead{Y}\\
    (a) $\rel{only2hop}$ & Tore Risch & Rafi Ahmed \\
    (b) \rel{XwithYnotZ} & Arjan Durresi & Raj Jain \\
    \hline
  \end{tabular}
\end{minipage}
\label{fig:dblp-why-bind}}
\end{minipage}
\\[-6mm]
\subfloat[\small Runtime of \rel{ordPriority}]{
  \includegraphics[width=0.50\columnwidth,trim=0 60 0 0, clip]{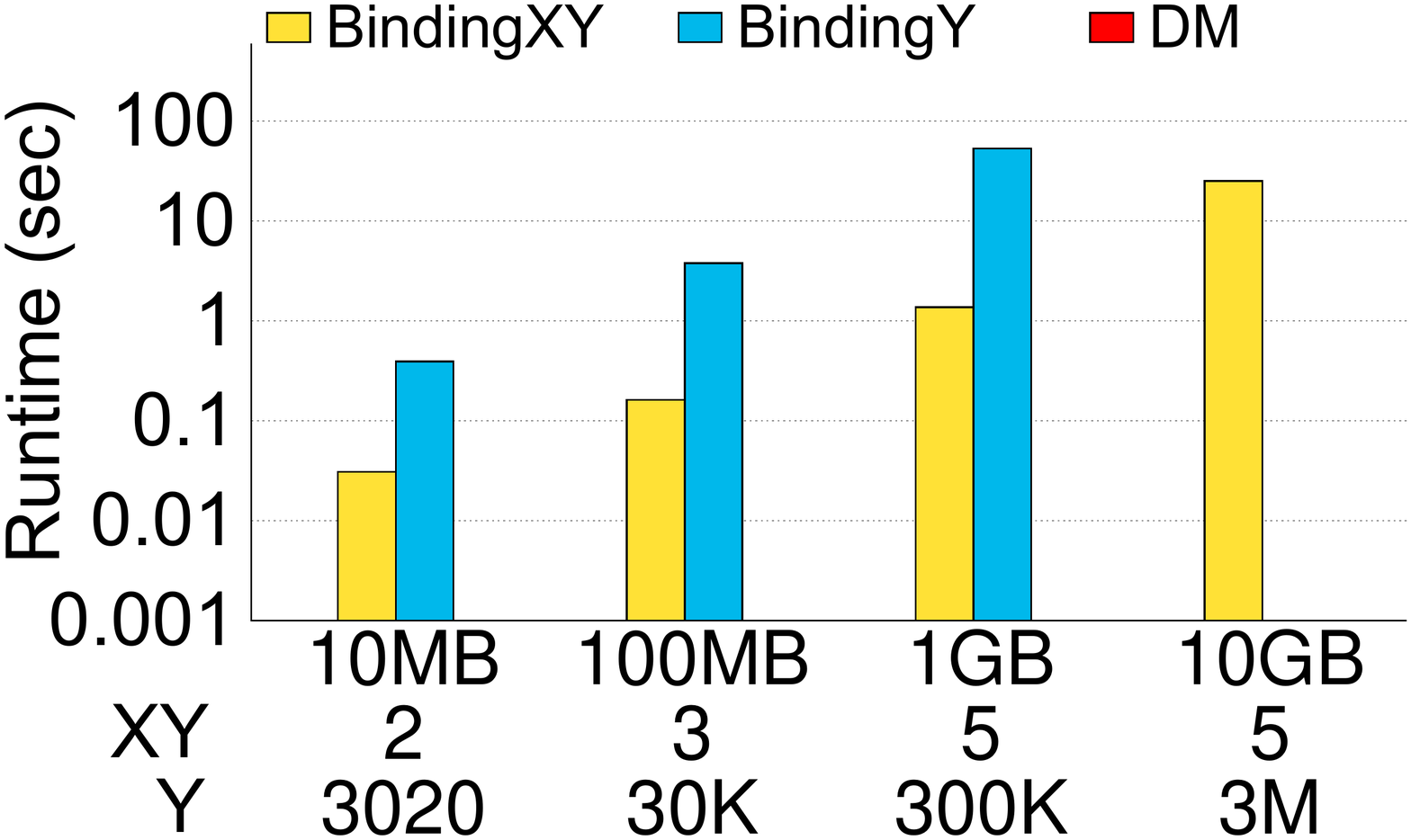}
\label{fig:urgent-why}}
\subfloat[\small Runtime of \rel{ordDisc}]{
  \includegraphics[width=0.50\columnwidth,trim=0 60 0 0, clip]{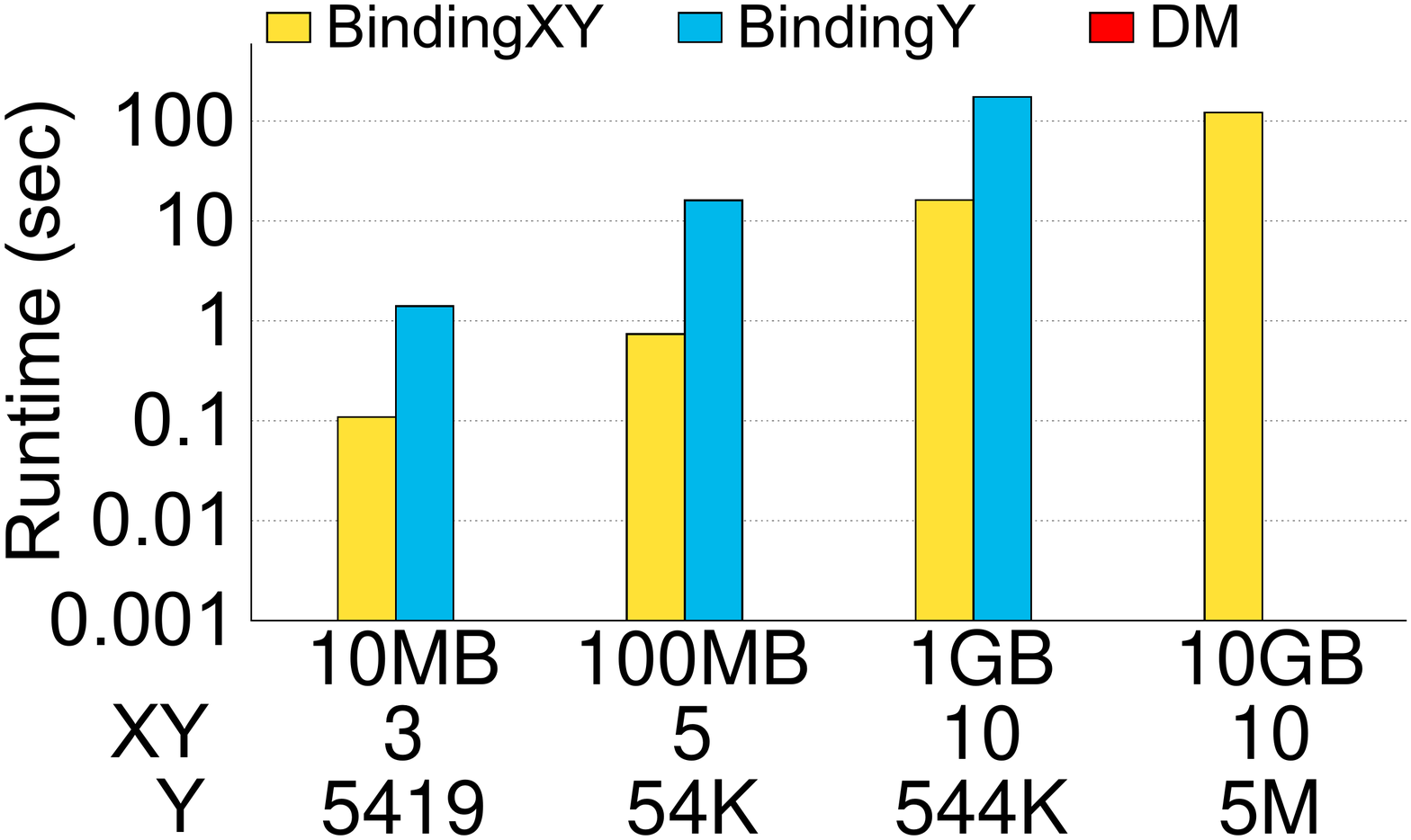}
\label{fig:zero-why}}
\end{minipage}
\\[3mm]
\begin{minipage}{0.98\linewidth}
\scriptsize\centering
\subfloat[\small Variable bindings for TPC-H $\provQ$s]{
  \begin{minipage}{0.98\linewidth}
 \centering$\,$\\[-2mm]
\begin{tabular}{|c|cc|}
\thead{Query \textbackslash\, Binding}&\thead{X}&\thead{Y}\\
(d) $\rel{ordPriority}$ & Customer16 & 1-URGENT \\
(e) $\rel{ordDisc}$ & Customer16 & 0 \\
\hline
\end{tabular}
\end{minipage}
\label{fig:tpch-why-bind}}
\end{minipage}
$\,$\\[-4mm]
\caption{Why questions: DBLP (top), TPC-H (bottom)}
\label{fig:perf-why}
\end{figure}

\mypartitle{Why Questions}
The runtime of generating explanations for  why questions
over the queries $r_1$, $r_2$, $r_4$, and $r_5$ (Fig.\,\ref{fig:experi-queries}) is shown in Fig.\,\ref{fig:perf-why}.
For the evaluation, we consider the effect of
 different binding patterns on performance.
Fig.\,\ref{fig:dblp-why-bind} and \ref{fig:tpch-why-bind} show which variables are bound by the provenance questions ($\provQ$s). Fig.\,\ref{fig:only2hop-why} and \ref{fig:xynotz-why} show the runtime for DBLP  queries $r_1$ and $r_2$, respectively.
We also provide the number of rule nodes in the explanation for each binding pattern below the X axis.
If only variable $X$ is bound (\texttt{BindingX}), then the queries determine authors that  occur together with the author we have bound to $X$
in the query result. For instance, the explanation
for $\rel{only2hop}$
with \texttt{BindingX}
explains why persons are indirect, but not direct, co-authors of ``Tore Risch''.
If both $X$ and $Y$ are bound (\texttt{BindingXY}), then the provenance for $r_1$ and $r_2$ is limited to a particular indirect and direct co-author, respectively.
The runtime for generating explanations grows roughly linear  in the dataset size
and outperforms \texttt{DM} even for small instances.
Furthermore, Fig.\,\ref{fig:urgent-why} and \ref{fig:zero-why}
(for $r_4$ and $r_5$, respectively) show that our approach can handle queries with many variables (attributes in TPC-H) where \texttt{DM} times out
even for the smallest dataset we have considered.
Binding one variable (\texttt{BindingY}) in queries $r_4$ and $r_5$
expresses a condition, e.g., $Y$ = `1-URGENT' in $r_4$ requires the order priority to be urgent.
If both variables are bound,
then the \provQ{} verifies the existence of orders for a certain customer
(e.g., why ``Customer16'' has at least one urgent order). Runtimes exhibit the same trend as for the DBLP queries.

\begin{figure}[t]
\begin{minipage}{0.98\linewidth}
$\,$\\[-9mm]
\centering
\subfloat[\small Runtime of \rel{only2hop}]{\includegraphics[width=0.50\columnwidth,trim=0 60 0 0, clip]{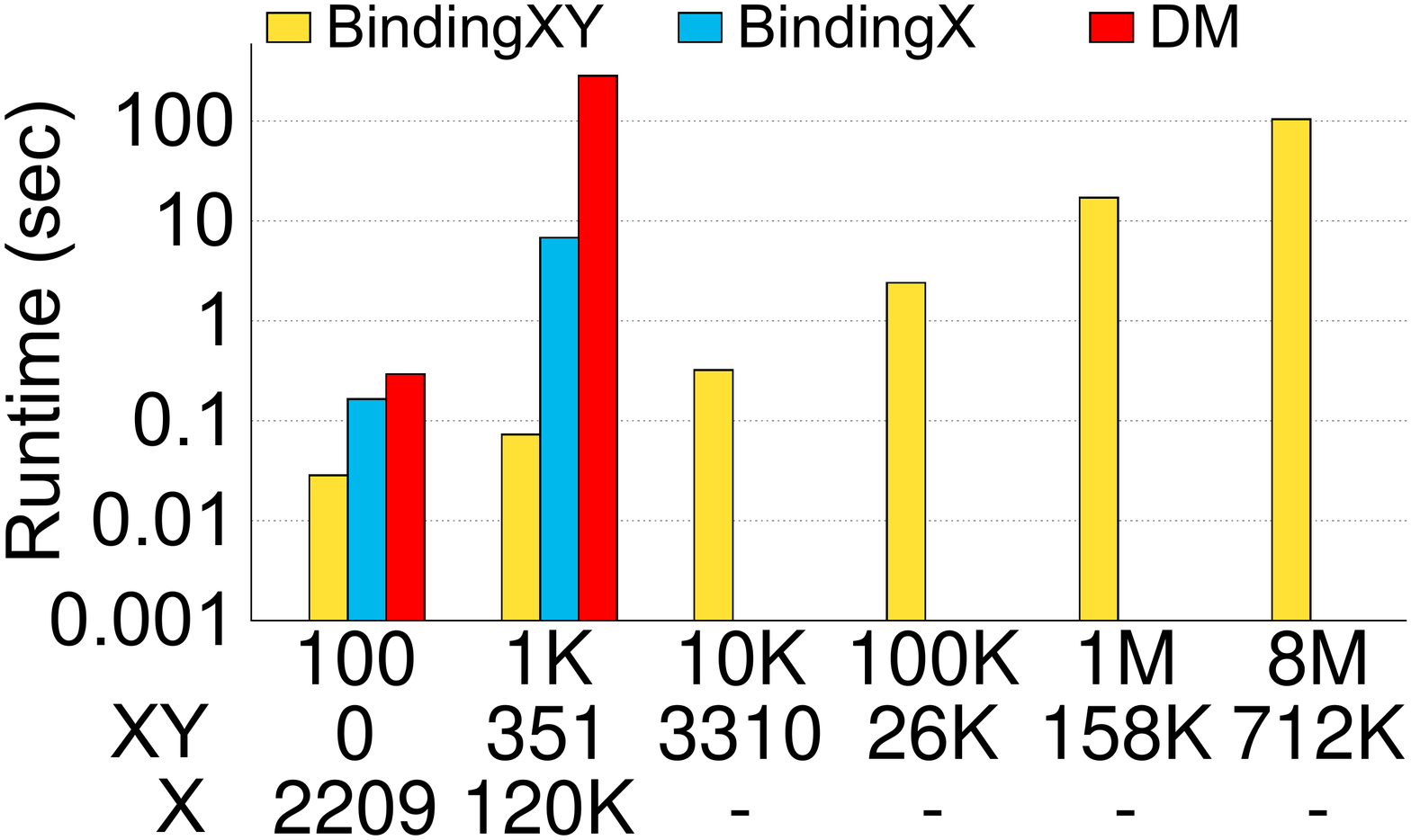}\label{fig:only2hop-whynot}}
\subfloat[\small Runtime of \rel{XwithYnotZ}]{\includegraphics[width=0.50\columnwidth,trim=0 60 0 0, clip]{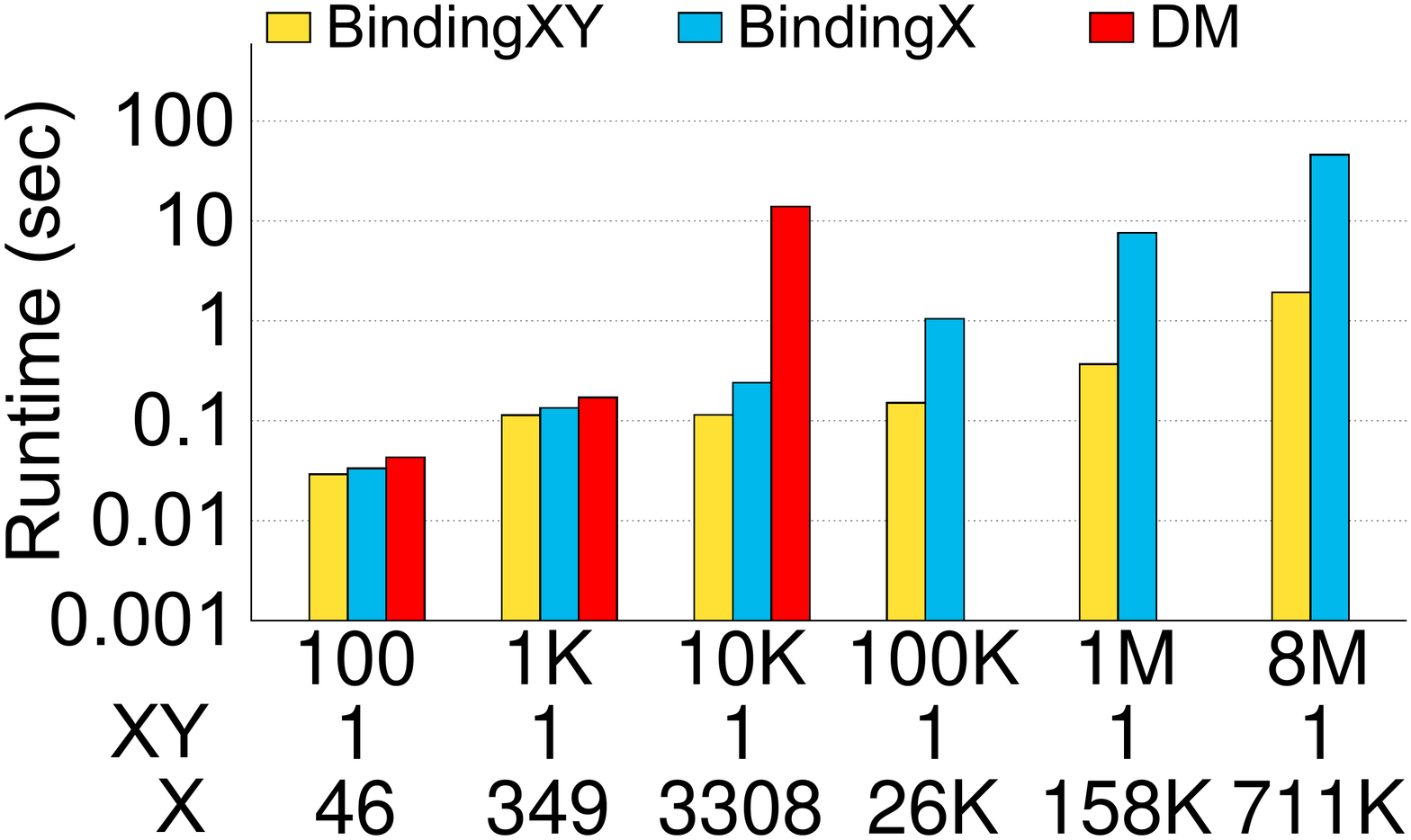}\label{fig:xynotz-whynot}}
\end{minipage}
\\[3mm]
\begin{minipage}{0.98\linewidth}
\scriptsize\centering$\,$\\[-2mm]
\subfloat[\small Variable bindings for DBLP $\provQ$s]{
\begin{tabular}{|c|cc|}
\thead{Query \textbackslash\, Binding}&\thead{X}&\thead{Y}\\
(a) \rel{only2hop} & Tore Risch & Svein Johannessen \\
(b) \rel{XwithYnotZ} & Tor Skeie & Joo-Ho Lee\\
\hline
\end{tabular}
\label{fig:dblp-whynot-bind}}
\end{minipage}
$\,$\\[-4mm]
\caption{Why-not questions over the DBLP dataset}
\label{fig:perf-whynot}
\end{figure}

\mypartitle{Why-not Provenance}
We use queries $r_1$ and $r_2$ from Fig.\,\ref{fig:experi-queries}
to evaluate the performance of computing explanations for failed derivations.
When binding all variables in the $\provQ$ (\texttt{BindingXY}) using the bindings from Fig.\,\ref{fig:dblp-whynot-bind}, these queries check if a particular set of authors do not appear together in the result.
For instance, for $\rel{only2hop}$ $(r_1)$, the query checks why ``Tore Risch'' is either
not an indirect co-author or
is a direct co-author of ``Svein Johannessen''.
\BGDel{If one variable is bound (\texttt{BindingX}), then  the why-not question explains for a given pairs of authors where one of the authors is bound to $X$, why the pair does not appear in the result.}
The results for queries $r_1$ and $r_2$ (DBLP) are shown in
Fig.\,\ref{fig:only2hop-whynot} and \ref{fig:xynotz-whynot}, respectively.
The number of tuples produced by the provenance computation (the number of rule nodes is shown below the X axis)
is quadratic in the database size resulting in a quadratic increase in runtime.
\texttt{DM} only finishes within the allocated time for very small datasets while our approach scales to larger instances.

\begin{figure}[t]
\begin{minipage}{0.98\linewidth}
$\,$\\[-9mm]
\centering
\subfloat[\small \rel{only3hop} (DBLP)]{\includegraphics[width=0.50\columnwidth,trim=0 60 0 0, clip]{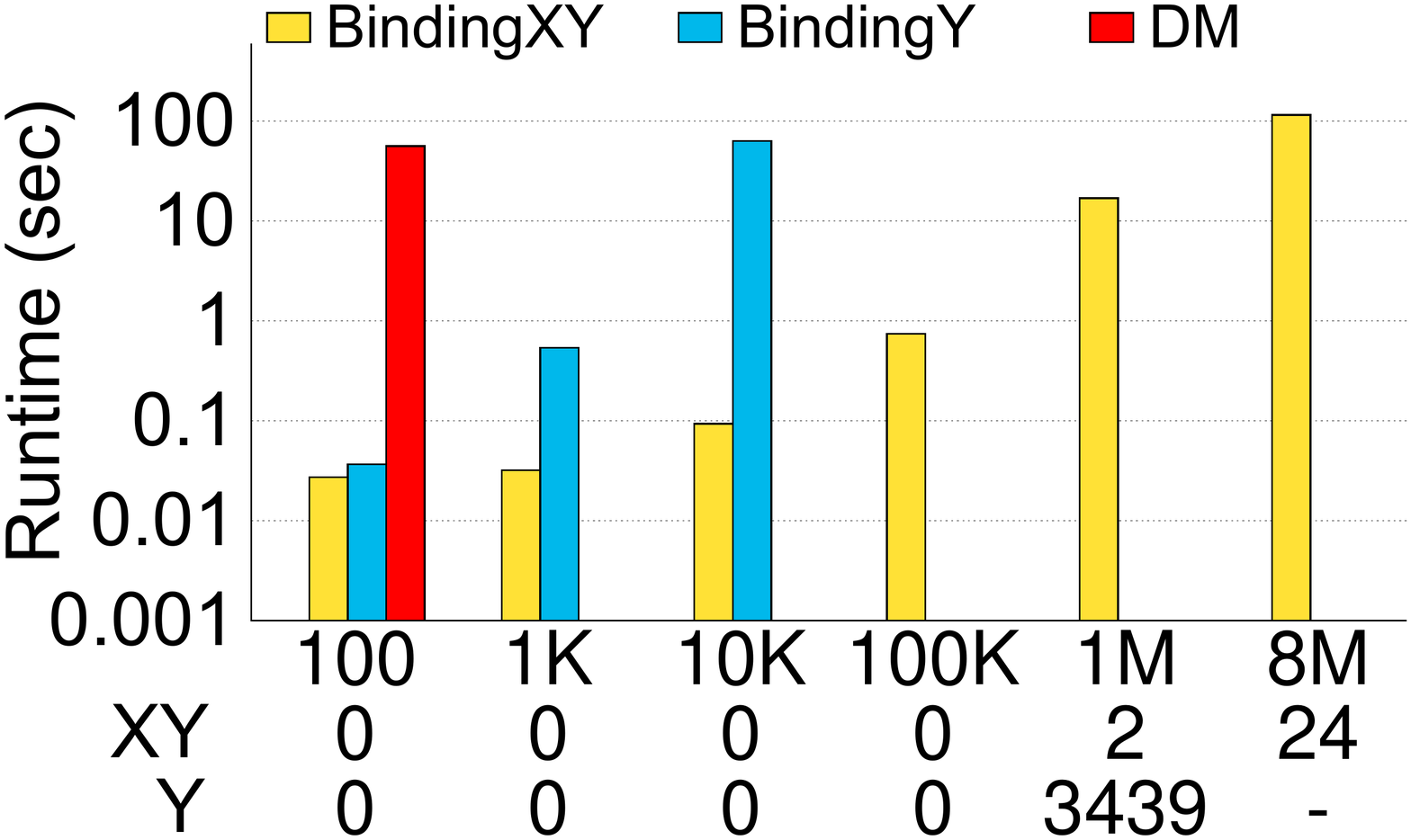}\label{fig:only3hop-why}}
\subfloat[\small \rel{partNotAsia} (TPC-H)]{\includegraphics[width=0.50\columnwidth,trim=0 60 0 0, clip]{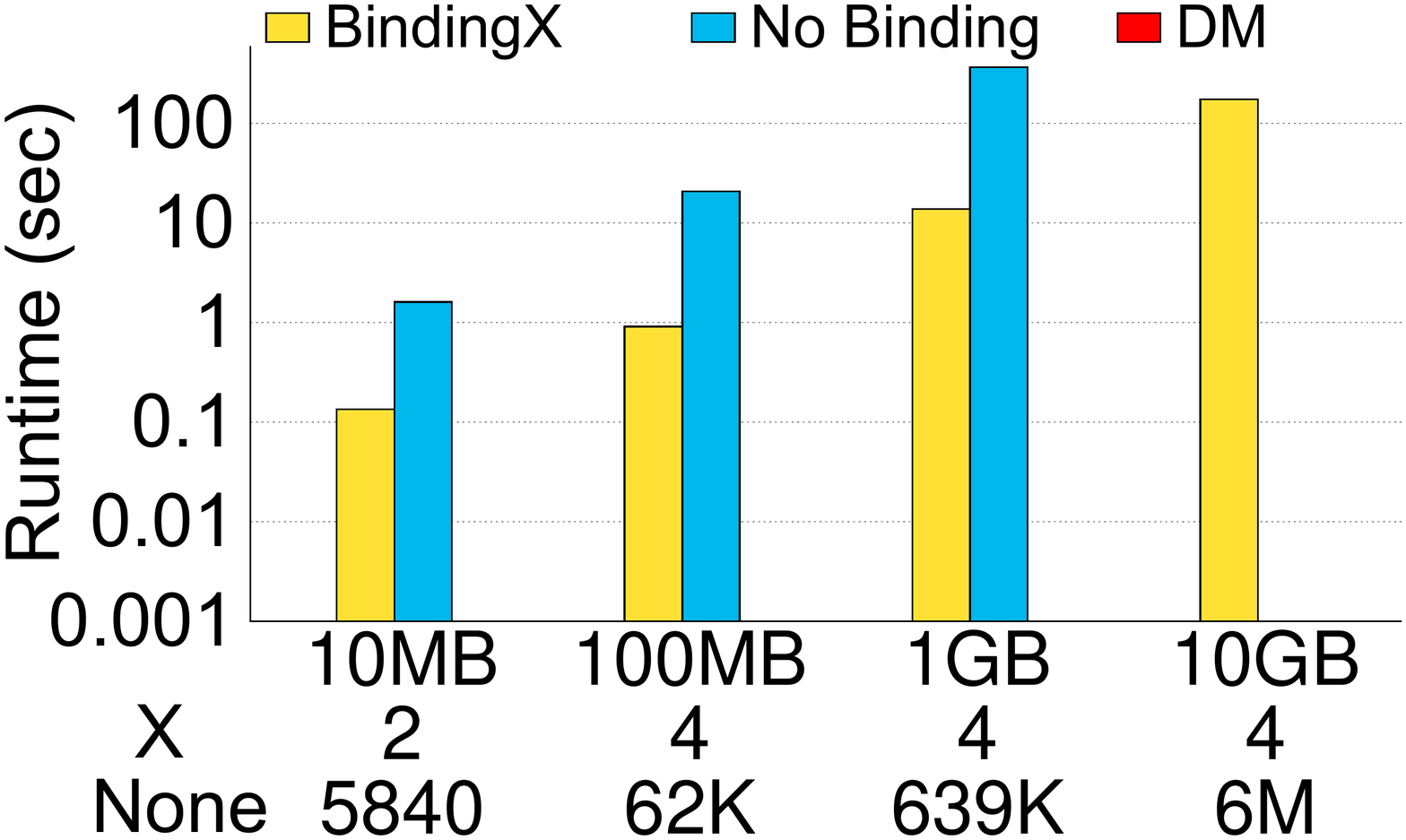}\label{fig:noasia-why}}
\end{minipage}
\\[3mm]
\begin{minipage}{0.98\linewidth}
\scriptsize\centering
\subfloat[\small Variable bindings for DBLP and TPC-H $\provQ$s]{
  \begin{minipage}{0.98\linewidth}
  \centering$\,$\\[-2mm]
\begin{tabular}{|c|cc|}
\thead{Query \textbackslash\, Binding}&\thead{X} &\thead{Y}\\
(a) $\rel{only3hop}$ & Alex Benton & Paul Erdoes \\
(b) $\rel{partNotAsia}$ & grcpi$^{1}$ & - \\
\hline
\end{tabular}\\[1mm]
${}^{1}$~grcpi = ghost royal chocolate peach ivory
\end{minipage}
\label{fig:neg-why-bind}}
\end{minipage}
$\,$\\[-4mm]
\caption{Why questions for queries with negation}
\label{fig:perf-uni}
\end{figure}
\mypartitle{Queries with Negation}
Recall that our approach also  handles queries with negation.
We choose
rules $r_3$ (multiple negated goals)
and $r_6$ (one negated goal) from Fig.\,\ref{fig:experi-queries} to evaluate the performance of answering why questions over such queries.
We use the bindings shown in Fig.\,\ref{fig:neg-why-bind}.
The results for $r_3$ and $r_6$ are shown in
Fig.\,\ref{fig:only3hop-why} and \ref{fig:noasia-why}, respectively. These results  demonstrate that
our approach efficiently computes explanations for such queries.
When increasing the database size, the runtimes of PQs for these queries exhibit the same trend
as observed for other why (why-not) questions and significantly outperform \texttt{DM}.
For instance, the performance of $\rel{partNotAsia}$ (Fig.\,\ref{fig:noasia-why}),
which contains many variables and negation
exhibits
the same trend as  queries that have no negation
(i.e., $r_4$ and $r_5$ in Fig.\,\ref{fig:urgent-why} and Fig.\,\ref{fig:zero-why}, respectively).

\begin{figure}[t]
\begin{minipage}{0.98\linewidth}
$\,$\\[-9mm]
\centering
\subfloat[\small \rel{Q7} (employee)]{\includegraphics[width=0.5\columnwidth,trim=0 80 0 0, clip]{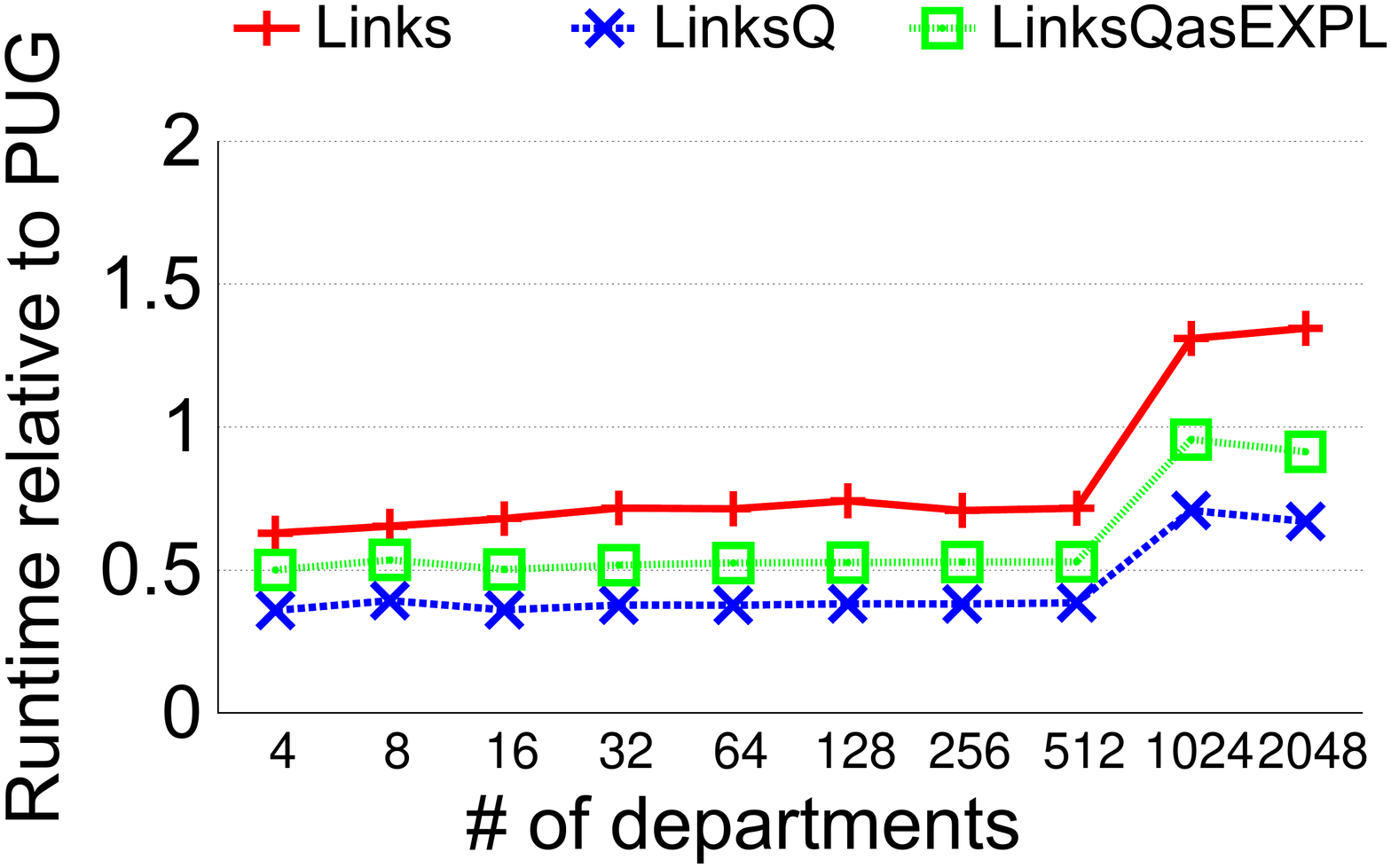} \label{fig:q7-lin}}
\subfloat[\small \rel{QF3} (employee)]{\includegraphics[width=0.5\columnwidth,trim=0 80 0 0, clip]{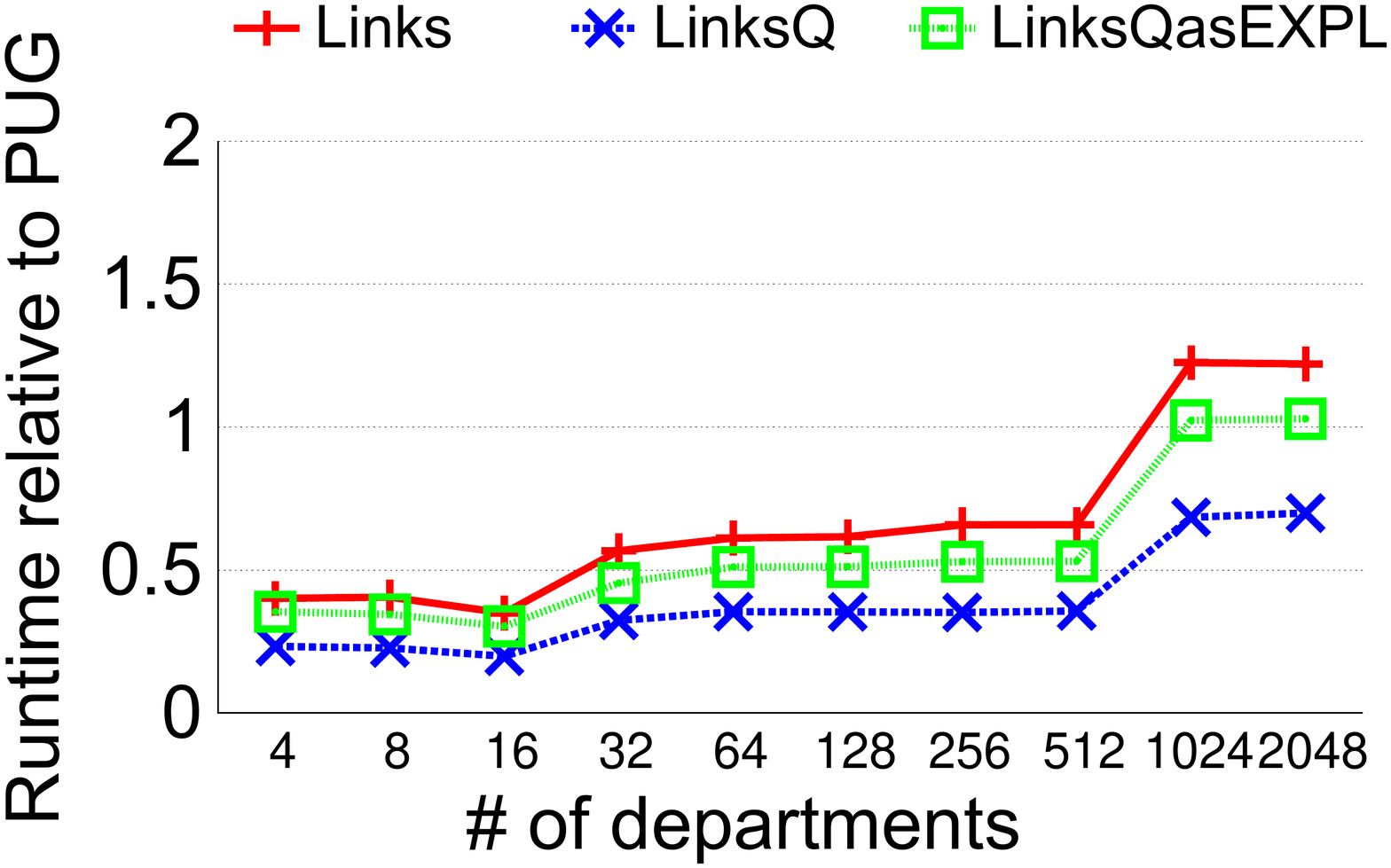} \label{fig:qf3-lin}}
\end{minipage}\\
$\,$\\[-7mm]
\caption{Comparing $\WhichProv$ in PUG with Links }
\label{fig:comp-pug-links}
\end{figure}

\mypara{Comparison with Links}
\label{sec:experiments-links}
In this experiment, we compare the runtime of computing $\explainq_{\WhichProv}$ (e.g., Fig.\,\ref{fig:lin-model}) with computation of Lineage in $\LinksL$ from~\cite{FS17}.
We show relative runtimes where PUG is normalized to $1$.
For this particular evaluation, we use Postgres as a backend since it is supported by both PUG and Links. Note that $\explainq_{\WhichProv}$ contains a full description of each tuple unlike $\LinksL$ which returns tuple identifiers (OIDs in Postgres). To get a nuanced understanding of the system's performance, we show three runtimes for Links: 1) \texttt{Links} is the actual implementation in Links which computes Lineage (only OIDs) and where the runtime includes the construction of in-memory Links types from the provenance fetched from Postgres; 2) \texttt{LinksQ} is the runtime of the queries that Links uses to capture Lineage; and 3) \texttt{LinksQasEXPL} which joins the output of \texttt{LinksQ} with the base tables (i.e., as informative as $\explainq_{\WhichProv}$). \iftechreport{
For \texttt{LinksQ}, the capture the SQL queries that \texttt{Links}  runs to  capture provenance. We then compare the runtime of these  queries to generation of $\explainq_{\WhichProv}$ in PUG.
The queries we use for \texttt{LinksQasEXPL} are generated as follows: we join the tuple id and relation
name pairs produced by \texttt{LinksQ}  with the corresponding relations in the
database to return full tuples as in PUG's
$\explainq_{\WhichProv}$ provenance type. Recall that PUG encodes $\explainq_{\WhichProv}$ is
as the edge relation of a DAG of tuple nodes. Therefore, when generating the queries for
\texttt{LinksQasEXPL}, we add code to generate this edge relation from the returned tuples using
string concatenation. For example, one tuple from the edge relation of the graph shown in
Fig.\,\ref{fig:lin-model}  is
$(\rel{Q_{3hop}}(s,s),\rel{Train}(s,s))$. $\LinksL$
and \texttt{LinksQ} would represent this part of the provenance as the tuple id of $(s,s)$ paired with the name of this relation
($\rel{Train}$). For instance, assuming the id of this tuple is $123$, we would get $(123,\rel{Train})$. Based on this information,
\texttt{LinksQasEXPL} computes the edge $(\rel{Q_{3hop}}(s,s),\rel{Train}(s,s))$ by
joining with relation $\rel{Train}$ using the tuple-id $123$ and by applying string concatenation. For example, using \lstinline!||! to denote string concatenation,
$\rel{Train}(s,s)$ is generated by evaluating the expression

\begin{center}
$\rel{Train}$ \lstinline!|| '(' || 's' || ',' || 's' || ')'!
\end{center}
}
\begin{figure}[t]
\begin{minipage}{0.98\linewidth}
$\,$\\[-9mm]
\centering
\subfloat[\small Runtime of \rel{suppCust}]{\includegraphics[width=0.5\columnwidth,trim=0 50 0 0, clip]{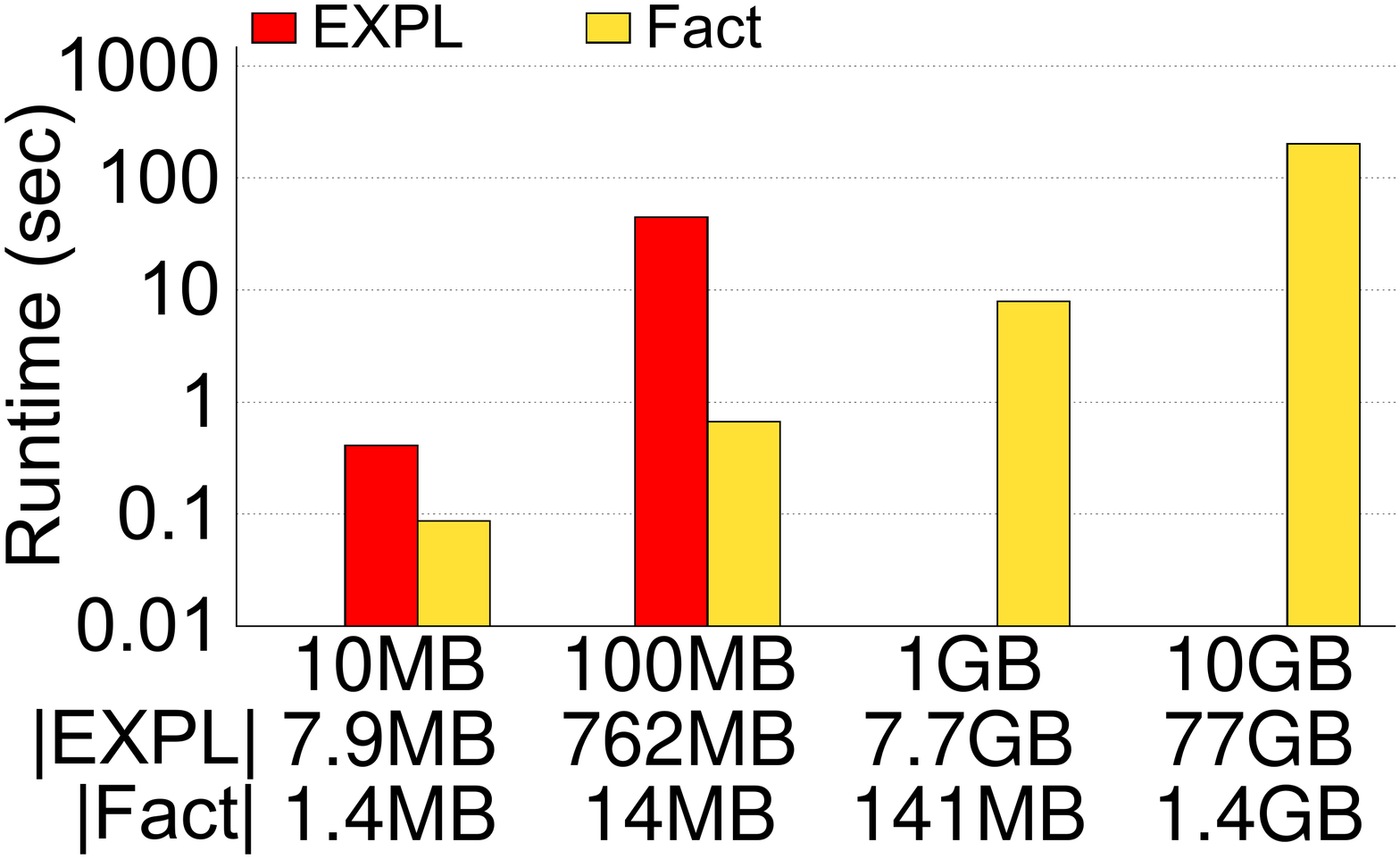} \label{fig:tpch-suppCust-fact}}
\subfloat[\small Runtime of \rel{ordDisc}]{\includegraphics[width=0.5\columnwidth,trim=0 50 0 0, clip]{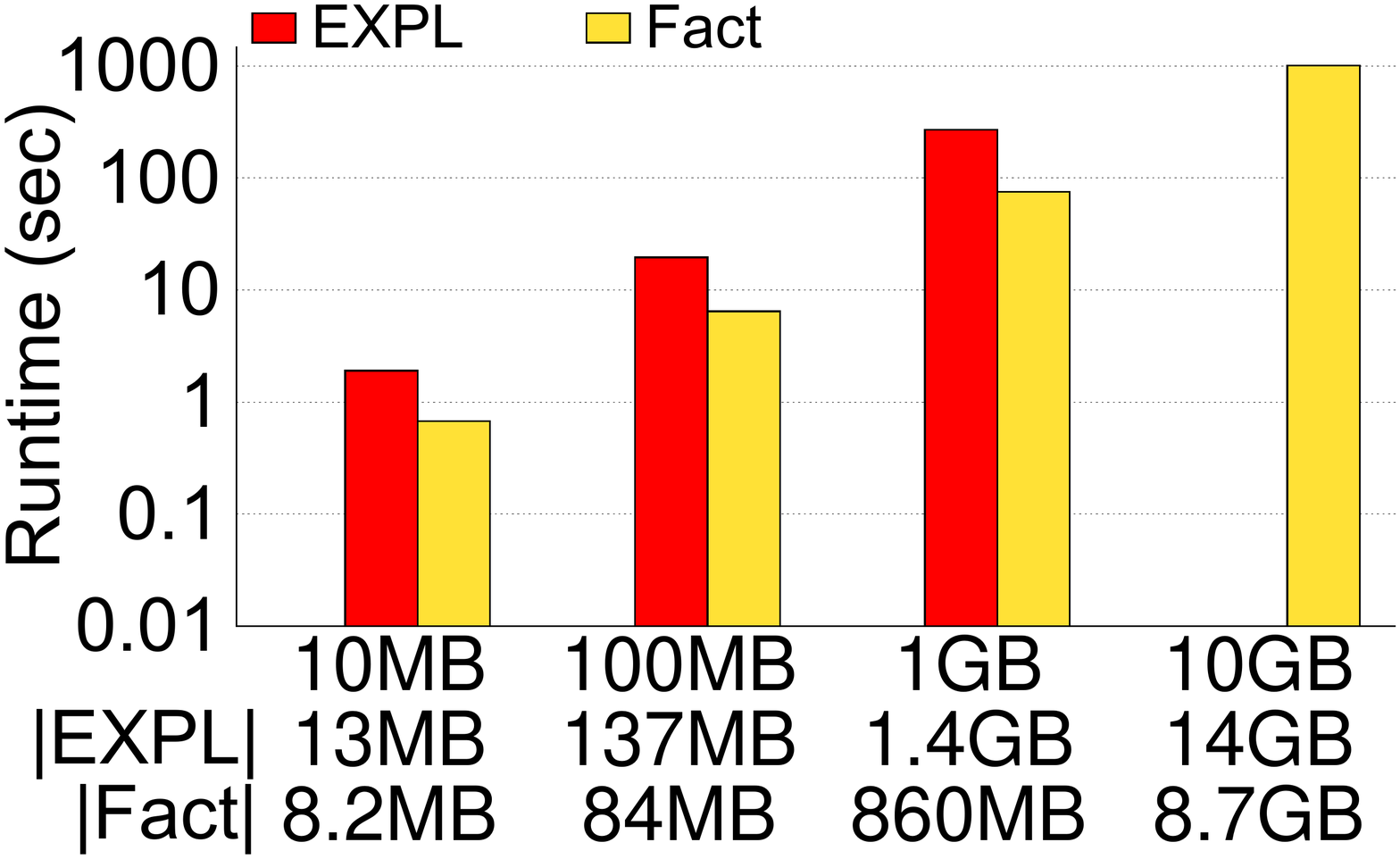} \label{fig:tpch-zero-fact}}
\end{minipage}
$\,$\\[-4mm]
\caption{Explanations vs. factorized explanations}
\label{fig:perf-fac}
\end{figure}
We choose  two queries from~\cite{FS17}.  The query Q7 applies a range condition to the result of  a two-way join.  QF3 is a self-join on equality with an additional inequality condition (see~\cite{FS17} for more details). The queries are expressed over two tables $\rel{dept}$ and $\rel{emp}$. The number of departments is varied from $4$ to $2048$ (by powers of $2$ to replicate the setting from~\cite{FS17}),
and each department has $100$ employees on average. The relation $\rel{dept}$ consists of one attribute (department name), and $\rel{emp}$ has three attribute (department name, employee name, and salary).
QF3 can be written in Datalog as: 

\noindent\begin{minipage}{1\linewidth}\vspace{-2mm}
\begin{minipage}{1\linewidth}
\centering \small
\begin{align*}
\rel{QF3}(N1,N2) \dlImp \rel{emp}(\_,D,N1,S), \rel{emp}(\_,D,N2,S), N1 \neq N2
\end{align*}
\end{minipage}
\end{minipage}\\[-2mm]

The runtimes of queries Q7 and QF3 are shown in Fig.\,\ref{fig:q7-lin} and \ref{fig:qf3-lin}, respectively.
Links performs better on smaller instances. The gap between Links and PUG shrinks with increasing dataset size. PUG outperforms \texttt{Links} and \texttt{LinksQasEXPL} on larger datasets.

\mypara{Factorized Explanations}
\label{sec:experiments-fact}
We now compare the performance of generating provenance for a query  (\texttt{EXPL}) and a factorized representation of provenance (\texttt{Fact}) by rewriting the input query (Sec.\,\ref{sec:factorize}). Factorization techniques perform best for many-to-many joins
(e.g., the query $r_7$ in Fig.\,\ref{fig:experi-queries}). The rewritten version of $\rel{suppCust}$ $(r_7)$ producing factorized provenance is shown below.

\begin{minipage}{1\linewidth}\vspace{-2mm}
\centering
\begin{minipage}{0.8\linewidth}
\centering\small
\begin{align*}
  &r_8:  \rel{suppCust}(N) \dlImp \rel{supp}(N), \rel{cust}(N)\\
  &r_{8'}: \rel{supp}(N) \dlImp \rel{SUPPLIER}(A,B,C,N,D,E,F)\\
  &r_{8''}: \rel{cust}(N) \dlImp \rel{CUSTOMER}(G,H,I,N,J,K,L,M)\\
\end{align*}
\end{minipage}
\end{minipage}\\[-5mm]

For this experiments, we use  a $15$ minute time-out. The runtimes for $r_7$ (yellow bars) and $r_8$ (red bars) are shown in Fig.\,\ref{fig:tpch-suppCust-fact}.
We show the total result size in bytes below the $X$ axis.
The runtime of \texttt{Fact} grows roughly linear
unlike \texttt{EXPL} whose growth is quadratic in dataset size. We also evaluate query $r_5$ which includes one-to-many joins to see how \texttt{Fact} performs for a query  (Fig.\,\ref{fig:tpch-zero-fact}) where factorization only reduces size by a constant factor.
This is confirmed by the measurements: the performance of \texttt{Fact} for $r_5$ is $\sim30\%$ that of \texttt{EXPL} independent of dataset size.

\section{Conclusions}
\label{sec:concl}

We present a provenance model and unified framework for explaining answers and non-answers over first-order queries expressed in Datalog. Our efficient middleware implementation generates a Datalog program that 
computes the explanation 
for a provenance question and compiles this program 
into SQL. We prove that our model is expressive enough to encode a wide range of provenance models from the literature and extend our approach to produce concise, factorized representations of provenance.
In future work, we will investigate summarization of provenance (we did present a proof-of-concept in~\cite{LN17}) to deal with the large size of explanations for missing answers. We plan to also support query-based explanations~\cite{BH14a,BH14,CJ09,TC10} and more expressive query languages (e.g., aggregation).

	 \end{document}